\documentclass[acmsmall, screen, nonacm]{acmart}

\AtBeginDocument{%
}

\setcopyright{none}          
\acmYear{2027}

\PassOptionsToPackage{disable}{todonotes}

\PassOptionsToPackage{table}{xcolor}

\usepackage{tikz}
\usetikzlibrary{arrows.meta}
\usetikzlibrary{fit,calc}
\usepackage{wrapfig}
\pgfdeclarelayer{background}
\pgfsetlayers{background,main}

\usepackage{pgfplots}
\pgfplotsset{compat=1.18}

\usepackage{import}

\usepackage[all]{nowidow}
\usepackage[lastparline]{impnattypo}

\usepackage{csquotes}

\usepackage[labelfont=bf,textfont=it]{caption}

\usepackage{enumitem}
\setlist{labelindent=\parindent, leftmargin=2\parindent}

\usepackage{float}

\usepackage{multicol}
\usepackage{parcolumns}

\usepackage{centernot}

\usepackage[thicklines]{cancel}

\usepackage{tikz}

\usetikzlibrary{cd, arrows}
\usetikzlibrary{decorations.pathreplacing}
\usetikzlibrary{positioning}

\usepackage{graphicx}
\graphicspath{{./}{../}{../../}}

\usepackage{fonttable}
\usepackage{expl3}
\usepackage{xparse, xpatch}

\usepackage{todonotes}
\reversemarginpar

\usepackage{tabularx, multirow}
  \usepackage{stackrel, stackengine}

\IfFileExists{colortbl.sty}{%
  \usepackage[table]{xcolor}%
}{%
  \usepackage{xcolor}%
}
\usepackage{hhline}

\usepackage{ebproof}

  \usepackage{amsfonts, mathtools, stmaryrd}
  \let\mathscr\mathcal
\usepackage{amsopn}

\usepackage{quiver}

\usepackage{hyperref}

\usepackage{subcaption}

\usepackage[export]{adjustbox}

\usepackage{xspace}

\usepackage[frozencache]{minted2}

\newcommand{\vehicle}{\textbf{Vehicle}\xspace}
\providecommand{\lstinline}[1]{\texttt{#1}} 

\usepackage{amsthm}
\usepackage{thmtools}

\usepackage[english]{babel}%

\usepackage[capitalise,nameinlink,capitalize,noabbrev]{cleveref}

\usepackage{stmaryrd}

\theoremstyle{plain}
\newtheorem{theorem}{Theorem}[section]
\newtheorem{proposition}[theorem]{Proposition}
\newtheorem{corollary}{Corollary}[theorem]
\newtheorem{lemma}[theorem]{Lemma}

\theoremstyle{definition}
\newtheorem{remark}[theorem]{Remark}
\newtheorem{definition}[theorem]{Definition}
\newtheorem{notation}[theorem]{Notation}

\newtheorem{example}[theorem]{Example}

\crefname{diagram}{Diagram}{Diagrams}
\crefname{rule}{Rule}{Rules}
\crefname{proof}{Derivation}{Derivations}
\crefname{axiom}{Axiom}{Axioms}

\newenvironment{eqalign}{\begin{equation}\begin{aligned}}{\end{aligned}\end{equation}}
\newenvironment{eqalign*}{\begin{equation*}\begin{aligned}}{\end{aligned}\end{equation*}}

\tikzset{
  relation/.style={
    draw=none,
    every to/.append style={
      edge node={node [sloped, allow upside down, auto=false]{$#1$}}}
  }
}

\tikzcdset{
  mono/.code={
    \pgfsetarrows{tikzcd to reversed-tikzcd to}
  }
}
\tikzcdset{
  epi/.code={
    \pgfsetarrowsend{tikzcd double to}
  }
}
\tikzcdset{
  into/.code={
    \pgfsetarrows{tikzcd right hook-tikzcd to}
  }
}
\tikzcdset{
  twocell/.style={Rightarrow, shorten >= 3ex, shorten <= 3ex}
}

\newcommand{\Overset}[2]{%
  \mathop{#2}\limits^{\vbox to -.1ex{%
  \kern -1.8ex\hbox{$#1$}\vss}}%
}
\newcommand{\Underset}[2]{%
  \mathop{#2}\limits_{\vbox to .1ex{%
  \kern -.6ex\hbox{$#1$}\vss}}%
}

\mathchardef\dash="2D

\renewcommand{\exp}{\operatorname{exp}}

\DeclareMathOperator{\argmin}{\mathrm{argmin}}

\newcommand{\longto}{\longrightarrow}

\newcommand{\twoto}{\Rightarrow}

\newcommand{\R}{\mathbb{R}}

\newcommand{\biginf}{\bigwedge}

\newcommand{\conv}[1][]{\underset{{#1}}{\longrightarrow}}

\newcommand{\cat}[1]{{\mathcal{#1}}}
\newcommand{\thecat}[1]{{\mathbf{#1}}}

\newcommand{\id}{\mathrm{id}}

\newcommand{\iso}[1][]{\overset{#1}{\cong}}

\newcommand{\adj}{\dashv}

\newcommand{\Set}{\thecat{Set}}

\newcommand{\op}{\mathsf{op}}

\newcommand{\V}{{\cat{V}}}

\newcommand{\parr}{\mathbin{\mathchoice
  {\rotatebox[origin=c]{180}{$\displaystyle\&$}}
  {\rotatebox[origin=c]{180}{$\textstyle\&$}}
  {\rotatebox[origin=c]{180}{$\scriptstyle\&$}}
  {\rotatebox[origin=c]{180}{$\scriptscriptstyle\&$}}
}}

\newcommand{\psum}[2][p]{\mathchoice{{\bigoplus_{#2}}^{#1}}{\bigoplus_{#2}^{#1}}{\bigoplus_{#2}^{#1}}{\bigoplus_{#2}^{#1}}}
\newcommand{\hsum}[1]{\psum[*]{#1}}

\newcommand{\Reals}{[-\infty, +\infty]}

\newcommand{\PosReals}{[0,\infty]}

\newcommand{\MulReals}{\PosReals_\tensor}

\newcommand{\pMulReals}[1][p]{\PosReals_{\padd[#1],\tensor}}
\newcommand{\pAddReals}[1][p]{\Reals_{\pAor[#1],\add}}

\newcommand{\mulimp}{\multimap}

\newcommand{\tensor}{\otimes}
\newcommand{\cotensor}{\mathbin{\tensor^*}}
\newcommand{\add}{\mathbin{\oplus}}
\newcommand{\padd}[1][p]{\add^{#1}}
\newcommand{\pcoadd}[1][p]{\add^{-#1}}
\newcommand{\coadd}{\add^*}
\newcommand{\pAor}[1][p]{\mathbin{\cup^{#1}}}
\newcommand{\pAand}[1][p]{\mathbin{\cap^{#1}}}

\newcommand{\One}{\mathbf{1}}
\newcommand{\true}{T}
\newcommand{\false}{F}

\newcommand{\entails}{\vdash}

\newcommand{\Prop}{{\bf Prop}}

\newcommand{\sem}[1]{\llbracket{#1}\rrbracket}

\newcommand{\red}[1]{\textcolor{red}{#1}}
\newcommand{\redbin}[1]{\mathbin{\red{#1}}}
\newcommand{\redbbin}[1]{\quad\red{#1}\quad}

\newcommand{\leftrule}{_{\text{L}}}
\newcommand{\rightrule}{_{\text{R}}}

\newcommand{\QLL}{\ensuremath{\hyperref[fig:pqll]{\text{QLL}}}\xspace}
\newcommand{\pQLL}[1][p]{\ensuremath{\hyperref[fig:pqll]{#1\text{QLL}}}\xspace}
\newcommand{\Free}[2]{
\ensuremath{%
  \mathchoice
    {#1^{\cancel{#2}{\scriptstyle}}}
    {#1^{\cancel{#2}{\scriptstyle}}}
    {#1^{\cancel{#2}{\scriptscriptstyle}}}
    {#1^{\cancel{#2}{\scriptscriptstyle}}}
}\xspace}
\newcommand{\CutFreepQLL}[1][p]{\Free{\pQLL[#1]}{\CutRule}}
\newcommand{\EFQFreepQLL}[1][p]{\Free{\pQLL[#1]}{\ExFalsoRule}}
\newcommand{\pQLLStar}[1][p]{\ensuremath{\hyperref[def:pqllstar]{#1\text{QLL}^*}}\xspace}

\newcommand{\MALL}{\ensuremath{\text{MALL}}\xspace}
\newcommand{\ISOMIXMALL}{\ensuremath{\hyperref[fig:isomixmall]{\text{MALL}}}\xspace}
\newcommand{\hISOMIXMALL}{\ensuremath{\hyperref[fig:isomixmall-h]{\text{HMALL}}}\xspace}

\newcommand{\eleq}{\sqsubseteq}

\newcommand{\plor}[1][p]{\lor^{#1}}
\newcommand{\pland}[1][p]{\land^{#1}}

\newcommand{\grammareq}{\ \Coloneqq \ }
\newcommand{\grammarsep}{\ | \  }
\newcommand{\Atoms}{V}

\newcommand{\hyper}[1]{\mathcal{#1}}
\newcommand{\rules}[1]{\mathscr{#1}}

\newcommand{\validity}[1]{{\left|#1\right|}}

\newcommand{\monop}{\bullet}
\newcommand{\dumonop}{\circ}
\newcommand{\dual}{\perp}
\newcommand{\Softales}{\thecat{Sftl}}

\newcommand{\model}[1][M]{\mathcal{#1}}

\newcommand{\SynSoftale}[2][p]{\mathbf{LT}^{#1}\!\left[#2\right]}
\newcommand{\Theory}{\mathbb{T}}

\DeclareMathOperator{\expected}{\mathbb{E}}

\newcommand{\atom}[1]{\ulcorner #1 \urcorner}

\usepackage{algorithmic}
\usepackage[gen]{eurosym}

\renewcommand{\EUR}{\ensuremath{\text{\euro}}}

\ebproofset{separation=.75em}

\renewcommand{\plor}[1][]{\lor^{#1}}
\renewcommand{\pland}[1][]{\land^{#1}}

\newcommand{\TensorLeftRule}{\ensuremath{\hyperref[subfig:mult-frag]{\tensor\leftrule}}\xspace}
\newcommand{\TensorRightRule}{\ensuremath{\hyperref[subfig:mult-frag]{\tensor\rightrule}}\xspace}
\newcommand{\ParLeftRule}{\ensuremath{\hyperref[subfig:mult-frag]{\parr\leftrule}}\xspace}
\newcommand{\ParRightRule}{\ensuremath{\hyperref[subfig:mult-frag]{\parr\rightrule}}\xspace}
\newcommand{\OneLeftRule}{\ensuremath{\hyperref[subfig:mult-frag]{\One\leftrule}}\xspace}
\newcommand{\OneRightRule}{\ensuremath{\hyperref[subfig:mult-frag]{\One\rightrule}}\xspace}
\newcommand{\DualLeftRule}{\ensuremath{\hyperref[subfig:mult-frag]{{(-)}^\dual\leftrule}}\xspace}
\newcommand{\DualRightRule}{\ensuremath{\hyperref[subfig:mult-frag]{{(-)}^\dual\rightrule}}\xspace}

\newcommand{\CutRule}{\ensuremath{\hyperref[subfig:struct-frag]{\text{CUT}}}\xspace}
\newcommand{\AxRule}{\ensuremath{\hyperref[subfig:struct-frag]{\text{AX}}}\xspace}
\newcommand{\EmptyRule}{\ensuremath{\hyperref[subfig:struct-frag]{\text{EMP}}}\xspace}
\newcommand{\MixStarRule}{\ensuremath{\text{\hyperref[rule:mix-star-and-split]{MIX$^*$}}}\xspace}
\newcommand{\MixRule}{\ensuremath{\hyperref[subfig:struct-frag]{\text{MIX}}}\xspace}

\newcommand{\pOrLeftRule}[1][]{\ensuremath{\hyperref[subfig:add-frag]{\plor[#1]\leftrule}}\xspace}
\newcommand{\pOrRightRule}[1][]{\ensuremath{\hyperref[subfig:add-frag]{\plor[#1]\rightrule}}\xspace}
\newcommand{\pAndLeftRule}[1][]{\ensuremath{\hyperref[subfig:add-frag]{\pland[#1]\leftrule}}\xspace}
\newcommand{\pAndRightRule}[1][]{\ensuremath{\hyperref[subfig:add-frag]{\pland[#1]\rightrule}}\xspace}
\newcommand{\BotLeftRule}{\ensuremath{\hyperref[subfig:add-frag]{\bot\leftrule}}\xspace}
\newcommand{\ExFalsoRule}{\ensuremath{\text{\hyperref[subfig:add-frag]{EFQ}}}\xspace}
\newcommand{\TopRightRule}{\ensuremath{\hyperref[subfig:add-frag]{\top\rightrule}}\xspace}

\newcommand{\ExtWeakRule}{\ensuremath{\text{\hyperref[rule:external-structural]{EW}}}\xspace}
\newcommand{\ExtCoweakRule}{\ensuremath{\text{\hyperref[rule:external-structural]{ECOW}}}\xspace}
\newcommand{\ExtContRule}{\ensuremath{\text{\hyperref[rule:external-structural]{EC}}}\xspace}

\newcommand{\ExtExchRule}[1][]{\ensuremath{\text{\hyperref[rule:external-structural]{EE}}_{#1}}\xspace}

\newcommand{\StructuralSchemaRule}{\ref{rule:structural-schema}\xspace}
\newcommand{\AtomRightRule}[1]{\ensuremath{\hyperref[rule:grounding-schema]{#1\rightrule}}\xspace}
\newcommand{\AtomLeftRule}[1]{\ensuremath{\hyperref[rule:grounding-schema]{#1\leftrule}}\xspace}

\newcommand{\pOrLeftInvRule}[1][p]{\ensuremath{\hyperref[th:add-inv-rules]{\textrm{$\plor\rightrule$-inv}}}\xspace}
\newcommand{\pAndRightInvRule}[1][p]{\ensuremath{\hyperref[th:add-inv-rules]{\textrm{$\pland\leftrule$-inv}}}\xspace}

\newcommand{\MALLTensorLeftRule}{\ensuremath{\hyperref[subfig:mall-mult-frag]{\tensor\leftrule}}\xspace}
\newcommand{\MALLTensorRightRule}{\ensuremath{\hyperref[subfig:mall-mult-frag]{\tensor\rightrule}}\xspace}
\newcommand{\MALLParLeftRule}{\ensuremath{\hyperref[subfig:mall-mult-frag]{\parr\leftrule}}\xspace}
\newcommand{\MALLParRightRule}{\ensuremath{\hyperref[subfig:mall-mult-frag]{\parr\rightrule}}\xspace}
\newcommand{\MALLOneLeftRule}{\ensuremath{\hyperref[subfig:mall-mult-frag]{\One\leftrule}}\xspace}
\newcommand{\MALLOneRightRule}{\ensuremath{\hyperref[subfig:mall-mult-frag]{\One\rightrule}}\xspace}
\newcommand{\MALLStarLeftRule}{\ensuremath{\hyperref[subfig:mall-mult-frag]{{}^\dual\leftrule}}\xspace}
\newcommand{\MALLStarRightRule}{\ensuremath{\hyperref[subfig:mall-mult-frag]{{}^\dual\rightrule}}\xspace}

\newcommand{\MALLOrLeftRule}{\ensuremath{\hyperref[subfig:mall-add-frag]{\lor\leftrule}}\xspace}
\newcommand{\MALLOrRightRule}[1]{\ensuremath{\hyperref[subfig:mall-add-frag]{{\lor\rightrule}_{#1}}}\xspace}
\newcommand{\MALLAndLeftRule}[1]{\ensuremath{\hyperref[subfig:mall-add-frag]{{\land\leftrule}_{#1}}}\xspace}
\newcommand{\MALLAndRightRule}{\ensuremath{\hyperref[subfig:mall-add-frag]{\land\rightrule}}\xspace}
\newcommand{\MALLBotLeftRule}{\ensuremath{\hyperref[subfig:mall-add-frag]{\bot\leftrule}}\xspace}

\newcommand{\MALLTopRightRule}{\ensuremath{\hyperref[subfig:mall-add-frag]{\top\rightrule}}\xspace}

\newcommand{\MALLAxRule}{\ensuremath{\hyperref[subfig:mall-struct-frag]{\text{AX}}}\xspace}
\newcommand{\MALLEmptyRule}{\ensuremath{\hyperref[subfig:mall-struct-frag]{\text{EMP}}}\xspace}
\newcommand{\MALLCutRule}{\ensuremath{\hyperref[subfig:mall-struct-frag]{\text{CUT}}}\xspace}
\newcommand{\MALLMixRule}{\ensuremath{\hyperref[subfig:mall-mult-frag]{\text{MIX}}}\xspace}

\newcommand{\qMALLTensorLeftRule}{\ensuremath{\hyperref[subfig:qmall-mult-frag]{\tensor\leftrule}}\xspace}
\newcommand{\qMALLTensorRightRule}{\ensuremath{\hyperref[subfig:qmall-mult-frag]{\tensor\rightrule}}\xspace}
\newcommand{\qMALLParLeftRule}{\ensuremath{\hyperref[subfig:qmall-mult-frag]{\parr\leftrule}}\xspace}
\newcommand{\qMALLParRightRule}{\ensuremath{\hyperref[subfig:qmall-mult-frag]{\parr\rightrule}}\xspace}
\newcommand{\qMALLOneLeftRule}{\ensuremath{\hyperref[subfig:qmall-mult-frag]{\One\leftrule}}\xspace}
\newcommand{\qMALLOneRightRule}{\ensuremath{\hyperref[subfig:qmall-mult-frag]{\One\rightrule}}\xspace}
\newcommand{\qMALLStarLeftRule}{\ensuremath{\hyperref[subfig:qmall-mult-frag]{{}^\dual\leftrule}}\xspace}
\newcommand{\qMALLStarRightRule}{\ensuremath{\hyperref[subfig:qmall-mult-frag]{{}^\dual\rightrule}}\xspace}

\newcommand{\qMALLOrRightRule}[1]{\ensuremath{\hyperref[subfig:qmall-add-frag]{{\lor\rightrule}_{#1}}}\xspace}
\newcommand{\qMALLAndLeftRule}[1]{\ensuremath{\hyperref[subfig:qmall-add-frag]{{\land\leftrule}_{#1}}}\xspace}

\newcommand{\qMALLpOrLeftRule}[1][p]{\ensuremath{\hyperref[subfig:qmall-add-frag]{\plor[#1]\leftrule}}\xspace}
\newcommand{\qMALLpOrRightRule}[1][p]{\ensuremath{\hyperref[subfig:qmall-add-frag]{\plor[#1]\rightrule}}\xspace}
\newcommand{\qMALLpAndLeftRule}[1][p]{\ensuremath{\hyperref[subfig:qmall-add-frag]{\pland[#1]\leftrule}}\xspace}
\newcommand{\qMALLpAndRightRule}[1][p]{\ensuremath{\hyperref[subfig:qmall-add-frag]{\pland[#1]\rightrule}}\xspace}
\newcommand{\qMALLBotLeftRule}{\ensuremath{\hyperref[subfig:qmall-add-frag]{\bot\leftrule}}\xspace}
\newcommand{\qMALLExFalsoRule}{\ensuremath{\text{\hyperref[subfig:qmall-add-frag]{EFQ}}}\xspace}
\newcommand{\qMALLTopRightRule}{\ensuremath{\hyperref[subfig:qmall-add-frag]{\top\rightrule}}\xspace}

\newcommand{\qMALLAxRule}{\ensuremath{\hyperref[subfig:qmall-struct-frag]{\text{AX}}}\xspace}
\newcommand{\qMALLEmptyRule}{\ensuremath{\hyperref[subfig:qmall-struct-frag]{\text{EMP}}}\xspace}
\newcommand{\qMALLCutRule}{\ensuremath{\hyperref[subfig:qmall-struct-frag]{\text{CUT}}}\xspace}
\newcommand{\qMALLMixRule}{\ensuremath{\hyperref[subfig:qmall-struct-frag]{\text{MIX}}}\xspace}

\newcommand{\hMALLTensorLeftRule}{\ensuremath{\hyperref[subfig:hmall-mult-frag]{\tensor\leftrule}}\xspace}
\newcommand{\hMALLTensorRightRule}{\ensuremath{\hyperref[subfig:hmall-mult-frag]{\tensor\rightrule}}\xspace}
\newcommand{\hMALLParLeftRule}{\ensuremath{\hyperref[subfig:hmall-mult-frag]{\parr\leftrule}}\xspace}
\newcommand{\hMALLParRightRule}{\ensuremath{\hyperref[subfig:hmall-mult-frag]{\parr\rightrule}}\xspace}
\newcommand{\hMALLOneLeftRule}{\ensuremath{\hyperref[subfig:hmall-mult-frag]{\One\leftrule}}\xspace}
\newcommand{\hMALLOneRightRule}{\ensuremath{\hyperref[subfig:hmall-mult-frag]{\One\rightrule}}\xspace}
\newcommand{\hMALLStarLeftRule}{\ensuremath{\hyperref[subfig:hmall-mult-frag]{{}^\dual\leftrule}}\xspace}
\newcommand{\hMALLStarRightRule}{\ensuremath{\hyperref[subfig:hmall-mult-frag]{{}^\dual\rightrule}}\xspace}

\newcommand{\hMALLOrLeftRule}{\ensuremath{\hyperref[subfig:hmall-add-frag]{\lor\leftrule}}\xspace}
\newcommand{\hMALLOrRightRule}{\ensuremath{\hyperref[subfig:hmall-add-frag]{\lor\rightrule}}\xspace}
\newcommand{\hMALLAndLeftRule}{\ensuremath{\hyperref[subfig:hmall-add-frag]{\land\leftrule}}\xspace}
\newcommand{\hMALLAndRightRule}{\ensuremath{\hyperref[subfig:hmall-add-frag]{\land\rightrule}}\xspace}
\newcommand{\hMALLBotLeftRule}{\ensuremath{\hyperref[subfig:hmall-add-frag]{\bot\leftrule}}\xspace}

\newcommand{\hMALLTopRightRule}{\ensuremath{\hyperref[subfig:hmall-add-frag]{\top\rightrule}}\xspace}

\newcommand{\hMALLAxRule}{\ensuremath{\hyperref[subfig:hmall-struct-frag]{\text{AX}}}\xspace}
\newcommand{\hMALLEmptyRule}{\ensuremath{\hyperref[subfig:hmall-struct-frag]{\text{EMP}}}\xspace}
\newcommand{\hMALLCutRule}{\ensuremath{\hyperref[subfig:hmall-struct-frag]{\text{CUT}}}\xspace}
\newcommand{\hMALLMixRule}{\ensuremath{\hyperref[subfig:hmall-mult-frag]{\text{MIX}}}\xspace}

\newcommand{\hMALLExtWeakRule}{\ensuremath{\text{\hyperref[subfig:hmall-external-structural]{EW}}}\xspace}
\newcommand{\hMALLExtContRule}{\ensuremath{\text{\hyperref[subfig:hmall-external-structural]{EC}}}\xspace}

\begin{document}

\title{Adequate Losses via Quantitative Linear Logic}

\author{Matteo Capucci}
\authornotemark[1]
\affiliation{
	\institution{University of Strathclyde and Independent Researcher}
	\city{Modena}
	\country{IT}
}
\email{matteo.capucci@gmail.com}

\author{Robert Atkey}
\authornotemark[1]
\affiliation{
	\institution{University of Strathclyde}
	\city{Glasgow}
	\country{UK}
}
\email{robert.atkey@strath.ac.uk}

\author{Charles Grellois}
\affiliation{
	\institution{Department of Computer Science, University of Sheffield}
	\city{Sheffield}
	\country{UK}
}
\email{c.grellois@sheffield.ac.uk}

\author{Matthew Daggitt}
\affiliation{
	\institution{University of Western Australia}
	\city{Perth}
	\country{Australia}
}
\email{matthew.daggitt@uwa.edu.au}

\author{Ekaterina Komendantskaya}
\authornote{Author supported by ARIA grant MSAI-PR01-P05.}
\affiliation{
	\institution{Heriot-Watt and Southampton Universities}
	\city{Edinburgh and Southampton}
	\country{UK}
}
\email{ek1u23@soton.ac.uk}

\renewcommand{\shortauthors}{Capucci et al.}

\begin{abstract}
	As neural components are increasingly embedded in existing symbolic software---including safety-critical systems---the question arises of how to specify and enforce the safety of the newly introduced neural parts. Unlike traditional logical specifications, these must be amenable not only to the standard Boolean interpretation, but also to training and optimisation. The latter calls for a quantitative interpretation of the logical syntax, subject to further requirements such as smoothness and differentiability. Moreover, the qualitative and quantitative sides of the logic must share a unifying proof-theoretic and categorical semantics. Finally, the new logic should link cleanly to the substructural and program logics that underpin the verification of existing symbolic programs.

	In this paper, we present a logic that ticks all of these boxes.
	We introduce a family of calculi, \pQLL, indexed by a hardness degree $p$, prove a cut-elimination theorem for them, and establish completeness with respect to enriched residuated `soft' lattices.
	At $p = \infty$, \pQLL reduces to multiplicative additive linear logic (MALL), and provability in \pQLL converges to provability in MALL as $p \to \infty$.
	We express optimisation objectives in the syntax of this logic and prove the quantitative adequacy of neuro-symbolic loss functions---a result that has eluded the neuro-symbolic machine learning community for nearly a decade.
\end{abstract}

\ccsdesc[500]{Theory of computation~Linear logic}
\ccsdesc[500]{Theory of computation~Proof theory}
\ccsdesc[500]{Theory of computation~Semantics and reasoning}
\ccsdesc[500]{Theory of computation~Logic and verification}
\ccsdesc[500]{Computing methodologies~Machine learning approaches}
\keywords{Linear Logic, Differentiable Logic, Fuzzy Logic, Sequent Calculus, Neuro-Symbolic AI, Neural Network Verification, Property-Driven Training, Enriched Category Theory}


\maketitle


\section{Introduction}
\label{sec:intro}

Neural networks are rapidly becoming ordinary software components: they are embedded in controllers, pipelines, and safety-critical systems, and so inherit the specification and verification obligations that programming language research has long articulated for conventional code.
Unfortunately, Neural Network Verification (NNV) is hard~\cite{CordeiroDGIJKKLMSW25}: it can be difficult both to formulate safety specifications and to ensure that networks satisfy them.
%

To address the former, the field has produced numerous domain-specific languages (DSLs) for encoding the interaction between logical reasoning and neural networks.
These DSLs span a broad design space. Systems such as Scallop~\cite{liScallopLanguageNeurosymbolic2023} and DeepProbLog~\cite{manhaeveDeepProbLogNeuralProbabilistic2018} treat the outputs of neural networks as probabilistic inputs to a symbolic reasoning engine, so that logical inference can be integrated into the learning pipeline without directly constraining the learned function. In contrast, systems such as Logic Tensor Networks~\cite{badreddineLogicTensorNetworks2022}, DL2~\cite{fischer2019dl2}, and \vehicle~\cite{daggittVehicleBridgingEmbedding2025,daggitt2026compositional} compile logical specifications into architectures or loss functions that push the network itself to satisfy the specified properties.
This paper is motivated by the latter class of systems,
and studies \textbf{property-driven training}, a set of methods that train neural networks directly on a specification\footnote{Note that training is but one form of model repair/synthesis, see e.g.~\cite{yuviler2025enhancing,ijcai2022p767} for other forms.}.

These methods share a common computational abstraction. First, the specification is encoded as a Boolean predicate $\varphi : (\R^m \to \R^n) \to \Prop$. Next, this predicate is translated into a differentiable function $\sem{\varphi} : (\R^m \to \R^n) \to \R$ which produces an estimate of \emph{how true} the specification is.
The translation is performed using what practitioners commonly call a \emph{differentiable logic}: an interpretation of logical connectives and quantifiers as differentiable operations over real-valued truth degrees that allows the compiler to perform a compositional translation of the Boolean specification.
Subsequently, when given a concrete neural network $f_\theta : \R^m \to \R^n$ parameterised by weights $\theta \in \R^v$, the gradient $\frac{\partial \sem{\varphi}(f_\theta)}{\partial \theta}$ can be used during training to guide the optimisation process towards a $\theta$ such that $\varphi(f_\theta)$ holds.

At first glance, differentiable logics appear to be little more than a modern incarnation of many-valued logics.
Logics with truth values in some subset of the real line have been studied for at least a century, counting from the works of {\L}ukasiewicz~\cite{lukasiewicz1920three} and G\"{o}del~\cite{metcalfeProofTheoryFuzzy2009}, and moving on to fuzzy logic~\cite{zadehFuzzySets1965,cintulaHandbookMathematicalFuzzy2011}, continuous model theory \cite{yaacovModelTheoryMetric2008}, and recent work on quantitative algebra \cite{mardareQuantitativeAlgebraicReasoning2016,bacciPropositionalLogicsLawvere2023,bacciPolynomialLawvereLogic2024,bacciInductionRecursionPrinciples2026}.

Consequently, many differentiable logics and associated tools have been proposed and studied~\cite{fischer2019dl2,van2022analyzing,varnaiRobustnessMetricsLearning2020,badreddineLogicTensorNetworks2022,badreddineLogLTNDifferentiableFuzzy2023}.
The central claim common to all of these works is that performing gradient descent to minimise the output of the compiled specification drives the network to satisfy the specification's Boolean interpretation.
From a logical perspective, this amounts to claiming \emph{adequacy} of the differentiable logic (and hence of the real-valued semantics of the DSL) as a proxy for satisfaction: the logic should be \emph{sound}, in the sense that a low loss guarantees the model satisfies the specification, and \emph{complete}, in the sense that a satisfying model is reflected by a correspondingly low loss.

\subsection{The problems}
\label{sec:problem}

\begin{wrapfigure}{r}{0.38\textwidth}
	\centering
	\vspace*{-1\baselineskip}
	\resizebox{0.36\textwidth}{!}{%
	\begin{tikzpicture}[>=stealth, every node/.style={font=\small}, line width=0.6pt]
		\node[font=\footnotesize, draw, rounded corners, inner sep=4pt] (opt)  at (0,1.5)       {optimisation};
		\node[font=\footnotesize, draw, rounded corners, inner sep=4pt] (sat)  at (-1.55,-0.75) {satisfaction};
		\node[font=\footnotesize, draw, rounded corners, inner sep=4pt] (prov) at (1.55,-0.75)  {provability};
		\draw[<->] (opt)  to[bend right=20] node[pos=0.6, above, sloped, inner sep=5pt, font=\footnotesize\bfseries]{adequacy} (sat);
		\draw[<->] (sat)  to[bend right=20] node[midway, below, inner sep=1.5pt, font=\footnotesize\bfseries]{completeness} (prov);
		\draw[<->] (prov) to[bend right=20] node[pos=0.4, above, sloped, inner sep=5pt, font=\footnotesize\bfseries]{justification} (opt);
	\end{tikzpicture}%
	}
	\caption{\small{The three facets a differentiable logic must reconcile.}}
	\label{fig:triangle}
\end{wrapfigure}

Yet a century of real-valued logic has not settled the matter~\cite{van2022analyzing,flinkow2025comparing,affeldt2026foundationdifferentiablelogicsusing,ldl-coq}, because the mismatch runs deeper than the truth values of formulae.
Using loss functions as logical semantics demands three things (\cref{fig:triangle}) that no existing logic can meet at once, which we now examine in turn.

The first aspect concerns the match between \emph{satisfaction} of the specification and the result of \emph{optimisation} of the loss: we call this \emph{adequacy} (\cref{fig:triangle}).
Consider, for example, the result in \cite{fischer2019dl2}  that proves that the proposed differentiable logic (DL2) is adequate in the sense that $\sem{\varphi} = 0$ if and only if the network satisfies $\varphi$, as verified by a second Boolean semantics.
This kind of \emph{extremal} adequacy is of little use since, in practice, the loss is never \emph{exactly} $0$!
Even worse, such a notion does not explain why \emph{optimising} the logical loss should lead to better \emph{satisfaction} of the specification---which, empirically, it does.

The hard-to-swallow pill here is that a \textbf{qualitative, `yes or no' notion of satisfaction has little hope of ever being adequate}.
Indeed, the logical loss is a real-valued quantity that ranges over a broad spectrum of values: if we want these values to mean something logical, then the logic must be able to accommodate this variation and the nuance that comes with it at a syntactic/proof-theoretic level.
Secondly, the properties we are checking are seldom qualitative, clear-cut ones: they amount to measuring errors, distances, ratios, and the like.

For instance, a common specification enforced in neural classifiers is \emph{robustness} \cite{casadioNeuralNetworkRobustness2022}:
predictions being robust under small perturbations of the input:\footnote{In this paper, we assume that $| \_ |$ is given by the $L^\infty$ norm. However, other norms, such as the Euclidean ($L^2$) norm, are sometimes used in the literature. The presented results do not depend on the specific choice of norm.}
\begin{equation}\label{eq:robustness}
	\tag{robustness}
	\forall x \in X,\ |x- \hat{x}| \leq \varepsilon \implies |f_\theta(x) - f_\theta(\hat{x})| \leq \delta
\end{equation}
where $\varepsilon$ and $\delta$ are chosen constants, and $(\hat x, \hat y) \in D$ is a fixed datapoint.
Our intuition tells us that there are degrees to the violation of this specification, and it is somewhat unreasonable that two nearly identical functions, like those in \cref{fig:robustness-comparison}, should receive completely opposite verdicts.
The main thesis of this paper is thus 
that
\begin{center}
	\itshape
	the verification of machine learning models should be \\ founded on a \textbf{quantitative} reconception of logic.
\end{center}

\begin{figure}[t]
	\centering
	\pgfplotsset{
		robustness plot/.style={
			width=\linewidth,
			height=.68\linewidth,
			axis lines=middle,
			axis line style={-Stealth, thin},
			xmin=-1.8, xmax=2,
			ymin=-1.55, ymax=1.55,
			xtick={-1,1},
			xticklabels={\hspace{-4ex}$-\varepsilon$, \hspace{4ex}$+\varepsilon$},
			ytick={-1,1},
			yticklabels={\raisebox{-5ex}{$-\delta$}, \raisebox{4ex}{$+\delta$}},
			ticklabel style={font=\scriptsize},
			every tick/.style={thin, black},
			clip=false,
		},
		/pgf/declare function={sigm(\x) = 0.8*tanh(3.5*\x);},
		/pgf/declare function={plot1(\x) = sigm(\x);},
		/pgf/declare function={plot2(\x) = sigm(\x) + 0.45*exp(-12*(\x-0.6)^2);}
	}
	\begin{subfigure}[t]{.48\linewidth}
		\centering
		\begin{tikzpicture}
			\begin{axis}[robustness plot]
				\draw[densely dashed, gray] (-1,-1) rectangle (1,1);
				\addplot[thick, blue, domain=-2:2, samples=200, smooth]
					{plot1(x)};
				\draw[green!60!black, very thick] (0.614,0.779) -- (0.614,1);
				\fill (0,0) circle[radius=1.2pt]
					node[above left, font=\scriptsize] {$(\hat{x}, f_\theta(\hat{x}))$};
			\end{axis}
		\end{tikzpicture}
		\label{fig:robust-jumpy}
	\end{subfigure}
	\hfill
	\begin{subfigure}[t]{.48\linewidth}
		\centering
		\begin{tikzpicture}
			\begin{axis}[robustness plot]
				\draw[densely dashed, gray] (-1,-1) rectangle (1,1);
				\addplot[thick, blue, domain=-2:2, samples=200, smooth]
					{plot2(x)};
				\draw[red, very thick] (0.614,1) -- (0.614,1.228);
				\fill (0,0) circle[radius=1.2pt]
					node[above left, font=\scriptsize] {$(\hat{x}, f_\theta(\hat{x}))$};
			\end{axis}
		\end{tikzpicture}
		\label{fig:robust-smooth}
	\end{subfigure}
	\caption{\small Two near-identical behaviours judged very differently by a binary notion of satisfaction.}
	\label{fig:robustness-comparison}
\end{figure}

The second aspect concerns the seemingly irreconcilable tension plaguing differentiable logic: the trade-off between the algebraic/proof-theoretic and analytic properties of the connectives.
In fact, to behave well during gradient-descent, the connectives of a differentiable logic should be entrywise strictly increasing differentiable operations, i.e.\ such that $\partial_i {\odot} > 0$ for both arguments $i=1,2$ of $\odot$.
This allows both operands to receive a meaningful gradient signal, avoiding situations in which one \emph{shadows} (in the words of\ \cite{varnaiRobustnessMetricsLearning2020}) the other since its truth already `suffices' to guarantee the truth of the whole expression.
The paradigmatic example is max ($\lor$): when $f(x) > g(x)$ then none of the gradient signal coming from $f(x) \lor g(x)$ will flow in the right branch, leading to fragile results (see~\cite{van2022analyzing,flinkow2025comparing}).

On the other hand, if we allow substructural logic as a proof-theoretical foundation, we are led to organize our logical connectives in either \emph{additive} or \emph{multiplicative} ones (using Girard's terminology for linear logic \cite{Gir95}).
Roughly speaking, additives are the connectives users expect to write their specifications in, since their rules reflect the most intuitive usage, compared to multiplicatives.
Then the established algebraic semantics based on residuated lattices\footnotemark~\cite{galatos2007residuated,cintulaHandbookMathematicalFuzzy2011a} forces the additives to be interpreted as the lattice operations of join ($\lor$) and meet ($\land$).
\footnotetext{Specifically, residuated structures on sublattices of $(\Reals, \leq)$. Recall that a residuated lattice is given by $(R, \land, \lor, \odot, 1, \multimap)$, where $(R, \land, \lor)$ is a lattice, $(R, \odot, 1)$ is a monoid with a neutral element $1$, and the following \emph{residuation property} holds: $a \odot b \leq c$ iff $a \leq b \multimap c$.}

Here's then the dilemma: the lattice operations on $\R$ do not enjoy the analytical properties mentioned above, and thus differentiable logics are forced to choose between empirical performance and logical elegance (the \emph{justification} arrow in \cref{fig:triangle}).
For example, Varnai et al.\ prove a famous \emph{no-go result} that no logic can have connectives that are simultaneously idempotent, associative and \emph{shadow-lifting}, and propose to sacrifice associativity~\cite{varnaiRobustnessMetricsLearning2020}.

And yet, on the reals there are obvious candidate `additive' and `multiplicative' operations, \emph{namely addition and multiplication themselves}!
Most importantly, these are readily interpretable operations which are already used in quantitative mathematics (probability theory, statistics, machine learning, etc.), in the very losses where the semantics of a specification is put in property-driven training.

When this \textbf{naturality principle} is not respected, even sacrificing logic for analysis does not guarantee success: an empirical study~\cite{flinkowQuantitativeLinearLogic2026} reports that, on a set of verification problems, none of the most popular differentiable logics successfully translates minimisation of the logical loss into good verification performance, as measured by special-purpose verifiers such as Marabou \cite{katz2019marabou}.

Hence our second thesis is that
\begin{center}
	\itshape
	a logic for NNV should have semantics in \textbf{natural} and \textbf{analytically adequate} operations,
	\\\textbf{without sacrificing the additives of the proof calculus}.
\end{center}

\subsection{Our contribution: quantitative logic as a new paradigm}
In this paper we propose a new logical paradigm to solve all these challenges at once, \textbf{quantitative logic}.
Logics \emph{about quantity} (such as fuzzy logics, continuous model theory \cite{yaacovModelTheoryMetric2008}, or the logics of the Lawvere quantale from \cite{bacciPropositionalLogicsLawvere2023,bacciPolynomialLawvereLogic2024}), concern formulae with intended semantics in the reals, but their notion of logical validity remains a yes-or-no affair.
Rather, following Martin-L\"of \cite{martin-lofMeaningsLogicalConstants1996} and Lawvere \cite{lawvereMetricSpacesGeneralized1973}, we maintain that \textbf{a quantitative logic shall have quantitative \emph{proofs} and quantitative \emph{judgments}}, and thus also a quantitative notion of \emph{satisfaction}.
In particular, Lawvere argued convincingly that such \emph{generalized logics} \cite{lawvereMetricSpacesGeneralized1973} arise by \emph{enriching} the entailment relation in a suitably structured \textbf{alethic theory} (our terminology) $\V$---for us, a structure of extended reals.

\providecommand{\columncolor}[1]{}
\providecommand{\cellcolor}[1]{}
\begin{table}[tbp]
	\captionsetup{position=bottom, skip=4pt}
	\footnotesize
	\centering
	\begin{tabularx}{\textwidth}{|c|c|*{3}{>{\centering\arraybackslash}X|}}
		\cline{2-5}
		\multicolumn{1}{c|}{} & \textbf{polarity} & \multicolumn{2}{c|}{\textbf{additive}} & \textbf{multiplicative} \\
		\hline
		\multirow{2}{9ex}[-2ex]{\centering\parbox[c]{9ex}{\centering\textbf{duality}\\$a^* := 1/a$}} & \textbf{positive} &
		\multicolumn{2}{>{\columncolor{red!15}}c|}{\begin{tikzcd}[ampersand replacement=\&, column sep=small]
			\begin{array}{c}
				\begin{gathered}
					\bot := 0 \\[-.5ex]
					a \padd[p] b
				\end{gathered}
			\end{array}
			\&\&
			\begin{array}{c}
				\begin{gathered}
					\bot := 0 \\[-.5ex]
					a \padd[\infty] b := a \lor b
				\end{gathered}
			\end{array}
			\arrow["{p \to \infty\ }"', to=1-3, from=1-1]
		\end{tikzcd}} & \cellcolor{blue!15}$\begin{gathered} \One := 1 \\[-.5ex] a \tensor b\end{gathered}\quad { \infty \tensor 0 := 0}$ \\[2ex]
		\hhline{~|----|}
		      & \textbf{negative} & \multicolumn{2}{>{\columncolor{blue!15}}c|}{\begin{tikzcd}[ampersand replacement=\&, column sep=small]
				\begin{array}{c}
					\begin{gathered}
						\top := \infty \\[-.5ex]
						a \pcoadd[p] b
					\end{gathered}
				\end{array}
				\&\&
				\begin{array}{c}
					\begin{gathered}
						\top := \infty \\[-.5ex]
						a \pcoadd[\infty] b := a \land b
					\end{gathered}
				\end{array}
				\arrow["{p \to \infty\ }"', to=1-3, from=1-1]
			 \end{tikzcd}} & \cellcolor{red!15}$\begin{gathered} \One := 1 \\[-.5ex] a \cotensor b\end{gathered} \quad { \infty \cotensor 0 := \infty}$ \\[2ex]
		\hline
	\end{tabularx}
	\vspace*{1ex}
	\caption{\small The family of algebraic structures we collectively refer to as the \textbf{multiplicative reals} $\pMulReals$.
	Each $0 < p \leq \infty$ gives a choice of additives.
	Red and blue backgrounds denote, respectively, disjunctive and conjunctive operations, based on the behaviour of their restriction to $\{0,+\infty\}$.
	Positive and negative polarities relate to each other via De Morgan duality.
	An isomorphic presentation on the (additive) reals is given in \cref{table:add-reals}.}
	\label{table:posreals-operations}
\end{table}

Our first contribution is to propose a proof-theoretic counterpart to Lawvere's ideas in the quantitative realm.
We introduce \textbf{quantitative sequent calculi}, which are calculi `enriched' in one of the algebraic structures presented in \cref{table:posreals-operations} (and later in \cref{sec:pos-reals}).
In such an `enriched' calculus, we can choose any operation on $\V$ to join the premises of a rule and any constant on $\V$ to be the premise of an axiom.
The proof-theoretic properties of the calculus become then tightly related to the way this assignment is done.

Thus the choices of \cref{table:posreals-operations} and \cref{fig:pqll} are essentially dictated by cut-elimination (\cref{th:cut-elimination}), for which we need that the operations used for multiplicative rules distribute over the ones we choose for additive ones.
Indeed, having decided to use multiplication for multiplicative rules, our additives must necessarily be \emph{$p$-sums}~\cite{aubrunMultiplicativePropertyCharacterizes2011,leinsterMultiplicativeCharacterizationPower2012}\footnotemark~which are the operations defined as:
\footnotetext{Note these are Yager's `$s$-norms' from \cite{yagerGeneralClassFuzzy1980} but without truncation to $[0,1]$.}
\begin{equation}
	0 < p < \infty, \qquad
	a \padd b = (a^p + b^p)^{1/p},
	\qquad
	a \pcoadd b = (a^{-p} + b^{-p})^{-1/p},
\end{equation}
with the convention $a+\infty=\infty$. Notably, by taking the limit $p \to \infty$ one recovers $\pcoadd[\infty] = \land$ and $\padd[\infty] = \lor$.
We describe this structure in more detail in \cref{sec:pos-reals}.


\begin{figure}
	\centering
	\begin{tabular}{c}
			\begin{subfigure}[t]{\linewidth}
				\caption{Structural Rules}
				\vspace*{1ex}
				\label{subfig:struct-frag}
				\centering
				\scalebox{0.81}{%
				\begin{tabular}{ccccc}
					\begin{prooftree}
						\hypo{\red{1}}
						\infer1[\AxRule]{A \entails A}
					\end{prooftree}
					&
					\begin{prooftree}
						\hypo{\red{1}}
						\infer1[\EmptyRule]{\ \entails \ }
					\end{prooftree}
					&
					\begin{prooftree}
						\hypo{\red{0}}
						\infer1[\ExFalsoRule]{\Gamma \entails \Delta}
					\end{prooftree}
					&
					\begin{prooftree}
						\hypo{\Gamma \entails A, \Delta \redbbin{\tensor} \Gamma' , A \entails \Delta'}
						\infer1[\CutRule]{\Gamma, \Gamma' \entails \Delta, \Delta'}
					\end{prooftree}
					&
					\begin{prooftree}
						\hypo{\Gamma \entails \Delta \redbbin{\tensor} \Gamma' \entails \Delta'}
						\infer1[\MixRule]{\Gamma, \Gamma' \entails \Delta, \Delta'}
					\end{prooftree}
				\end{tabular}%
				}
			\end{subfigure}
		\\[2ex]
			\begin{subfigure}[t]{\linewidth}
				\caption{Rules for Multiplicatives}
				\vspace*{1ex}
				\label{subfig:mult-frag}
				\centering
				\scalebox{0.81}{%
				\begin{tabular}{c}
					\begin{tabular}{cccc}
						\begin{prooftree}
							\hypo{\red{1}}
							\infer1[\OneLeftRule]{\One \entails\ }
						\end{prooftree}
						&
						\begin{prooftree}
							\hypo{\Gamma \entails A, \Delta}
							\infer1[\DualLeftRule]{\Gamma, A^\dual \entails \Delta}
						\end{prooftree}
						&
						\begin{prooftree}
							\hypo{\Gamma, A \entails \Delta}
							\infer1[\DualRightRule]{\Gamma \entails A^\dual, \Delta}
						\end{prooftree}
						&
						\begin{prooftree}
							\hypo{\red{1}}
							\infer1[\OneRightRule]{\ \entails \One}
						\end{prooftree}
					\end{tabular}
					\\[4ex]
					\begin{tabular}{cccc}
						\begin{prooftree}
							\hypo{\Gamma, A \entails \Delta \redbbin{\tensor} \Gamma', B \entails \Delta'}
							\infer1[\ParLeftRule]{\Gamma, \Gamma', A \parr B \entails \Delta, \Delta'}
						\end{prooftree}
						&
						\begin{prooftree}
							\hypo{\Gamma, A,B \entails \Delta}
							\infer1[\TensorLeftRule]{\Gamma, A \tensor B \entails \Delta}
						\end{prooftree}
						&
						\begin{prooftree}
							\hypo{\Gamma \entails A,B, \Delta}
							\infer1[\ParRightRule]{\Gamma \entails A \parr B, \Delta}
						\end{prooftree}
						&
						\begin{prooftree}
							\hypo{\Gamma \entails A, \Delta \redbbin{\tensor} \Gamma',\entails B, \Delta'}
							\infer1[\TensorRightRule]{\Gamma,\Gamma' \entails A \tensor B, \Delta, \Delta'}
						\end{prooftree}
					\end{tabular}
				\end{tabular}%
				}
			\end{subfigure}
		\\[2ex]
			\begin{subfigure}[t]{\linewidth}
				\caption{Rules for Additives}
				\vspace*{1ex}
				\label{subfig:add-frag}
				\centering
				\scalebox{0.81}{%
				\begin{tabular}{c}
					\begin{tabular}{cc}
						\begin{prooftree}
							\hypo{\red{\infty}}
							\infer1[\BotLeftRule]{\Gamma, \bot \entails \Delta }
						\end{prooftree}
						&
						\begin{prooftree}
							\hypo{\red{\infty}}
							\infer1[\TopRightRule]{\Gamma \entails \top, \Delta}
						\end{prooftree}
					\end{tabular}
					\\[4ex]
					\setlength{\tabcolsep}{3pt}
					\begin{tabular}{cccc}
						\begin{prooftree}
							\hypo{\Gamma, A \entails \Delta \redbbin{\pcoadd[p]} \Gamma, B \entails \Delta}
							\infer1[\pOrLeftRule]{\Gamma, A \plor B \entails \Delta}
						\end{prooftree}
						&
						\begin{prooftree}
							\hypo{\Gamma, A \entails \Delta \redbbin{\padd[p]} \Gamma, B \entails \Delta}
							\infer1[\pAndLeftRule] {\Gamma, A \pland B \entails \Delta}
						\end{prooftree}
						&
						\begin{prooftree}
							\hypo{\Gamma \entails A , \Delta \redbbin{\padd[p]} \Gamma \entails B, \Delta}
							\infer1[\pOrRightRule]{\Gamma \entails A \plor B , \Delta }
						\end{prooftree}
						&
						\begin{prooftree}
							\hypo{\Gamma \entails A , \Delta \redbbin{\pcoadd[p]} \Gamma \entails B , \Delta}
							\infer1[\pAndRightRule]{\Gamma \entails A \pland B, \Delta}
						\end{prooftree}
					\end{tabular}
				\end{tabular}%
				}
			\end{subfigure}
		\end{tabular}
	\caption{Rules of $p$-hard Quantitative Linear Logic (\pQLL).}
	\label{fig:pqll}
\end{figure}

By decorating multiplicative rules with multiplication and additive rules with (harmonic) $p$-sums, we then obtain a family of quantitative analogues to (isomix) \MALL, which we call \pQLL (\emph{$p$-hard Quantitative Linear Logic})---see \cref{fig:pqll}.
The annotations on the rules are then used to define a quantitative notion of \emph{validity}, in the form of a function $\validity{-}$ that evaluates proof trees by traversing them and aggregating the constants in the leaves using the operations specified by the rules.
A sequent's \ref{eq:provability} is then the supremum of the validity of all its proofs.

With this notion in hand, we can solve all the aforementioned problems.
Here is an outline of the paper and of such a solution:

\begin{enumerate}
	\item After elaborating on \cref{table:posreals-operations} in \cref{sec:pos-reals}, in \cref{sec:QLL} we introduce quantitative calculi and \pQLL specifically.
	In \cref{sec:proof-theory}, we prove various proof-theoretic results: \pQLL admits cut-elimination (\cref{th:cut-elimination} proven in \cref{app:cut-elim}), with all its usual consequences; we show provability of a sequent in \MALL is equivalent to that sequent having positive provability in \pQLL (\cref{th:isomixmall-cons-and-adeq-over-qual-pqll}) or provability above $1$ in \pQLL[\infty] (\cref{th:isomixmall-cons-and-adeq-over-infty-qll}); and finally that provability in \pQLL converges to provability in \pQLL[\infty] as $p \to +\infty$.

	\item Then, in~\cref{sec:softale} we look at the model theory.
	We formulate a notion of algebraic model, \textbf{softales}, which generalizes residuated lattices and metric spaces by enriching the order structure in the multiplicative reals, and supports a `soft' notion of lattice connectives. For $p < \infty$, \textbf{the proof calculi \pQLL provide the first instance in the literature of additive rules with a differentiable semantics}, specifically in the canonical model $\V=\MulReals$.
	This also allows us to have a quantitative notion of satisfaction, for which we prove a kind of completeness result (\cref{th:grounded-completeness})---this addresses the last of the arrows in \cref{fig:triangle}.

	\item Finally, in \cref{sec:ML} we show how our theory works as a foundation for NNV in the restricted case of propositional specifications.
	There we formulate the \textbf{first quantitative adequacy theorem for a differentiable logic} (\cref{th:nnv-adequacy}), in the form $e^{-\sem{\varphi}} = \validity{\entails \varphi}_\Theory$ where $\Theory$ is a suitably defined \pQLL-theory pertaining to the specific verification problem at hand.

\end{enumerate}

We do not cover
the empirical performance of \pQLL on NNV problems in this paper.
This is the topic of \cite{flinkowQuantitativeLinearLogic2026}, which compares \pQLL to other similar logics, showing very encouraging results.





\subsection{Related Work}

\textbf{Neuro-symbolic methods.}
Neuro-symbolic AI, seen as a collection of methods of drawing from both symbolic logic and neural machine learning, goes back to McCulloch and Pitts~\cite{mccullochLogicalCalculusIdeas1943}, and was developed later in the 90s and early 2000s by H\"olldobler~\cite{DBLP:journals/apin/HolldoblerKS99}, Hitzler, Garcez~\cite{DBLP:series/faia/BesoldGBBDHKLLPPPZ21,badreddineLogicTensorNetworks2022}, Gabbay et al.~\cite{garcezNeuralSymbolicLearningSystems2002,DBLP:series/cogtech/GarcezLG2009}, among others.
The common idea in all of these works is to represent logical connectives via activation functions of neurons, and thus logical formulae as neural networks.
Evaluating a formula then amounts to running the network, which extends naturally to quantitative interpretations~\cite{DBLP:series/sci/KomendantskayaLS07}.
For some recent work in this area, see \cite{perovicLogicsMostCountable2025,heimanNeuralNetworksFuzzy2026a}.



Unlike this line of research, our question is not whether/how certain neural network architectures can be represented logically, but how arbitrary logical specifications (akin to the ones written in the language \vehicle~\cite{daggittVehicleBridgingEmbedding2025,daggitt2026compositional,vehicleTutorial2026}) can inform training of arbitrary neural models. From that perspective, putting logical specifications in loss functions (rather than architectures) is a method of neuro-symbolic training that applies to any model.

\textbf{Logic of quantity.}
Our work is part of a long logical tradition trying to combine symbolic reasoning with quantitative aspects.
Besides fuzzy logics \cite{zadehFuzzyLogic1988,goguenLogicInexactConcepts1969} and continuous logic \cite{yaacovModelTheoryMetric2008a}, we mention the recent line of work of Bacci et al. \cite{bacciPropositionalLogicsLawvere2023,bacciPolynomialLawvereLogic2024,bacciRationalLawvereLogic2026}, which spun out of work on quantitative behavioural metrics for probabilistic and weighted automata.

As argued above, QLL differs from these logics by being quantitative at the level of judgment, and not just in the models.
Notable exceptions are Lawvere's seminal work \cite{lawvereMetricSpacesGeneralized1973} (which however remains purely semantical) and, in the fuzzy logic tradition, the work of Pavelka \cite{pavelkaFuzzyLogicManyvalued1979,pavelkaFuzzyLogicII1979,pavelkaFuzzyLogicIII1979} (but we refer to \cite{vychodilFuzzyInequationalLogic2015}).
Pavelka-style completeness in fact is the same kind of completeness we use in \cref{sec:model-theory}.

\textbf{Linear and other substructural logics.}
The expression `quantitative linear logic' is already in use for a different, if related, idea: the study of linear logic proofs through \emph{weighted relational models}~\cite{laird2013weighted}, in which formulae denote sets and proofs denote matrices with entries in a semiring, such as $[0,\infty]$.
Ehrhard et al.~\cite{ehrhardSumProduct2023} show that such models of multiplicative linear logic encompass Bayesian networks, and adapt the sum-product algorithm to compute the weighted denotation of a proof efficiently.
In a similar vein, the word `quantitative' is used by Dal Lago et al.~\cite{dallagoQuantitative} for semantics that track the resource consumption of proofs and programs.
In both cases, quantity lives in the semantics, while provability itself remains a Boolean matter; in \pQLL, quantity is instead integrated at the level of proof theory, as a measure of \emph{how valid} a proof is.
We expect \pQLL to also admit weighted relational models, as a proof-relevant refinement of the semantics of \cref{sec:model-theory}; in the soft regime $p < \infty$, where proofs can be summed, these would likely resemble models of differential linear logic~\cite{ehrhard} more than those of \MALL.

Proof-theoretically, $\pQLL$ is closely related to Hypersequent and Deep Inference calculi, which we discuss in \cref{sec:hyperseq} and \cref{sec:deep-inference}.


\subsubsection*{Acknowledgments}
We thank Jairo Marulanda--Giraldo, Enrico Marchioni, and David Jaz Myers for helpful discussions;
Janis Bailitis for proof-reading an early draft of this paper, and for mechanizing \pQLL in \cite{Bailitis2026}; and Thomas Flinkow for being the first to run machine learning experiments with \pQLL~\cite{flinkowQuantitativeLinearLogic2026}.
We thank the anonymous reviewers of an earlier draft of this paper for their helpful corrections and suggestions, and the
POPL 2027 reviewers for suggesting various connections to other work.
Capucci, Atkey, and Komendantskaya are funded by ARIA grant MSAI-PR01-P05 as part of the \emph{Safeguarded AI} programme.

\section{The positive reals}
\label{sec:pos-reals}
Consider the extended positive reals $\PosReals$ with their usual order $\leq$.
It supports a breadth of structure we now describe.
We start from \textbf{(conjunctive) multiplication}, i.e. the following extension of the usual arithmetic multiplication:
\begin{equation}\label{eq:conj-multiplication}
	\begin{tabular}{c|ccc}
		$a \tensor b$      & $0$ & $a \in (0,\infty)$ & $\infty$ \\
		\cline{1-4}
		$0$                & $0$ & $0$                & $0$      \\
		$b \in (0,\infty)$ & $0$ & $ab$               & $\infty$ \\
		$\infty$           & $0$ & $\infty$           & $\infty$
	\end{tabular}
\end{equation}
To see why we call it conjunctive, look at its action on $\{0,\infty\}$.
It is commutative, associative, and unital, and we call the resulting ordered monoid the \textbf{multiplicative reals} and  denote it by $\MulReals$.

Crucially, \textbf{inversion} yields a duality $(-)^* : \PosReals^\op \to \PosReals$
\begin{equation}
	\begin{tabular}{c|ccc}
		$a^*$ & $0$      & $a \in (0,\infty)$ & $\infty$ \\
		\cline{1-4}
		      & $\infty$ & $1/a$              & $0$
	\end{tabular}
\end{equation}
which is also uniquely characterized by the property
\begin{equation}\label{eq:duality}
	a \tensor b \leq c^* \iff a \leq (b \tensor c)^*.
\end{equation}
This makes $(\PosReals, \leq, 1, \tensor, (-)^*)$ into a {$*$-\nobreak{}autonomous} monoidal poset \cite{barrAutonomousCategories1979}, a kind of (totally ordered and) involutive commutative residuated lattice (see \cite[Chapter~IV]{cintulaHandbookMathematicalFuzzy2011}).

By De Morgan duality, we obtain a second `multiplicative' operation, $\cotensor$ (\textbf{disjunctive multiplication}), which differs from the first only in that $0 \cotensor \infty = \infty$.
In particular, $\tensor$ \emph{linearly distributes} over $\cotensor$ in the sense of \cite{cockettLinearlyDistributiveFunctors1999}:
\begin{equation}\label{eq:ldc}
	(a \cotensor b) \tensor c \leq a \cotensor (b \tensor c).
\end{equation}
Moreover, since $1^* = 1$---and thus $a \tensor b \leq a \cotensor b$---we say the operations are \emph{isomix}, in the terminology of \cite{cockettProofTheoryFull1997}.
The residuation is explicitly given by $a \mulimp  b = a^* \cotensor b$.
This means $a \leq b \mulimp c \iff a \tensor b \leq c$.
\begin{equation}
	\begin{tabular}{c|ccc}
		$a \mulimp b$      & $0$      & $a \in (0,\infty)$ & $\infty$ \\
		\cline{1-4}
		$0$                & $\infty$ & $0$                & $0$      \\
		$b \in (0,\infty)$ & $\infty$ & $b/a$              & $0$      \\
		$\infty$           & $\infty$ & $\infty$           & $\infty$
	\end{tabular}
\end{equation}
We sometimes abuse notation and write $b/a$ for $a \mulimp b$.

Now, on the poset $\PosReals$, consider \textbf{sum} $\add$, extended to $\infty$ by $x \add \infty = \infty$, and its De Morgan dual, \textbf{harmonic sum}~\cite{grandisCategoriesNormsWeights2007}:
\begin{equation}
	a \coadd b := (a^* \add b^*)^*.
\end{equation}

Choosing $0 < p < \infty$, we can conjugate these operations by exponentiation\footnotemark~to obtain \textbf{$p$-sums} and \textbf{harmonic $p$-sums}:
\footnotetext{Extended by $0^p = 0$ and $\infty^p = \infty$ for all $p$. We do not exponentiate either by $0$ or $\infty$.}
\begin{equation}
	a \padd[p] b := (a^p \add b^p)^{1/p},
	\qquad
	a \pcoadd[p] b := (a^p \coadd b^p)^{1/p}.
\end{equation}

These operations are all commutative and associative. The $p$-sums have unit $0$ and absorbing element $\infty$, while harmonic $p$-sums have unit $\infty$ and absorbing element $0$.
It is well-known that, as $p \to \infty$, these `soft' operations converge pointwise to `hard' ones, namely max ($\lor$) and min ($\land$) on $\PosReals$:

\begin{lemma}[Additive Collapse]
	\label{lemma:additive-collapse}
	\begin{equation}
		a \padd[p] b \conv[p \to \infty] a \lor b,
		\qquad
		a \pcoadd[p] b \conv[p \to \infty] a \land b,
	\end{equation}
\end{lemma}
\begin{proof}
	Classical, see e.g. \cite{mitrinovicAnalyticInequalities1970}.
\end{proof}

Thus we can consider $p$ as a \emph{degree of hardness} (with its inverse $1/p$ being \emph{softness} proper\footnote{These are similar to `temperature' and `coldness' in statistical mechanics \cite{baezWhatEntropy2024}.}), which governs, ultimately, how much idempotency we have renounced, since $2^{-1/p} \leq a/(a \padd a)$ for $p < \infty$.
Despite this difference, $\padd$ and $\pcoadd$ behave remarkably like their hard counterparts:

\begin{lemma}\label{lemma:p-sums-misc}
	Let $0 < p \leq \infty$, then:
	\begin{enumerate}
		\item for $a \leq a'$, we have
		\(
		a \pcoadd b \leq a' \pcoadd  b,
		\)
		and
		\(
		a \padd b \leq a' \padd b;
		\)
		\item for all $a,b \in \PosReals$,
		\(
		a \pcoadd b \,\leq\, a \,\leq\, a \padd b;
		\)
		\item for $p \leq q$,
		\(
		a \pcoadd[p] b \leq a \pcoadd[q] b, \quad a \padd[q] b \leq a \padd[p] b;
		\)
		and thus $a \pcoadd[p] b \leq a \land b$ and $a \lor b \leq a \padd[p] b$;
	\end{enumerate}
\end{lemma}





The most important fact about $p$-sums and multiplication, however, is the following:

\begin{lemma}\label{lemma:tensor-over-add-dist}
	Conjunctive multiplication distributes over $p$-sum and, dually, disjunctive multiplication distributes over harmonic $p$-sum:
	\begin{equation}
		(a \tensor b) \padd[p] (a \tensor c) = a \tensor (b \padd[p] c),
		\qquad
		(a \cotensor b) \pcoadd[p] (a \cotensor c) = a \cotensor (b \pcoadd[p] c).
	\end{equation}
\end{lemma}
\begin{corollary}\label{cor:tensor-over-add-dist}
	\begin{eqalign}
		c \mulimp (a \pcoadd[p] b) = (c \mulimp a) \pcoadd[p] (c \mulimp b),
		\qquad
		(a \padd[p] b) \mulimp c = (a \mulimp c) \pcoadd[p] (b \mulimp c).
	\end{eqalign}
\end{corollary}

\begin{lemma}\label{lemma:cotensor-over-add-conormal-dist}
	We also have
	\begin{eqalign}
		0 \leq a \cotensor 0
		\quad\text{and}\quad
		(a \cotensor b) \padd[p] (a \cotensor c) = a \cotensor (b \padd[p] c),
		\\
		a \tensor \infty \leq \infty
		\quad\text{and}\quad
		a \tensor (b \pcoadd[p] c) = (a \tensor b) \pcoadd[p] (a \tensor c).
	\end{eqalign}
\end{lemma}
\begin{corollary}\label{cor:cotensor-over-add-conormal-dist}
	\begin{eqalign}
		c \mulimp (a \padd[p] b) = (c \mulimp a) \padd[p] (c \mulimp b),
		\qquad
		(a \pcoadd[p] b) \mulimp c = (a \mulimp c) \padd[p] (b \mulimp c).
	\end{eqalign}
\end{corollary}

As an anticipation of the logical meaning these operations will assume later, and in order to introduce the necessary terminology, we recommend revisiting \cref{table:posreals-operations}.


\subsection{Varying hardness}
All in all, for each $0 < p \leq \infty$, we get a structure we denote as $\pMulReals[p]$.
Later, in \cref{sec:softale}, we will introduce its abstract algebraic form---that of \emph{softale}.
For now, note that $\pMulReals[q] \iso \pMulReals[p]$ for every $0 < p,q < \infty$, where the isomorphism means that there is a bijection, given by the power map $(-)^{p/q}$, that commutes with duality, tensor, and $p$-sums.
For $q=\infty$ this map is not bijective and only \emph{lax} sum-preserving, i.e. $a^0 \padd b^0 \leq (a \lor b)^0$.

Thus there really are two isomorphism classes of the structure we just presented, the \textbf{soft} one, epitomized by $\pMulReals[1]$, and a \textbf{hard} one, $\pMulReals[\infty] \equiv \PosReals_{\lor,\tensor}$.
When it comes to practical verification, however, the choice of $p$ influences the denotations and therefore the loss values, motivating our study for any value of $p$.

Finally, while the algebraic structure just described is particularly natural on $\PosReals$, $\Reals^\op$ supports precisely the same structure, as witnessed by the so-called \emph{Napier's isomorphism} (\cite[§2.2]{capucciQuantifiersQuantitativeReasoning2024}) $-\log \adj 1/\exp$; \cref{table:add-reals} shows the construction.
Through this isomorphism, the \emph{Lawvere quantale} $\mathbb{L}$ \cite{lawvereMetricSpacesGeneralized1973,bacciPropositionalLogicsLawvere2023} can be identified with the restriction at $[1,\infty]$ of the hard (multiplicative) reals.

\subsection{Qualitative truth}
When using $\pMulReals[p]$ as the object of truth values, we will often feel the need to compare it back to the usual truth values $\Prop$.
There are various possibilities, that we call \textbf{qualitative modes}, with different desirable properties.\footnote{Other choices are discussed in \cite[§2.3]{capucciQuantifiersQuantitativeReasoning2024}.}

In the \textbf{additive mode}, we consider $a=0$ `false' and everything else `true'.
Thus the map casting from `quantitative' to `qualitative' truth is
\begin{equation}\label{eq:add-mode-map}
	0<\_ : \PosReals \longto \Prop
\end{equation}
Notice this commutes with the conjunctive operations ($\tensor$ and $\pcoadd$), which are both mapped to $\land$.
Likewise $\padd$ is mapped to $\lor$, but note $\cotensor$ is not preserved.
Rather, we can only say that $a > 0$ \emph{and} $b>0$ imply $a \cotensor b > 0$ (so $\cotensor$ laxly commutes with $\land$), and that $a \cotensor b > 0$ implies $a > 0$ \emph{or} $b > 0$ ($\cotensor$ \emph{colaxly} commutes with $\lor$).
Negation is not preserved either way.
%
%

In the \textbf{multiplicative mode}, we pick $1$ as threshold of truth and consider
\begin{equation}\label{eq:mult-mode-map}
	1 \leq \_ : \PosReals \longto \Prop
\end{equation}
Conjunctive operations ($\tensor$ and $\pcoadd$) laxly commute with $\land$, while disjunctive ones ($\cotensor$ and $\padd$) colaxly commute with $\lor$.
When $p=\infty$, $\padd$ and $\pcoadd$ are preserved strongly.
We also have $a > 1$ implies $1 \leq a^*$.


\section{Quantitative Calculi}
\label{sec:QLL}

We now fix $0 < p \leq \infty$, and thus $\pMulReals[p]$ as the `alethic theory' (as introduced in \cref{sec:intro}) in which to `enrich' our logic.
We now define precisely how we make use of the structure of $\pMulReals[p]$ in a quantitative calculus.

\subsection{Definitions}\label{sec:defs}
Let $\Phi$ be a \textbf{language}, i.e. a set of \textbf{formulae}.
A \textbf{$p$-hard (quantitative) calculus} pertaining to $\Phi$ deals with sequents $\Gamma \entails \Delta$ where $\Gamma, \Delta$ are cedents $\Gamma,\,\Delta\ldots$, i.e. finite multisets\footnote{Being multisets, we do not consider them ordered, and thus exchange rules are implicit in the calculi.} of formulae, including the empty one.

In a quantitative calculus, sequents are combined in \textbf{structures}, defined as follows:
\begin{equation}\label{eq:hyperseq-grammar}
	\hyper{H} \grammareq
	\Gamma \entails \Delta \grammarsep \redbin{r} \grammarsep
	\hyper{H} \redbin{\tensor} \hyper{H} \grammarsep \hyper{H} \redbin{\cotensor} \hyper{H} \grammarsep \hyper{H} \redbin{\padd[p]} \hyper{H} \grammarsep \hyper{H} \redbin{\pcoadd[p]} \hyper{H}
\end{equation}
where $\Gamma \entails \Delta$ ranges in the possible sequents in the language $\Phi$ while $r$ ranges in $\PosReals$.

A structure of sequents involving no sequent is called \textbf{closed}, while a structure of sequents consisting of a single sequent is called \textbf{unary}.
We remark that closed structures can be \emph{evaluated} to real numbers, but \emph{per se} are syntactic objects.
We call the result of evaluating a closed structure $\hyper{K}$ its \textbf{alethic value} or \textbf{validity}, and write it as $\validity{\hyper{K}}$.


A \textbf{rule} has the form
\begin{prooftree}
	\hypo{\hyper{H}}
	\infer1{\hyper{H}'}
\end{prooftree}.
Its \textbf{premises} are the sequents appearing in $\hyper{H}$ and its \textbf{conclusions} those appearing in $\hyper{H}'$.
It is an \textbf{axiom} if the top structure of sequents is closed.
A set of rules $\rules{R}$ is a \textbf{calculus}.

\begin{definition}
	A \textbf{derivation} in the calculus $\rules{R}$ is defined inductively as follows (where $\red{\star} \in \{\red{\tensor}, \red{\cotensor}, \red{\padd[p]}, \red{\pcoadd[p]}\}$):
	\begin{enumerate}
		\item for each $\hyper{H}$, the \textbf{identity derivation}
		      \begin{prooftree}
			      \hypo{\hyper{H}}
			      \infer1[$\id_{\hyper{H}}$]{\hyper{H}}
		      \end{prooftree}\ %
		      is a derivation;
		\item every instantiation of a rule in $\rules{R}$ is a derivation,
		\item if \begin{prooftree}
			      \hypo{\hyper{H}}
			      \ellipsis{$\delta_1$}{\hyper{I}}
		      \end{prooftree}
		      \quad and
		      \begin{prooftree}
			      \hypo{\hyper{I}}
			      \ellipsis{$\delta_2$}{\hyper{K}}
		      \end{prooftree}
		      \quad are derivations, so is \begin{prooftree}
			      \hypo{\hyper{H}}
			      \ellipsis{$\dfrac{\delta_1}{\delta_2}$}{\hyper{K}}
		      \end{prooftree}
		\item if \begin{prooftree}
			      \hypo{\hyper{H}_1}
			      \ellipsis{$\delta_1$}{\hyper{K}_1}
		      \end{prooftree}
		      \quad and
		      \begin{prooftree}
			      \hypo{\hyper{H}_2}
			      \ellipsis{$\delta_2$}{\hyper{K}_2}
		      \end{prooftree}
		      \quad are derivations, so is
		      \begin{prooftree}
			      \hypo{\hyper{H}_1 \redbin{\star} \hyper{H}_2}
			      \ellipsis{$\delta_1 \redbin{\star} \delta_2$}{\hyper{K}_1 \redbin{\star} \hyper{K}_2}
		      \end{prooftree}
	\end{enumerate}
\end{definition}
We stress that a derivation can only mention finitely many rules.

\begin{notation}
	A derivation $\delta$ of $\hyper{H}$ is thus denoted as \begin{prooftree}
		\hypo{}
		\ellipsis{$\delta$}{\hyper{H}}
	\end{prooftree}
	\quad or \begin{prooftree}
		\hypo{}
		\ellipsis{$\delta'$}{}
		\infer1[$R$]{\hyper{H}}
	\end{prooftree}
	if we want to make the last rule explicit.
	It follows that a composite derivation in which the last rule of each operand is explicit can be written in either of two ways:
	\begin{equation}
		\begin{prooftree}
			\hypo{}
			\ellipsis{$\delta_1' \redbin{\star} \delta_2'$}{}
			\infer1[$R_1 \redbin{\star} R_2$]{\hyper{H}_1 \redbbin{\star} \hyper{H}_2}
			\infer1[$R$]{\hyper{K}}
		\end{prooftree}
		\hspace*{10ex}
		\begin{prooftree}
			\hypo{}
			\ellipsis{$\delta_1'$}{}
			\infer1[$R_1$]{\hyper{H}_1}
			\hypo{\!\red{\star}}
			\hypo{}
			\ellipsis{$\delta_2'$}{}
			\infer1[$R_2$]{\hyper{H}_2}
			\infer3[$R$]{\hyper{K}}
		\end{prooftree}
	\end{equation}
\end{notation}

Even though this definition accommodates derivations with DAG structure, when a calculus is \textbf{counary} (i.e. all of its rules have conclusions given by unary structures), derivations are in the shape of forests.
This will be the case for \pQLL.

A \textbf{closed} derivation is a derivation where the top structure of sequents is closed.
A \textbf{proof} is a closed derivation where the bottom structure of sequents is unary.
The \textbf{validity} $\validity{\delta}$ of a closed derivation $\delta$ is the alethic value of its top structure of sequents.

\begin{example}
	Here is the computation of the validity of a proof in \pQLL[1]:
	\begin{equation}
		\label{proof:soft-idempotency}
		\validity{\begin{prooftree}
			\hypo{\red{1}}
			\infer1[\AxRule]{A \entails A}
			\hypo{\red{\pcoadd[1]}}
			\hypo{\red{1}}
			\infer1[\AxRule]{A \entails A}
			\infer3[\pOrLeftRule]{A \plor A \entails A}
		\end{prooftree}} = \validity{\begin{prooftree}
			\hypo{\red{1}}
			\infer1[\AxRule]{A \entails A}
		\end{prooftree}} \pcoadd[1] \validity{\begin{prooftree}
			\hypo{\red{1}}
			\infer1[\AxRule]{A \entails A}
		\end{prooftree}} = 1 \pcoadd[1] 1 = 1/2.
	\end{equation}
\end{example}

The \textbf{provability} $\validity{\Gamma \entails \Delta}_{\rules{R}}$ of a sequent $\Gamma \entails \Delta$ in the calculus $\rules{R}$ is the supremum of the validity of its proofs:
\begin{equation}\label{eq:provability}
	\tag{provability}
	\validity{\Gamma \entails \Delta}_{\rules{R}}
	:= \bigvee\!\left\{
	\,\validity{\pi}\;\;\middle|\;\;
	\begin{prooftree}
		\hypo{}
		\ellipsis{$\pi$}{\Gamma \entails \Delta}
	\end{prooftree}
	\text{ closed}
	\right\}
\end{equation}
By replacing $\Gamma \entails \Delta$ with a structure of sequents, we extend the above definition%
\footnote{Observe there is a different way to extend the definition of provability to structures of sequents, by induction on their structure. For a counary calculus these two coincide, but not in general.}.
Observe that, as a consequence, any derivation from $\hyper{H}$ to $\hyper{H}'$ witnesses that $\validity{\hyper{H}} \leq \validity{\hyper{H'}}$.



We consider a sequent \textbf{(qualitatively) provable} if its provability is strictly positive.
This happens, for example, when a sequent admits no proofs since $\bigvee \varnothing = 0$ but, more interestingly, also when it admits \emph{only proofs of validity $0$} (i.e. \textbf{invalid} proofs).

Note, therefore, that to know the provability of a sequent it is not enough to find \emph{a} proof, but potentially \emph{all proofs} are necessary.
For \pQLL below, we prove cut-elimination and thus are able to \emph{decide} the provability of a sequent, i.e. find in finite time its provability, by showing only finitely many proofs need to be checked, see \cref{th:decidability}.
In any case, in applications one is usually concerned with bounding below the provability of a sequent, rather than establishing its precise value, and for that a single proof suffices.

Finally, we say that a rule \begin{prooftree}
	\hypo{\hyper{H}}
	\infer1[$R$]{\hyper{H}'}
\end{prooftree} is \textbf{admissible} in a calculus $\rules{R}$ when $\validity{\hyper{K}}_{\rules{R}} = \validity{\hyper{K}}_{\rules{R} \cup \{R\}}$ for any structure of sequents $\hyper{K}$.
Usually this means there is an effective procedure which, given a closed derivation of $\hyper{H}$ yields a closed derivation of $\hyper{H}'$ of equal or greater validity.
This is an important notion that arises variously in the rest of the paper, most prominently in \cref{th:cut-elimination}.


\subsubsection{Hypersequent calculi \& enriched calculi}
\label{sec:hyperseq}

Quantitative calculi are reminiscent of hypersequent calculi\footnote{Due to Pottinger \cite{pottinger1983uniform} and Avron \cite{avronConstructiveAnalysisRM1987}, see \cite{metcalfeProofTheoryFuzzy2009} for a modern account.}, in which sequents are also organized in a further level of structure.
Indeed, the fact that some fuzzy logics require more complex sequent structures has been observed many times~\cite{BaazCF03,CiabattoniLR21,CiabattoniG18}.
Let us show that our enriched setting coheres with (and generalises) that tradition.
Formally, a hypersequent is a (possibly empty) multiset of sequents, written
\begin{equation}
	\hyper{H} = \Gamma_1 \entails \Delta_1 \mid \cdots \mid \Gamma_n \entails \Delta_n
\end{equation}
Rules in a hypersequent calculus have a conjunctive structure of hypersequents as premises, where a distinguished sequent is accompanied by a \emph{hypercontext} (or \emph{side sequents}), the $\hyper{G}$s below:
\begin{equation}
	\begin{prooftree}
		\hypo{\hyper{G} \mid \Gamma \entails \Delta}
		\hypo{\cdots}
		\hypo{\hyper{G}' \mid \Gamma' \entails \Delta'}
		\infer3[$R$]{\hyper{G}'' \mid \Gamma'' \entails \Delta''}
	\end{prooftree}
\end{equation}
We now comment on their link.

The above definitions can be easily generalized to alethic values other than $\pMulReals[p]$. It suffices to fix such an object $\V$ and some algebraic structure on it, and then a notion of \textbf{$\V$-enriched calculus} ensues.
Specifically, one can pick the usual truth values $\Prop$ and consider various structures on it. In sequent calculus the tacit choice of $(\Prop, \top, \land)$ is made: premises of rules are conjunctive structures of sequents and the only closed structure is $\top$.

Hypersequent calculi can be seen as $\Prop$-enriched calculi, but where $\Prop$ is now also considered with $\lor$.
Then hypersequents become truly disjunctive expressions of sequents:
\begin{equation}
	\hyper{H} = \Gamma_1 \entails \Delta_1 \redbin{\lor} \cdots \redbin{\lor} \Gamma_n \entails \Delta_n
\end{equation}
and rules are expressions in conjunctive normal form.
In fact, most hypersequent calculi have rules whose premises all share the same hypercontext $\hyper{G}$, in which case they reduce to ones of the form
\begin{equation}
	\begin{prooftree}
		\hypo{\hyper{G}}
		\hypo{\red{\lor}}
		\hypo{\red{(}\Gamma' \entails \Delta'}
		\hypo{\redbin{\land} \cdots \redbin{\land}}
		\hypo{\Gamma' \entails \Delta' \red{)}}
		\infer5[$R$]{\hyper{G} \redbbin{\lor} \Gamma'' \entails \Delta''}
	\end{prooftree}
\end{equation}
When presented in this way the hypercontext can be dropped since, by our definition of derivation in a quantitative calculus, those can always be added back.
This makes the translation from a sequent calculus completely trivial.

In \cref{sec:pqll-and-mall}, we will use these considerations to compare quantitative calculi such as \pQLL with its hypersequent relatives.

\subsubsection{Deep Inference Calculi}
\label{sec:deep-inference}

Also related is \emph{Deep Inference}, or the \emph{Calculus of Structures} \cite{guglielmiSystemInteractionStructure2007}. These calculi erase the distinction between formula and structure; proofs are sequences of structures ending in the multiplicative unit, and rules may be applied at any depth. Cut-elimination for such calculi typically requires intricate proof rewriting, whereas our hybrid structure--sequent formulation of \pQLL admits a more direct proof. Still, recent semantic techniques for cut-elimination in Deep Inference \cite{AK24} may yield a calculus for \pQLL that unifies formulae and structures while keeping cut-elimination tractable.

Deep Inference calculi are rarely read quantitatively. An exception is Horne's \emph{Sub-additive Logic} \cite{horneSubAdditivesProofTheory2019}, which, similar to \pQLL, introduces $[0,1]$-annotated ``weakenings'' of the additives, inteded to represent probabilistic choice operators. These are \emph{sub}-additive in that the ``join-like'' one implies the actual join; in this sense, our soft additives are instead \emph{super}-additive. Horne's logic, however, does not (as far as we can see) admit the differentiable quantitative reading our neuro-symbolic training application requires.

\subsection[The pQLL Calculus]{The \pQLL Calculus}
\label{sec:pqll}
We are finally ready to introduce the quantitative calculi this work is about, parametric in a choice of $0 < p \leq \infty$.

\begin{notation}
	We deviate from traditional linear logic notation in denoting by $\top$ and $\bot$ the \emph{additive} units and with $\One$ the (coinciding) \emph{multiplicative} ones.
	We keep \emph{tensor} as $\tensor$ and \emph{par} as $\parr$ but use $\pland/\plor$ for the additive connectives.
\end{notation}

\begin{definition}\label{def:qll-formulae}
	Fix a set of \textbf{propositional variables} $\Atoms$.
	A \textbf{formula of \pQLL} with propositional variables in $\Atoms$ is defined inductively as follows:
	\begin{eqalign*}
		A \grammareq & a, a^\dual && \textit{((negated) propositional variable $a \in \Atoms$)} \\
		& \grammarsep \One \grammarsep A \tensor A \grammarsep A \parr A && \textit{(multiplicatives)}\\
		& \grammarsep  \bot \grammarsep \top  \grammarsep A \plor A \grammarsep A \pland A  && \textit{(additives)}
	\end{eqalign*}
\end{definition}

\noindent Negation $(-)^\dual$ is a syntactic transformation defined by
\begin{equation}\label{eq:qll-negation}
	\begin{gathered}
		(a)^\dual \equiv a^\dual,
		\quad
		(a^\dual)^\dual \equiv a,
		\quad
		\One^\dual \equiv \One,
		\quad
		\bot^\dual \equiv \top,
		\quad
		\top^\dual \equiv \bot,\\
		(A \tensor B)^\dual \equiv A^\dual \parr B^\dual,
		\quad
		(A \parr B)^\dual \equiv A^\dual \tensor B^\dual,
		\\
		(A \plor B)^\dual \equiv A^\dual \pland B^\dual,
		\quad
		(A \pland B)^\dual \equiv A^\dual \plor B^\dual,
	\end{gathered}
\end{equation}
We also employ the abbreviation $A \mulimp B \equiv A^\dual \parr B$.

\pQLL is the calculus presented in \cref{fig:pqll}.
For conceptual clarity it is laid out as a two-sided calculus, though there is a perfect duality which makes the left rules derivable from the right ones, and \emph{vice versa}.
This is due to the definition of negation we just gave and the rules \DualLeftRule/\DualRightRule.

\pQLL is conceived as an analogue of the usual (two-sided) sequent calculus of \MALL (see \cref{sec:pqll-and-mall}) `in the alethic metatheory of the $p$-reals'---giving a \textbf{soft} version when $p<\infty$, or a \textbf{hard} one at $p=\infty$.
Like in \MALL, the rules for connectives are split into multiplicative and additive fragments (\cref{subfig:add-frag}--\ref{subfig:mult-frag}).

The multiplicative fragment of \pQLL is essentially identical to that of \MALL, with the notable exception of being `isomix' \cite{cockettProofTheoryFull1997}, in the sense that the units for $\tensor$ (tensor) and $\parr$ (par) coincide---this is implicit by having left and right introduction rules for the same constant $\One$, i.e. \OneLeftRule and \OneRightRule.

The additive fragment is also essentially identical to the corresponding rules in hypersequent calculi for \ISOMIXMALL (see e.g. \cref{fig:isomixmall-h}), though being quantitative has its consequences.
Let's look at {\pOrRightRule}: to introduce a disjunction $A \plor B$ on the right of the sequent, either $A$ or $B$ do not suffice---rather we need to \emph{sum} a proof of $A$ and one of $B$ in the same context.
The result is a suggestive `superposition principle' for proofs involving additives, whereby we never make a choice of what we prove but we have obligations for all possible branches, which are then accumulated by sum (or dually, harmonic sum).

\begin{remark}
	Note that this `disjunctive branching' is optional in hypersequent \MALL (\cref{fig:isomixmall-h}), since the usual rules inherited from the sequent calculus (in which only one premise in \pAndLeftRule and \pOrRightRule is sufficient) are interadmissible with these ones---see \cite[Lemma~4.30]{metcalfeProofTheoryFuzzy2009}.
	This is also true for $p=\infty$ in this calculus (\cref{lemma:adm-of-MALL-additives}), but not in the soft versions.
\end{remark}

As for the structural rules (\cref{subfig:struct-frag}), these also merit some attention.
\CutRule and \AxRule are standard. We note that, while \AxRule is restricted to the introduction of unary cedents, by \MixRule (and \EmptyRule) we can easily extend it to arbitrary cedents. 
Note that \EmptyRule is derivable by cutting \OneLeftRule against \OneRightRule, though we keep it in the core calculus to avoid having to add it back in the cut-free calculus.

Finally, \ExFalsoRule (\emph{Ex Falso Quodlibet}) is also quite innovative, since usually adding rules producing `invalid' proofs to the calculus is not something considered desirable (note here `invalid' is meant in the technical sense of \cref{sec:defs} (a proof of validity $0$), not as logically unsound). 
Despite allowing to construct a proof of any sequent, proofs thus constructed have validity $0$, thus we remain perfectly consistent (as proven later in \cref{th:consistency}).

Still, \ExFalsoRule is far from useless!
To see why, consider a proof like
\begin{equation}\label[proof]{proof:inj}
	\begin{prooftree}
		\hypo{\red{1}}
		\infer1[\AxRule]{A \entails A}
		\hypo{\red{\padd}}
		\hypo{\red{0}}
		\infer1[\ExFalsoRule]{A \entails B}
		\infer3[\pOrLeftRule]{A \entails A \plor B}
	\end{prooftree}
\end{equation}
Even though the branch for $A \entails B$ does not contribute to the validity of the proof, the whole proof---which has positive validity---could not be written down if we did not have a proof for $A \entails B$\footnote{Note that $A$ and $B$ are metavariables, thus the actual validity of $A \entails A \plor B$ might differ. Still, there are instances of sequents $A \entails B$ which are not provable otherwise, e.g. $a \entails b$ for $a,b$ propositional variables, and thus the argument still applies.}.
Predictably, this argument falls through at $p=\infty$ and there \ExFalsoRule \emph{is} eliminable, see \cref{lemma:efq-elimination}.

In fact, \ExFalsoRule might be considered a catch-all axiom for external structural rules, as can be evidenced by the proof of the following:

\begin{lemma}[Admissibility of the structural schema]
	\label[lemma]{lemma:struct-schema}
	Consider the algebraic structure defined in \eqref{eq:hyperseq-grammar}, and suppose $\red{t(}x_1\red{,} \ldots\red{,} x_n\red{)}$ and $\red{s(}x_1\red{,} \ldots\red{,} x_n\red{)}$ be terms\footnotemark~of it.
	\footnotetext{Thus, rather than `structures of sequents', $\red{t}$ and $\red{s}$ are constructed out of a countable supply of fresh variables--rather than sequents---and the same red connectives considered in \eqref{eq:hyperseq-grammar}.}
	If $\validity{t} \leq \validity{s}$ pointwise when evaluated\footnotemark in $\pMulReals[p]$, then the rule
	\footnotetext{Here we mean that, for all $a_1, \ldots, a_n \in \PosReals$, we can prove $\validity{t}(a_1, \ldots, a_n) \leq \validity{s}(a_1, \ldots, a_n)$ in the ambient theory where the structure of the alethic values live.}
	\begin{equation}
		\label{rule:structural-schema}
		\tag{\ensuremath{\mathfrak{str}}}
		\begin{prooftree}
			\hypo{\red{t(} \Gamma_1 \entails \Delta_1 \red{,}\, \ldots \red{,}\, \Gamma_n \entails \Delta_n\red{)}}
			\infer1{\red{s(} \Gamma_1 \entails \Delta_1 \red{,}\, \ldots \red{,}\, \Gamma_n \entails \Delta_n \red{)}}
		\end{prooftree}
	\end{equation}
	is admissible in any $p$-hard counary calculus admitting \ExFalsoRule.
\end{lemma}
\begin{proof}
	Consider a closed derivation $\delta$ ending in \StructuralSchemaRule.
	Since the calculus is counary (all rules end in a sequent), the derivation above \StructuralSchemaRule will be given by a combination of proofs of each of its premises, which in turn correspond to a variable occurrence in $t$.
	Thus, for each variable $x_i$ in the context of $t$, we define $\pi_i$ to be the proof of maximal validity amongst those that insist (in $\delta$) on an occurrence of $\Gamma_i \entails \Delta_i$ in the premise of \StructuralSchemaRule---note that $x \star y \leq (x \lor y) \star (x \lor y)$ for every connective $\star \in \{\tensor, \cotensor, \padd[p], \pcoadd[p]\}$.
	If $x_i$ does not appear in $t$, then $\pi_i = \ExFalsoRule$.
	Excising the derivation above \StructuralSchemaRule and replacing it with the derivation $\red{s(}\pi_1\red{,} \ldots\red{,} \pi_n\red{)}$, we obtain a new derivation with validity $s(\validity{\pi_1}, \ldots, \validity{\pi_n}) \geq t(\validity{\pi_1}, \ldots, \validity{\pi_n}) \geq \validity{\delta}$, as desired.
\end{proof}

This result effectively grants us the use of any of the properties of $\tensor$, $\cotensor$, $\padd$, and $\pcoadd$ in $\pMulReals[p]$ at the level of structures of sequents, since any such `generalized' derivation is in fact rewritable to a proper derivation, as illustrated by the proof.

\cref{lemma:struct-schema}'s significance is in granting \pQLL (and other calculi to which it applies) the same ergonomics that traditional sequent calculi enjoy, in which we do not worry about discharging premises in a particular order.
For instance, we can show that the following `external structural rules' familiar from hypersequent calculi (see \cref{fig:isomixmall-h} but also \cite[Chapter~4]{metcalfeProofTheoryFuzzy2009}) are admissible:
\begin{equation}
\label{rule:external-structural}
	\begin{prooftree}
		\hypo{\Gamma \entails \Delta}
		\infer1[\ExtWeakRule]{\Gamma \entails \Delta \redbbin{\padd[p]} \Gamma' \entails \Delta'}
	\end{prooftree}
	\qquad
	\begin{prooftree}
		\hypo{\Gamma \entails \Delta \redbbin{\pcoadd[p]} \Gamma' \entails \Delta'}
		\infer1[\ExtCoweakRule]{\Gamma \entails \Delta}
	\end{prooftree}
	\qquad
	\begin{gathered}
		\begin{prooftree}
			\hypo{\Gamma \entails \Delta &\redbin{\bullet}\ \, \Gamma' \entails \Delta'}
			\infer1[$\ExtExchRule[\red{\bullet}]$]{\Gamma' \entails \Delta' &\redbin{\bullet}\ \, \Gamma \entails \Delta}
		\end{prooftree}
		\\
		\text{\footnotesize{for $\red{\bullet} \in \{\red{\tensor}, \red{\cotensor}, \red{\padd[p]}, \red{\pcoadd[p]}\}$}}
	\end{gathered}
\end{equation}

\begin{example}
\label{ex:ew}
	To illustrate the proof of \cref{lemma:struct-schema}, consider \ExtWeakRule, which corresponds to taking $\red{t}(x_1) = x_1$ and $\red{s}(x_1, x_2) = x_1 \redbin{\padd[p]} x_2$: indeed $a \leq a \padd[p] b$ for all $a,b \in \PosReals$.
	A closed derivation $\delta = \dfrac \pi \ExtWeakRule$ is depicted below left.
	Following the recipe in the proof, we set $\pi_1 := \pi$ (the unique premise insisting on $x_1$), and $\pi_2 := \ExFalsoRule$ since $x_2$ does not occur in $\red{t}$.
	Excising the application of \ExtWeakRule and replacing it with $\red{s}(\pi_1, \pi_2) = \pi_1 \redbin{\padd[p]} \pi_2$, we obtain the rewritten derivation below right whose validity is $\validity{\pi} \padd[p] 0 = \validity{\pi} = \validity{\delta}$, as desired.
	\begin{equation}
		\begin{prooftree}
			\hypo{}
			\ellipsis{$\pi$}{\Gamma \entails \Delta}
			\infer1[\ExtWeakRule]{\Gamma \entails \Delta \redbbin{\padd[p]} \Gamma' \entails \Delta'}
		\end{prooftree}
		\qquad \mapsto \qquad
		\begin{prooftree}
			\hypo{}
			\ellipsis{$\pi$}{\Gamma \entails \Delta}
			\hypo{\red{\padd[p]}}
			\hypo{\red{0}}
			\infer1[\ExFalsoRule]{\Gamma' \entails \Delta'}
			\infer3[$\id$]{\Gamma \entails \Delta \!\redbbin{\padd[p]}\! \Gamma' \entails \Delta'}
		\end{prooftree}
	\end{equation}
\end{example}


\section{Proof-theoretic and computational properties}
\label{sec:proof-theory}
In this section, we discuss the common results traditionally expected from a sequent system, and explain the key technical challenges that arise in our journey towards a combination of \emph{linearity} and \emph{quantitativity} of the truth values.

\subsection{Cut-elimination and its consequences}
Cut-elimination is a desirable property of sequent calculi as it affords the subformula property, the consistency, and even, in some cases, completeness~\cite{Buss98} of the logic in question.
Because of this, any new logic introduced in the last century is subject to the test, and so is ours:

\begin{theorem}[Cut-elimination]
	\label{th:cut-elimination}
	Let \CutFreepQLL be the calculus obtained by removing \CutRule from \pQLL.
	Then
	\(
	\validity{\Gamma \entails \Delta}_{\pQLL} = \validity{\Gamma \entails \Delta}_{\CutFreepQLL}.
	\)
	In other words, \CutFreepQLL admits \CutRule. Moreover, the process of eliminating CUT from a \pQLL proof is effective.
\end{theorem}
\begin{proof}
	Since $\CutFreepQLL \subseteq \pQLL$ (as rule-sets), then clearly provability for the latter cannot be smaller than provability for the former (there are more proofs in \pQLL, which means the supremum of \eqref{eq:provability} runs over a larger set).
	Conversely, in \cref{app:cut-elim}, we describe a proof transformation from \pQLL to \CutFreepQLL that is non-decreasing in validity.
	Our proof method follows Melliès' \cite{mellies:categorical-sem-LL} and the Linear Logic Handbook \cite{LLHandbook}.%
	\footnote{We followed the latter in our reconstruction of the termination argument.}
	Here we sketch the rough idea: given a proof
	\vspace*{-1.2ex}
	\begin{equation}
		\begin{prooftree}
			\hypo{}
			\ellipsis{$\pi$}{\Gamma \entails A, \Delta \redbbin{\tensor} \Gamma', A \entails \Delta'}
			\infer1[\CutRule]{\Gamma, \Gamma' \entails \Delta, \Delta'}
		\end{prooftree}
	\end{equation}
	we rewrite it to a proof $\pi'$ with the same conclusion $\Gamma, \Gamma' \entails \Delta, \Delta'$, such that either of the three conditions on applications of \CutRule in  $\pi'$ holds: their number decreases, or they appear at a lower depth in the proof, or the cut formulae pertain to a structurally smaller sequent.
	We also show that in this case $\validity{\pi} \leq \validity{\pi'}$.

	This defines a deterministic, terminating, validity non-decreasing, proof transformation, as required.
\end{proof}

This cut-elimination result offers support for our general agenda---to show that the proposed quantitative upgrade to the sequent calculus is (1) non-trivial and (2) natural from the point of view of preceding proof-theoretic tradition.

For (1), note that our cut-elimination result relies on the notion of proof validity understood in a quantitative way,
and directly reasons about quantitative validity of proofs that are subject to cut-elimination.
To the best of our knowledge, this is the first result of this kind, as all prior cut-elimination proofs for fuzzy logic (given in e.g.~\cite{metcalfeProofTheoryFuzzy2009}), while relying on quantitative formula validity,  still had a Boolean validity of sequents and proofs at the alethic level.

For (2), notice that, despite  a seemingly radical revision of the prior tradition, the structure of the inductive argument and the termination proof largely follows a classical scheme.

We observe, moreover, that in most cases that do not involve additives ($\plor$ and $\pland$) we can in fact establish that $\validity{\pi} = \validity{\pi'}$.
In the other cases, proof validity actually increases, showing that the quantitative upgrade did not break the key structural principles that any cut-elimination result relies on.
In fact, this result generalises cut elimination for MALL, where only one proof validity is possible: $\validity{\pi} = \validity{\pi'} = 1$.

As an illustration of such matters, consider the following exemplar case from our cut-elimination proof.
We are aiming to show that the proof that introduces $A \plor B$ as a cut formula:
\begin{equation}
	\begin{prooftree}
		\hypo{}
		\ellipsis{$\pi_1$}{\Gamma \entails A, \Delta}
		\hypo{\red{\padd[p]}}
		\hypo{}
		\ellipsis{$\pi_2$}{\Gamma \entails B, \Delta}
		\infer3[\pOrRightRule]{\Gamma \entails A \plor B, \Delta}
		\hypo{\redbin{\tensor}}
		\hypo{}
		\ellipsis{$\pi_3$}{\Gamma', A \entails \Delta'}
		\hypo{\redbin{\pcoadd[p]}}
		\hypo{}
		\ellipsis{$\pi_4$}{\Gamma', B \entails \Delta'}
		\infer3[\pOrLeftRule]{\Gamma', A \plor B \entails \Delta'}
		\infer3[\CutRule]{\Gamma,\Gamma' \entails \Delta,\Delta'}
	\end{prooftree}
\end{equation}
can be transformed into a proof with the same root sequent and an application of \CutRule on a structurally smaller formula.
We see that, from the above, we can obtain either of these two proofs:
\begin{equation}
	\begin{prooftree}
		\hypo{}
		\ellipsis{$\pi_1$}{\Gamma \entails A, \Delta}
		\hypo{\redbin{\tensor}}
		\hypo{}
		\ellipsis{$\pi_3$}{\Gamma', A \entails \Delta'}
		\infer3[\CutRule]{\Gamma,\Gamma' \entails \Delta,\Delta'}
	\end{prooftree}
	\hspace*{5ex}
	\begin{prooftree}
		\hypo{}
		\ellipsis{$\pi_2$}{\Gamma \entails B, \Delta}
		\hypo{\redbin{\tensor}}
		\hypo{}
		\ellipsis{$\pi_4$}{\Gamma', B \entails \Delta'}
		\infer3[\CutRule]{\Gamma,\Gamma' \entails \Delta,\Delta'}
	\end{prooftree}
\end{equation}
Following the standard inductive argument, we obtain that in either of these two cases, the \CutRule is applied on a structurally smaller formula, that is $A$ or $B$, instead of $A \plor B$. This is analogous to the additive disjunction case in Melliès in~\cite{mellies:categorical-sem-LL}.
However, we are not done---we also need to show that we can guarantee that the proof transformation increases proof validity!
To proceed, we need to prove that
\begin{equation}
	(\validity{\pi_1} \padd[p] \validity{\pi_2}) \tensor (\validity{\pi_3} \pcoadd[p] \validity{\pi_4}) \leq (\validity{\pi_1} \tensor \validity{\pi_3}) \lor (\validity{\pi_2} \tensor \validity{\pi_4}),
\end{equation}
which holds in the reals (\cref{lemma:max-prod-ineq}).
Thus we can, in a deterministic way, pick the option of maximal validity~(and, by convention, the one involving $A$ if both validities are equal)\footnote{
	Note that we can't `keep both proofs' by putting them together using \ExtContRule since $(a \padd[p] b) \tensor (c \pcoadd[p] d) \nleq 2^{-1/p} \tensor ((a \tensor c) \padd[p] (b \tensor d))$.
	Still, the proof transformation that makes \ExtContRule admissible (described in \cref{lemma:struct-schema}) acts exactly by choosing the maximal validity option, as we do here.
}.

Thus we can guarantee that our proof transformation will construct a proof that has equal or greater validity than the one we started from.
This completes our argument.

Cut-elimination still grants the usual consequences, namely, the subformula property and consistency.

\begin{theorem}[Consistency]
	\label{th:consistency}
	For every $0 < p \leq \infty$, \pQLL[p] is consistent, i.e. $\validity{\ \entails \bot} = 0$.
\end{theorem}
\begin{proof}
	In $\CutFreepQLL$, $\entails \bot$ can only possibly be proved via \ExFalsoRule, which has validity $0$.
\end{proof}

\begin{theorem}[Decidability]
	\label{th:decidability}
	For any \pQLL sequent ${\Gamma \entails \Delta}$, the supremum defining \ref{eq:provability} is in fact a maximum, i.e. there is a proof $\pi$ that realizes its provability.
\end{theorem}
\begin{proof}
	This is essentially a direct consequence of subformula property.
	We just need to argue that we have finitely many proofs to search, though in principle there are infinitely many cut-free ones.
	For space reasons, we moved this short argument to \cref{app:decidability}.
\end{proof}



The main take-away message of this section is that, despite the addition of the quantitative proof validity, the calculus still retains the constructive nature of the sequent calculus. As further evidence of this fact, Bailitis et al.~\cite{Bailitis2026} have recently formalised \QLL in Rocq, following closely the formalisation by Laurent~\cite{Laurent26}, and adding the validity calculation via an ad-hoc recursive function, rather than as a part of the inductive relation that defines the calculus rules. As a result, the cut-elimination proof by \citeauthor{Bailitis2026} follows a structurally identical argument to that of Laurent, modulo an additional argument concerning the increase of validity in the cut-free proofs.

\subsection{\texorpdfstring{Hard and qualitative \pQLL}{Hard and qualitative pQLL}}
\label{sec:pqll-and-mall}
\pQLL is related to \MALL in two different ways.
The first is in a \emph{structural} way: `qualitatively', \pQLL is equivalent to (the isomix\footnotemark~variant of) the hypersequent calculus for MALL, which we call \hISOMIXMALL.
\footnotetext{i.e. \MALL + \MALLMixRule + \MALLEmptyRule (which in particular makes the units for tensor and par coincide), terminology from \cite{cockettProofTheoryFull1997}.}
The second way is `in the limit': as mentioned in \cref{sec:pos-reals}, at $p=\infty$ the additive operations in $\pMulReals[p]$ collapse to their `hard' variants, i.e. max and min.
The corresponding logic, \pQLL[\infty] (``\emph{hard} \QLL''), is thus even closer to \ISOMIXMALL, due to the fact that now the rules for additives are interadmissible with the classical ones.

To unburden proofs of inessential bureaucracy, we adopt the same language for formulae of both calculi.
We thus liberally confuse the sequents of both calculi.\footnote{The reader may find it helpful to consult \cref{fig:isomixmall,fig:isomixmall-q,fig:isomixmall-h} in \cref{app:isomix-mall}, laying out a sequent calculus and a hypersequent calculus for \ISOMIXMALL (\hISOMIXMALL), as well as the qualitative version of \pQLL.}

\subsubsection{Qualitative \pQLL}
By `qualitative \pQLL' we mean the calculus $\pQLL^{>0}$ obtained by replacing every $\red{\tensor}$, $\red{\pcoadd}$ appearing in its rules with $\red{\land}$ and every $\red{\padd}$ with $\red{\lor}$.

\begin{theorem}\label{th:isomixmall-cons-and-adeq-over-qual-pqll}
	Let ${\Gamma \entails \Delta}$ be a sequent in qualitative \pQLL, then
	\begin{equation}
		\text{$\Gamma \entails \Delta$ provable in \hISOMIXMALL}\ \Leftrightarrow\ \validity{\Gamma \entails \Delta}_{\pQLL[p]} > 0.
	\end{equation}
\end{theorem}
\begin{proof}
	Follows from the considerations made in \cref{sec:hyperseq}: the derivations of the two calculi essentially mirror each other, up to the management of hypercontext.
	\ExFalsoRule and \cref{lemma:struct-schema} mutually translate with \hMALLExtWeakRule (as seen in \cref{ex:ew}) and \hMALLExtContRule.
	Details can be found in \cref{app:isomix-mall}.
\end{proof}

\subsubsection{Hard \QLL}
We now turn to \pQLL[\infty]: this calculus is still quantitative, unlike $\pQLL[p]^{> 0}$, but validity of proofs will always be either $0$, $1$, or $\infty$, unless axioms of a different validity are introduced.
Indeed, the additives at $p=\infty$ are `hard' and thus expected to behave as usual.

\begin{lemma}\label{lemma:adm-of-MALL-additives}
	The rules below are admissible in \pQLL[\infty] and, \emph{vice versa}, adding these to $\Free{\pQLL[\infty]}{\pOrRightRule[\infty]}$ make \pOrRightRule[\infty] admissible again.
	\begin{eqalign*}
		\label{eq:MALL-additives}
		\begin{prooftree}
			\hypo{\Gamma \entails A, \Delta}
			\infer1[${\pOrRightRule[\infty]}_1$]{\Gamma \entails A \plor B, \Delta}
		\end{prooftree}
		\hspace*{8ex}
		\begin{prooftree}
			\hypo{\Gamma \entails B, \Delta}
			\infer1[${\pOrRightRule[\infty]}_2$]{\Gamma \entails A \plor B, \Delta}
		\end{prooftree}
	\end{eqalign*}
\end{lemma}
\begin{proof}
	Note that one can derive these two rules from $\pOrRightRule[p]$ even for finite $p$, by using \ExtWeakRule to introduce the neglected branch.
	Conversely, suppose given ${\pOrRightRule[\infty]}_1$ and ${\pOrRightRule[\infty]}_2$, now a derivation involving
	\begin{equation}
		\begin{prooftree}
			\hypo{}
			\ellipsis{$\pi_A$}{\Gamma \entails A, \Delta}
			\hypo{\red{\padd[\infty]}}
			\hypo{}
			\ellipsis{$\pi_B$}{\Gamma \entails B, \Delta}
			\infer3[\pOrRightRule[\infty]]{\Gamma \entails A \plor B, \Delta}
		\end{prooftree}
	\end{equation}
	will have validity at least $\validity{\pi_A} \lor \validity{\pi_B}$, which is the same lower bound we get from the unary rules.
	Note this reasoning only works for $p=\infty$.
\end{proof}

The above obviously dualize to the case of additive conjunction.

We can also show that:

\begin{lemma}[\ExFalsoRule-elimination]
	\label{lemma:efq-elimination}
	Let \EFQFreepQLL[\infty] be the calculus obtained by removing \ExFalsoRule from \pQLL[\infty].
	Then
	\begin{equation}
		\validity{\Gamma \entails \Delta}_{\pQLL[\infty]} = \validity{\Gamma \entails \Delta}_{\EFQFreepQLL[\infty]},
	\end{equation}
	in other words, \EFQFreepQLL[\infty] admits \ExFalsoRule.
\end{lemma}
\begin{proof}
	Let $\pQLL[\infty]'$ be the calculus obtained by replacing \pOrRightRule[\infty] and \pAndLeftRule[\infty] with their unary versions.
	As proven in \cref{lemma:adm-of-MALL-additives},
	$\validity{\Gamma \entails \Delta}_{\pQLL[\infty]} = \validity{\Gamma \entails \Delta}_{\pQLL[\infty]'}$,
	As discussed previously, removing \ExFalsoRule is only problematic in the presence of rules involving additive disjunctive structures as their premise.
	However, $\pQLL[\infty]'$ has no such rules.
\end{proof}

So \pQLL[\infty] is even more similar to \ISOMIXMALL than the soft calculi; with a stronger conservativity theorem:


\begin{theorem}\label{th:isomixmall-cons-and-adeq-over-infty-qll}
	Let $\Gamma \entails \Delta$ be a sequent in \ISOMIXMALL or \pQLL[\infty], then
	\begin{equation}
		\text{$\Gamma \entails \Delta$ provable in \ISOMIXMALL}\ \Leftrightarrow\ \validity{\Gamma \entails \Delta}_{\pQLL[\infty]} \geq 1.
	\end{equation}
\end{theorem}
\begin{proof}
	Fairly routine: both derivations correspond one-to-one, up to \ExFalsoRule (which is addressed by \cref{lemma:efq-elimination}) and the additives (addressed by \cref{lemma:adm-of-MALL-additives}).
	See for a detailed argument \cref{app:isomix-mall}.
\end{proof}

We could summarize the section as follows: in the `additive' mode of truth (i.e. being strictly positive), \pQLL is equivalent to \ISOMIXMALL, while in the `multiplicative' mode (i.e. being equal or greater than $1$) only $\pMulReals[\infty]$ (the Lawvere quantale) recovers \ISOMIXMALL.

Joining \cref{th:isomixmall-cons-and-adeq-over-qual-pqll} and \cref{th:isomixmall-cons-and-adeq-over-infty-qll} with the well-known fact that \hISOMIXMALL is equivalent to \ISOMIXMALL, we have that if $\Gamma \entails \Delta$ is a sequent in either \ISOMIXMALL or \pQLL[p] or \pQLL[\infty], then this sequent
$
	\text{$\Gamma \entails \Delta$ provable in \ISOMIXMALL}
$
 iff %
$
	\validity{\Gamma \entails \Delta}_{\pQLL[\infty]} \geq 1
$
iff %
$
	\validity{\Gamma \entails \Delta}_{\pQLL[p]} > 0.
$

\subsubsection{Approximation}
A major motivation for the choice of $\padd/\pcoadd$ as semantics for soft additives is that the parameter $p$ can be used to tune the detachment from hard additives (\cref{lemma:additive-collapse}).
Here we reap the logical consequences of this fact.

Write as $(-)^p$ the obvious translation of proofs from \pQLL[1] (or any \pQLL[q], $q < \infty$) to \pQLL[p].

\begin{theorem}[Approximation Theorem]
	\label{th:approximation}
	Let ${\Gamma \entails \Delta}$ be a \pQLL sequent, then
	\begin{equation}
		\validity{\Gamma \entails \Delta}_{\pQLL[p]} \conv[p \to \infty] \validity{\Gamma \entails \Delta}_{\pQLL[\infty]}.
	\end{equation}
\end{theorem}
\begin{proof}
	Let $\pi$ be a proof of $\Gamma \entails \Delta$ in \pQLL[1], and consider $\validity{\pi^p}$ as a function of $p$.
	Since it is an expression involving $\padd$, $\pcoadd$, and $\tensor$ only, it is continuous on the interior of $\PosReals$.
	Thus assume the top structure of sequents in $\pi$, call it $\hyper{H}$, does not involve either $0$ or $\infty$---by continuity and \cref{lemma:additive-collapse}, we get the desired claim.
	Now suppose instead $\hyper{H}$ involves just $0$ and $\infty$, then $\validity{\pi^p}$ is constant and hence the claim is met.
	The remaining cases trivially reduce to these two.
\end{proof}


\subsubsection{A Resource Interpretation}
\label{sec:interpretation}
Finally, we want to hint at an intuitive `interpretation' of \QLL.
Linear Logic has a \emph{resource interpretation} whereby a sequent ${\Gamma \entails A}$ asserts the convertibility of the resources in $\Gamma$ to the resource $A$.
Quantitative Linear Logic, especially \pQLL[1], refines this interpretation by quantifying convertibility: the provability of a sequent $\Gamma \entails A$ represents a \emph{conversion rate} from the resources in $\Gamma$ to $A$.
For example, $\validity{1\EUR \entails \text{Pizza}} = 5$ means a pizza costs 5\EUR. The \pQLL calculus rules can then be interpreted as a calculus of conversion rates, e.g. if $\validity{1\EUR \entails \text{Soda}} = 4$ then, by \pOrRightRule[1], $\validity{1\EUR \entails \text{Pizza} \plor \text{Soda}} = 9$.
Clearly, if the conversion rate from $A$ to $B$ is $> 0$, then we are right to say $A$ is convertible to $B$---this is \cref{th:isomixmall-cons-and-adeq-over-qual-pqll}.

\section{Softales}
\label{sec:softale}
Following \cite{lawvereMetricSpacesGeneralized1973}, the algebraic/categorical semantics of \pQLL takes place in a kind of enriched preorder we call a \emph{softale}.
These are $\MulReals$-enriched preorders equipped with a monoidal structure, a suitable duality, and having all `soft $p$-joins'.
In order to explain the definition, we start by recalling some standard notions of enriched category theory \cite{kellyBasicConceptsEnriched1982}.

\begin{definition}\label{def:v-preorder}
	A \textbf{$\MulReals$-preorder} $\cat{A}=(A, \eleq)$ is a set $A$ equipped with a function ${\eleq} : A \times A \to \PosReals$ satisfying for all $a,b,c \in A$,
	\begin{gather}
		\label{eq:enriched-refl}
		\tag{reflexivity}
		1 \leq (a \eleq a),\\
		\label{eq:enriched-trans}
		\tag{transitivity}
		(a \eleq b) \otimes (b \eleq c) \leq (a \eleq c).
	\end{gather}
	A \textbf{map} or \textbf{functor of $\MulReals$-preorders} $F:\cat{A} \to \cat{B}$ is a function between their underlying sets such that for all $a,a'\in A$,
	\begin{equation}
		(a \eleq a') \leq (Fa \eleq Fa').
	\end{equation}
\end{definition}

We denote all enriched order relations as $\eleq$, since their arguments suffice to disambiguate.
%
%
Every $\MulReals$-preorder has two associated \textbf{underlying preorders}: the \textbf{multiplicative} one $\cat{A}_{\geq 1} = (A, \eleq_{\geq 1})$ and the \textbf{additive} one, defined as
\begin{equation}
	a \eleq_{\geq 1} b \iff 1 \leq (a \eleq b),
	\quad
	a \eleq_{> 0} b \iff 0 < (a \eleq b).
\end{equation}
The first is the one enriched category theory defaults to.
We say ${a,b \in \cat{A}}$ are \textbf{multiplicatively} or \textbf{additively isomorphic} and write $a \iso_{\geq 1} b$ or $a \iso_{> 0} b$ when they are equivalent in the respective underlying preorder.

\begin{definition}[Softale]
	\label{def:p-softale}
	A \textbf{($p$-)softale} is a $\MulReals$-enriched preorder $\cat S$ equipped with
	\begin{enumerate}
		\item a \textbf{multiplicative structure} consisting of a \textbf{unit} element $\One \in S$ and a \textbf{tensor} map $\monop : S \times S \to S$ satisfying
		\begin{gather}
			\label{eq:softale-interchange}
			(a \eleq b) \tensor (c \eleq d) \leq (a \monop c) \eleq (b \monop d),
			\\
			\label{eq:softale-unit-ass-comm}
			a \monop \One \iso_{\geq 1} a,
			\quad
			(a \monop b) \monop c \iso_{\geq 1} a \monop(b \monop c),
			\quad
			a \monop b \iso_{\geq 1} b \monop a,
		\end{gather}
		and a \textbf{duality} $(-)^\dual : S \to S$ satisfying
		\begin{gather}
			\label{eq:softale-involutivity}
			\One^\dual = \One, \quad a^{\dual\dual} = a,
			\\
			\label{eq:softale-duality}
			(a \eleq b) = (b^\dual \eleq a^\dual),
			\\
			\label{eq:softale-star-aut}
			(a \monop b \eleq c^\dual) = \left(a \eleq (b \monop c)^\dual\right);
		\end{gather}
		\item and an \textbf{additive structure}, consisting of a \textbf{bottom} element $\bot \in S$, and a choice of \textbf{$p$-join} $a \plor b$ for all $a,b \in \cat S$, satisfying
		\begin{align}
			\label{eq:softale-soft-bottom}
			(\bot \eleq a)     & = \infty,
			\\
			\label{eq:softale-p-joins-minimality}
			(a \eleq c) \pcoadd[p] (b \eleq c) & \leq (a \plor b \eleq c),
			\\
			\label{eq:softale-p-joins-upper-bound}
			(c \eleq a) \padd[p] (c \eleq b)   & \leq (c \eleq a \plor b),
		\end{align}
	\end{enumerate}

	A \textbf{(strict) map of softales} ${F:\cat{S} \to \cat{T}}$ is a map of $\MulReals$-preorders such that
	\begin{gather}
		\One = F\One, \quad {Fa \monop Fb = F(a \monop b)}, \quad
		F(a^\dual) = (Fa)^\dual,
		\quad
		\bot = F\bot, \quad Fa \plor Fb = F(a \plor b).
	\end{gather}
	Softales and their maps form a category named $\Softales$.
\end{definition}


%

\begin{remark}
	The multiplicative structure on a softale, by itself, could be classified as a \textbf{$*$-autonomous $\MulReals$-preorder}, and indeed induces a {$*$-autonomous} structure on the underlying ($\Set$-)preorder.
	Just as in the $\Set$-enriched case (see \cite{barrAutonomousCategories1979}), we have an induced monoidal product \textbf{par}:
	$a \dumonop b := (a^\dual \monop b^\dual)^\dual$.
	Moreover, from \eqref{eq:softale-star-aut} and then \eqref{eq:softale-involutivity} we can prove
	\begin{equation}\label{eq:softale-mix}
		1 \leq (a \monop (b \dumonop c) \eleq (a \monop b) \dumonop c),
		\quad
		1 \leq (a \monop b \eleq a \dumonop b)
	\end{equation}
	where the last is the so-called `mix' property \cite{cockettProofTheoryFull1997}.
	This makes it \emph{closed} too, in the sense that $b \to c := b^\dual \monop c$ is a residuation in the enriched sense, that is
	\begin{equation}
		(a \monop b \eleq c) = (a \eleq b \to c).
	\end{equation}
\end{remark}

\begin{remark}
	 To define $p$-joins, we take inspiration from the calculus and from an observation on the `categorical' definition of joins in the hard case: usually, the definition by representability
	\begin{equation}\label{eq:hard-joins-representability}
		\tag{representability}
		(a \eleq c) \land (b \eleq c) \iff (a \lor b) \eleq c
	\end{equation}
	is equivalent to the one given as a universal property, which can be split in two statements:
	\begin{align}\label{eq:hard-joins-minimality}
		\tag{minimality}
		(a \eleq c) \land (b \eleq c) & \implies (a \lor b) \eleq c,
		\\
		\label{eq:hard-joins-upper-bound}
		\tag{majorization}
		(c \eleq a) \lor (c \eleq b)  & \implies c \eleq (a \lor b).
	\end{align}
	This equivalence, however, breaks in the quantitative setting, since to prove \ref{eq:hard-joins-upper-bound} from \ref{eq:hard-joins-representability} one crucially needs $\land$ to distribute over $\lor$, and that does not apply to their soft versions\footnote{Specifically, the needed direction is $(x \land y) \lor (x \land z) \leq x \land (y \lor z)$, whereas in $\pMulReals[p]$ it is the other direction that holds.}.
	Likewise, in the converse direction, one needs the idempotence of $\padd[\infty]$ to conclude \ref{eq:hard-joins-representability} given \ref{eq:hard-joins-upper-bound} and \ref{eq:hard-joins-minimality}, which thus fail for $p<\infty$.
\end{remark}

Just as we defined $\dumonop$, one can define \textbf{$p$-meets} for a softale by setting $a \pland b \equiv (a^\dual \plor b^\dual)^\dual$.
Clearly they satisfy dual properties to \eqref{eq:softale-p-joins-minimality}--\eqref{eq:softale-p-joins-upper-bound}.

When \eqref{eq:softale-p-joins-minimality} and \eqref{eq:softale-p-joins-upper-bound} hold in the converse direction too we say the additives are \textbf{representable}, and note they make the multiplicative underlying preorder a lattice in the traditional sense.



\begin{example}[Multiplicative reals]
	The structure on $\PosReals$ we sketched in \cref{sec:pos-reals} is that of a softale, where ${\eleq} := {\mulimp}$, $\monop := \tensor$, $(-)^\dual := (-)^*$, $\dumonop \equiv \cotensor$, $\plor := \padd$, and $\pland \equiv \pcoadd$.
	This softale has representable additives, as seen in \cref{cor:tensor-over-add-dist,cor:cotensor-over-add-conormal-dist}.
\end{example}


\begin{example}[Varying hardness]
	Observe that 
	every $\pMulReals[q]$ can become a $1$-softale $\pMulReals[q]^{1/q}$ whose preorder is $a \mulimp^{1/q} b := (a \mulimp b)^{1/q}$, including when $q=\infty$. 
	However, we have $\pMulReals[1] \iso \pMulReals[q]^{1/q}$ as $1$-softales only when $q < \infty$.
\end{example}


\begin{example}[Involutive commutative residuated lattices]
	Every commutative residuated lattice $(R, \preceq, 1, \odot, \lor)$ such that $a^\dual \equiv a \to 1$ is involutive and satisfies \eqref{eq:softale-duality}, is a $p$-softale, by setting $a \eleq b = \infty$ if $a \preceq b$ and $0$ otherwise.
\end{example}

Conversely, we have an analogue of the results we have seen in \cref{sec:pqll-and-mall}:

\begin{proposition}\label{prop:underlying-softale-is-res-lat}
	The additive underlying preorder of a $p$-softale and the multiplicative underlying preorder of an $\infty$-softale are always involutive symmetric residuated lattices.\footnote{We mean the preorder version of these concepts, so that e.g. symmetry is commutativity up to unique isomorphism. For a reference on residuated lattices, see \cite[Chapter~IV]{cintulaHandbookMathematicalFuzzy2011}.}
\end{proposition}
\begin{proof}
	By the above considerations, $\cat S_{\geq 1}$ would clearly be symmetric residuated monoidal (and the same applies to $\cat S_{> 0}$, since laws holding at $1$ hold positively too), as well as have joins in the traditional sense.
	The duality $(-)^\dual$ coincides with $a \to 1$, making this involutive too.
\end{proof}

Recall (e.g. from \cite{kellyBasicConceptsEnriched1982}) that for $\cat{A},\cat{B}$ $\MulReals$-enriched, the \textbf{preorder of maps} $\cat{B}^\cat{A}$ is again $\MulReals$-enriched, with order
\begin{equation}\label{eq:preorder-of-maps}
	F \twoto G \equiv \biginf_{a \in A} (Fa \eleq Ga).
\end{equation}

\begin{example}\label{ex:power-preorder}
	Let $\cat{A}$ be an enriched order, consider the order of maps $\PosReals^{\cat{A}}$.
	This is always a $p$-softale with all the structure defined pointwise.
	The key observation is that, for any binary operation $\star$ on $\PosReals$, `$\biginf \star \biginf \leq \biginf \star$' since the latter corresponds to taking the infimum over a smaller set.

\end{example}

Interestingly, $\PosReals^{\cat{A}}$ does not admit a softale structure given by convolution\footnote{In the sense of Day, see \cite{dayClosedCategoriesFunctors1970}.}: in fact $\One \eleq -$ is not a valid unit since it does not coincide with its dual, namely $(- \eleq \One)^*$.
Note also that $\PosReals^{\cat{A}}$ does not have representable additives because infima do not commute with (harmonic) $p$-sums.


\subsection{Model theory}
\label{sec:model-theory}
We prove soundness in \pQLL by showing that the axioms of softales make each rule of \pQLL into a valid reasoning scheme within a softale.
We then show completeness by constructing \emph{classifying softales}, that is, softales of \pQLL formulae and whose enriched order is given by \ref{eq:provability} of sequents.
Finally, we discuss a version of completeness for $\PosReals$, limited to a certain kind of theories which arise in applications to neural network verification.

\subsubsection{Soundness and completeness}
\label{sec:soundness-and-completeness}
Fix a $p$-softale $\cat{S}$.
A \emph{cedent in $\cat{S}$} is a multiset $\Gamma = \{A_1, \ldots, A_n\}$ of elements of $\cat{S}$.
We write $\Gamma^\monop$ for the element $A_1 \monop \cdots \monop A_n$, and $\Gamma^\dumonop$ for $A_1 \dumonop \cdots \dumonop A_n$.
The empty cedent reduces to $\One$ in both ways.
A \emph{sequent in $\cat{S}$} is an inequality expression of the form $\Gamma^\monop \eleq \Delta^\dumonop$ for $\Gamma$ and $\Delta$ cedents in $\cat{S}$.
A \emph{structure of sequents in $\cat{S}$} is formed as per the grammar in \eqref{eq:hyperseq-grammar}, by replacing sequents with sequents in $\cat{S}$.

The \emph{semantic validity of a sequent in $\cat{S}$}, written $\Gamma^\monop \eleq \Delta^\dumonop$, is its value as assigned by the enriched order relation $\eleq$ of $\cat{S}$.
This extends to structures of sequents in $\cat{S}$ by induction, yielding a semantic valuation $\validity{-}_{\cat{S}}$.

\begin{proposition}[Soundness]
	\label[proposition]{prop:soundness}
	For every rule \begin{prooftree}
		\hypo{\hyper{H}}
		\infer1{\hyper{H}'}
	\end{prooftree} of \pQLL,
	\(
	\validity{\hyper{H}}_{\cat{S}} \leq \validity{\hyper{H'}}_{\cat{S}}.
	\)
\end{proposition}
\begin{proof}
	The proof is straightforward as the algebraic properties of softales are directly related to the rules of \pQLL.
	We cover the argument in detail in \cref{app:soundness}.
\end{proof}

Furthermore, we can get completeness by exhibiting a `syntactic softale' construction, simply by constructing their order relations out of the rules of \pQLL and potentially additional axioms.

\begin{definition}
	A \textbf{\pQLL theory} $\Theory$ is a choice of atomic propositions $\Atoms$ plus a set of \textbf{axioms} given by rules of the form\ \ %
	\begin{prooftree}
		\hypo{\red{r}}
		\infer1{\Gamma \entails \Delta}
	\end{prooftree}.
\end{definition}

\begin{definition}\label{def:classifying-softale}
	Let $\Theory$ be a \pQLL theory, its \textbf{classifying $p$-softale} $\SynSoftale{\Theory}$ is defined as follows:
	\begin{enumerate}
		\item its elements are the formulae of \cref{def:qll-formulae} with the atomic propositions of $\Theory$,
		\item its order relation is defined by provability in the theory
		      \begin{equation}
			      \varphi \Rightarrow \psi := \validity{\varphi \entails \psi}_\Theory,
		      \end{equation}
		      where the right-hand side denotes the \ref{eq:provability} of a sequent;
		\item its multiplicative and additive structures are given by the connectives of \pQLL[p], i.e. respectively by $\tensor$, $(-)^*$, and $\plor$.
	\end{enumerate}
	The proof of the necessary laws is in \cref{app:syn-softale}.
\end{definition}

\begin{definition}\label{def:model}
	Any map $\model:\SynSoftale{\Theory} \to \cat{S}$ exhibits a \textbf{model} (or \emph{sound interpretation}) of $\Theory$ in the $p$-softale $\cat{S}$, meaning syntactic validity (provability) bounds below semantic validity:
	\begin{equation}
		\label{eq:soundness}
		\validity{\varphi \entails \psi} \leq (\model(\varphi) \eleq \model(\psi)).
	\end{equation}
	Note that a model satisfies the axioms of $\Theory$, i.e. $r \leq (\model(\Gamma) \eleq \model(\Delta))$ for every axiom \begin{prooftree}
		\hypo{\red{r}}
		\infer1{\Gamma \entails \Delta}
	\end{prooftree} in $\Theory$.
\end{definition}

The classifying softale $\SynSoftale{\Theory}$ enjoys the following universal property: for every $\model$ as above that satisfies the axioms of $\Theory$ there exists a unique lift $\hat \model$ (defined identically to $\model$ on formulae):
\begin{equation}
\label{eq:unique-interp}
	\begin{tikzcd}[ampersand replacement=\&]
		{\SynSoftale{\Atoms}} \& {\cat{S}} \\[-1ex]
		{\SynSoftale{\Theory}}
		\arrow["\model", from=1-1, to=1-2]
		\arrow[two heads, from=1-1, to=2-1]
		\arrow["{{\exists! \hat \model}}"'{pos=0.7}, curve={height=12pt}, dashed, from=2-1, to=1-2]
	\end{tikzcd}
\end{equation}



When $\Theory$ is the empty theory on a set of atoms $\Atoms$, a model $\SynSoftale{\Atoms} \to \cat{S}$ is completely determined by its value on propositional variables, and \emph{vice versa}: on compound formulae the value of $\model$ is forced by preservation of the connectives, while \eqref{eq:soundness} ensures that this assignment is enriched functorial.
Thus we get:

\begin{proposition}\label{prop:freeness}
	The classifying $p$-softales of empty theories are free over sets: for $S$ the underlying set of $\cat{S}$,
	\begin{equation}
		\Softales(\SynSoftale{\Atoms}, \cat{S}) \iso \Set(\Atoms, S).
        \end{equation}
\end{proposition}

Therefore for any valuation $v:\Atoms \to S$ of the propositional variables, we get a unique model $\sem{-}_v$ that exhibits soundness in the sense of provability.

\begin{corollary}[Completeness]
	\pQLL is sound and complete for softales, i.e. for every \pQLL sequent $\varphi \entails \psi$, any $p$-softale and any valuation $v:\Atoms \to S$, we have
	$
		\validity{\varphi \entails \psi} \leq (\sem{\varphi}_v \eleq \sem{\psi}_v),
	$
	and in fact the left-hand side is a minimum:
	\begin{equation}\label{eq:completeness}
		\biginf_{\cat{S} \in \Softales} \biginf_{v:\Atoms \to S} (\sem{\varphi}_v \eleq \sem{\psi}_v) = \validity{\varphi \entails \psi}.
	\end{equation}
\end{corollary}

\subsection{Grounded completeness}
Despite its technical elegance, \pQLL falls short of having a complete semantics in $\pMulReals[p]$.
However, for our goals, a milder form of completeness suffices, which is satisfied by some theories for the following extension of \pQLL.

\begin{definition}\label{def:pqllstar}
	Let \pQLLStar be the quantitative calculus obtained by adding the following rule to \pQLL:
	\begin{equation}\label{rule:mix-star-and-split}
		\begin{prooftree}
			\hypo{\Gamma_1 \entails \Delta_1 \redbbin{\cotensor} \Gamma_2 \entails \Delta_2}
			\infer1[\MixStarRule]{\Gamma_1, \Gamma_2 \entails \Delta_1, \Delta_2}
		\end{prooftree}
	\end{equation}
	Correspondingly, call a $p$-softale \textbf{mix*} when it satisfies
	\begin{equation}\label{eq:softale-interchange-star}
		(a \eleq b) \cotensor (c \eleq d) \leq (a \monop c) \eleq (b \dumonop d);
	\end{equation}
\end{definition}

All the definitions we gave in the previous section apply trivially to \pQLLStar too (note also that \pQLLStar enjoys the same proof-theoretic properties as \pQLL, including cut-elimination).
Moreover, \cref{prop:soundness} still applies (note how soundness of \MixStarRule can be proven just like that of \MixRule by swapping \eqref{eq:softale-interchange} for \eqref{eq:softale-interchange-star}), and therefore we have:

\begin{proposition}
	\pQLLStar theories admit a classifying mix* $p$-softale construction, and it is sound and complete in mix* $p$-softales.
\end{proposition}


We follow the usual terminology and say that a formula is \emph{grounded} if it does not contain propositional variables.
Similarly, in \pQLLStar, we consider theories defined via a set of atoms $\Atoms$ together with a valuation $v:\Atoms \to \PosReals$ that is used to axiomatically ground such atoms.
The \textbf{grounded \pQLLStar-theory $\Theory_v$} is defined by the following axiom schema:
\begin{equation}\label{rule:grounding-schema}
	\begin{prooftree}
		\hypo{\red{v(a)}}
		\infer1[\AtomRightRule{a}]{\ \entails a}
	\end{prooftree}
	\hspace*{10ex}
	\begin{prooftree}
		\hypo{\red{v(a)^*}}
		\infer1[\AtomLeftRule{a}]{a \entails\ }
	\end{prooftree}
	\quad
	\text{for all $a \in \Atoms$.}
\end{equation}
We denote its classifying mix* softale as $\SynSoftale{\Atoms, v}$.

\begin{example}[Bayesian probability]\label{ex:prob}


	As a natural example, we consider the grounded theory associated to a finite set $\Omega$ of \emph{outcomes}, endowed with a probability density ${p:\Omega \to [0,1]}$---we claim this is a logic environment for Bayesian probability.

	Indeed, every event $A \subseteq \Omega$ corresponds to a formula in this theory, namely $\bigvee_{a \in A} \atom{a}$ (where $\atom{a}$ means $a$ considered as constant of the theory), which we denote\footnotemark~again by $A$.
	\footnotetext{This encoding is not perfect, since $A$ is considered a multiset in this way. A first-order encoding would be much more satisfactory.}
	Every sequent $A \entails B$ where $A,B$ are such propositions, corresponds then to the relative probability of $B$ compared to $A$, i.e. the \emph{odds} of $B$ to $A$.\footnote{
		This encoding of Bayesian probability in line with the conversion rate interpretation seen in \cref{sec:interpretation}, since \emph{odds} can be conceived as a conversion rate between bets and payoffs (see \cite{definettiTheoryProbabilityCritical2017}).
	}
	In such a theory, the sequent $A \entails A \cap B$ evaluates to $p(B \mid A)$:
	\begin{equation}
		\begin{prooftree}
			\hypo{\red{\psum[]{b \in A \cap B}\hsum{a \in A}\ p(a)^* \cotensor p(b)}}
			\infer1[$\red{\psum[]{b}\hsum{a}}\, (\AtomLeftRule{a} \redbin{\cotensor} \AtomRightRule{b})$]{\red{\psum[]{b \in A \cap B}\hsum{a \in A}}\ (a \entails \redbin{\cotensor} \entails b)}
			\infer1[$\red{\psum[]{b}\hsum{a}}\,\MixStarRule$]{\red{\psum[]{b \in A \cap B}\hsum{a \in A}}\ a \entails b}
			\infer1[$\red{\psum[]{b}}\,\pOrLeftRule[1]$]{\red{\psum[]{b \in A \cap B}}\ A \entails b}
			\infer1[\pOrRightRule[1]]{A \entails A \cap B}
		\end{prooftree}
	\end{equation}
	and this proof has validity
	\begin{equation}\label{eq:cond-prob-validity}
		\dfrac{\psum[]{b \in A \cap B} p(b)}{\psum[]{a \in A} p(a)} = \dfrac{p(A \cap B)}{p(A)} = p(B \mid A)
	\end{equation}
	whenever $p(A) > 0$.
	When instead $p(A) = 0$, the sequent has validity $\infty$ regardless of $B$ (e.g. by \MixStarRule on \AtomLeftRule{a} and \ExFalsoRule, since $\infty \cotensor 0 = \infty$), consistently with setting $p(B \mid A) = \infty$ for null $A$.
\end{example}

For grounded theories, we have the following completeness theorem:

\begin{theorem}[Grounded Completeness]\label{th:grounded-completeness}
	Let $\varphi$ be a \pQLL formula and let $v:\Atoms \to \PosReals$ be a grounding, and let $\sem{-}_v : \SynSoftale{\Atoms, v} \to \PosReals$ be the unique model of $\Theory_v$ induced by $v$.
	Then
	\begin{equation}
		\sem{\varphi}_v = \validity{\entails \varphi}_{\Theory_v}.
	\end{equation}
\end{theorem}
\begin{proof}
	Since we have soundness already, it suffices to prove the left-hand side is less than or equal to the right-hand side.
	The result follows from an easy induction on $\varphi$ which the reader can find in \cref{app:grounded-completeness-proof}.
	In particular, the grounding schema \eqref{rule:grounding-schema} is used for the base case when $\varphi$ is an atom or the dual of an atom, while \MixStarRule is used to prove the induction step for $\varphi = \varphi_1 \parr \varphi_2$.
\end{proof}

\begin{remark}[\MixStarRule-elimination]
	\MixStarRule is \emph{barely} needed, in the sense that one can prove that it is eliminable in $\Theory_v$ as long as $v$ does not ground any atom at $0$ or $\infty$.
	We call such a valuation \textbf{Cromwellian},\footnote{Cromwell's rule in Bayesian probability asserts that no contingent proposition should be assigned probability $0$ or $1$.} in light of \cref{ex:prob}.
	The elimination proof is an easy induction, following the scheme of \cref{app:grounded-completeness-proof}.
	First one notes that, in every case in which the branches on \MixStarRule have finite non-zero validity, the rule can be replaced by \MixRule since $\tensor$ and $\cotensor$ only differ on $\{0,\infty\}$.
	The remaining cases involve either \TopRightRule/\BotLeftRule, in which case one can show \MixStarRule can be directly eliminated, or involve \AtomLeftRule{a}/\AtomRightRule{a} for $v(a) \in \{0,\infty\}$, but these are excluded in a Cromwellian valuation.

	We note that the theories of neuro-symbolic specifications of \cref{def:nesy-specs-th} below are \emph{not} Cromwellian, so we will need \MixStarRule.
\end{remark}

In particular, this observation gives us:

\begin{corollary}
	If $\varphi$ is a grounded formula, then $\sem{\varphi} = \validity{\ \entails \varphi}_{\pQLL}$, where the first is its unique interpretation in $\pMulReals$ in the sense of \eqref{eq:unique-interp}.
\end{corollary}

\begin{example}[Incompleteness]
	To see concretely how groundedness makes a difference, consider distributivity\footnote{We thank the anonymous reviewer for suggesting this example.}, i.e. the sequent
	$(A \tensor C) \pland (A \tensor B) \entails A \tensor (B \pland C)$.
	This sequent is not provable in \pQLL unless $A,B,C$ are atom-free or we are in a grounded \pQLLStar-theory.
	Indeed the maximal validity proof for it (which is completely routine) necessarily leads to the structure
	$
		{(C \entails B \redbin{\padd} \red{1}) \redbin{\pcoadd} (\red{1} \redbin{\padd} B \entails C)}.
	$
	If both $C \entails B$ and $B \entails C$ have no proof of positive validity, then the validity of this structure is $0$ because $0$ is absorbing for harmonic sum.
	Conversely, when $\validity{B \entails C}^* \leq \validity{C \entails B}$, one can show that
	\begin{equation}
	\label{eq:creation}
		1 \leq (1 \padd[p] (B \entails C)) \pcoadd[p] ((C \entails B) \padd[p] 1)
	\end{equation}
	and thus distributivity holds fully.
	By \cref{lemma:dualism}, $(B \entails C)^* \leq (C \entails B)$ holds for grounded \pQLLStar-theories, and inspecting the proof of the lemma one can appreciate the role of \MixStarRule.

	The inequality \eqref{eq:creation} and its dual (swap sums and harmonic sums) seem pivotal for many other structural properties of the calculus.
	For instance, it arises also when trying to prove $\eta$-expansion for additives: the transformation is validity non-increasing precisely when \eqref{eq:creation} and its dual hold, i.e. in grounded \pQLLStar-theories.
	Other consequences of \eqref{eq:creation} are representability of the additives in the classifying softale of grounded theories, and `soft prelinearity' (in the form $2^{-1/p} \leq (A \entails B) \padd[p] (B \entails A)$).
	Understanding \eqref{eq:creation} thus seems to be of primary importance for furthering the study of \pQLL.
\end{example}

\section{QLL in neuro-symbolic learning and verification}
\label{sec:ML}

In this section, we finally deploy our theoretical framework in the context of neuro-symbolic learning, and specifically property-driven training.
We start with an opinionated account of property-driven training of neural models.

We assume to be training a parameterised model $f_\theta : X \to Y$ with trainable parameters $\theta \in \Theta$, where $X= \R^m$, $Y=\R^n$ and $\Theta = \R^v$ are task-specific smooth manifolds.
In the following, we conceive of $f_\theta$ as a classifier, so that we think of $n$ as a \emph{set} of \emph{classes} or \emph{labels}, and each $f_\theta(x)_i$ as a \emph{logit} (an unnormalized log-probability) assigned to class $i \in n$.

We are also given a logical property specifying the desired behaviour of $f_\theta$ on certain regions of the input space.
Here, we always assume the property to have the following form
\begin{equation}\label{eq:constraint}
	\forall x \in X,\  a \leq x \leq b \implies \varphi(x)
\end{equation}
where $a$, $b$, and $\varphi$ may depend on a choice of point $(\hat{x}, \hat{y}) \in \mathcal{D}$ in a sample of the training data.
The \ref{eq:robustness} property we presented in the introduction is an example of this.

We call $\varphi$ the \textbf{specification} proper, and we conceive it as a proposition in a propositional theory parameterized by $\hat{x}, \hat{y}, x$ and $f_\theta$ that we describe in \cref{sec:nesy-theory} below.

Common practice in machine learning~\cite{van2022analyzing,flinkow2025comparing,flinkowQuantitativeLinearLogic2026} broadly follows the work of Fischer et al.~\cite{fischer2019dl2} in incorporating~\eqref{eq:constraint} into the optimisation objective as follows:
\begin{equation}\label{eq:nesy-objective}
	\tag{\textsc{Opt}}
	\argmin\limits_\theta\ \expected_{(\hat{x},\hat{y}) \sim \mathcal{D}} \biggl[\mathcal{L}_{\mathrm{pred}}(f_\theta(\hat{x}), \hat{y})\ \pAand[1]\ \lambda \sup_{a \leq x \leq b} \sem{\varphi}(f_\theta(x), \hat{y})\biggr],
\end{equation}
where (1) $\mathcal{D} \subseteq X \times Y$ is a sample of the data distribution, (2) $\mathcal{L}_{\mathrm{pred}}$ is any standard loss function concerning the prediction task, (3) $\pAand[1]$ is an operation we define and motivate below (usually $+$), and (4) $\sem{-}$ is the logic-specific semantic function that interprets the property $\varphi$ as a real-valued function.

The inner supremum, computed via Projected Gradient Descent (PGD) \cite{mosbachLogitPairingMethods2019,croceScalingRandomizedGradientFree2020}, makes the optimisation \emph{adversarial}: at each step the optimiser searches for a worst-case violator of $\sem{\varphi}$ in $H = \{a \leq x \leq b\}$, and then nudges $\theta$ so as to reduce the violation. The latter part of the optimisation objective is responsible for ``property-driven training''.

The function $\sem{-}$ is where a differentiable logic comes in, `compiling' the specification language to (differentiable) real-valued functions.
As discussed in \cref{sec:intro}, the state of the art treats this compiler as a heuristic: each choice of $\sem{-}$ comes with its own ad-hoc justification, and none with a theorem relating the \emph{value} of $\sem{\varphi}$ to the logical status of $\varphi$.

Rather, we consider $\sem{-}$ to be a \emph{model} (\cref{def:model}), so that we can directly tie the value of $\sem{\varphi}$ to a notion of logical satisfaction of a certain theory, namely the theory of neuro-symbolic specifications at hand.
This is adequacy (\cref{fig:triangle}).
By \cref{th:grounded-completeness}, we can furthermore equate satisfaction in the model with provability in the theory---a \emph{logical} form of adequacy.

The rest of this section constructs such a result in the framework of \QLL.

\subsection{Specifications as formulae in a theory}\label{sec:nesy-theory}
The property \eqref{eq:constraint} is really a first-order sequent, so it would ideally be accounted for in the logic via first-order facilities, but given a propositional calculus, we content ourselves with accounting for it externally: the adversarial search for $x \in H$ described above discharges the quantifier, and we are left with the propositional specification $\varphi$, instantiated at concrete reals.
These are the formulae of the following \pQLLStar-theory:

\begin{definition}
\label{def:nesy-specs-th}
	The \textbf{theory of neuro-symbolic specifications} for $f_\theta$ at $(\hat{x}, \hat{y}) \in \mathcal{D}$ and $x \in X$, is the grounded theory with set of propositional variables as below left and induced by the valuation as below right:
	\begin{equation}
		K = \{\atom{\hat{y}_1}, \ldots, \atom{\hat{y}_n},\ \atom{y_1}, \ldots, \atom{y_n}\},
		\qquad
		\sem{\atom{\hat{y}_i}}_p = \hat{y}_i,
		\qquad
		\sem{\atom{y_i}}_p = e^{-f_\theta(x)_i}.
	\end{equation}
	This amounts to the rules of \pQLL plus \MixStarRule plus the following grounding schema
	\begin{equation}
		\tag{grounding schema}
		\begin{prooftree}
			\hypo{\red{\hat{y}_i}}
			\infer1{\entails \atom{\hat{y}_i}}
		\end{prooftree}
		\qquad
		\begin{prooftree}
			\hypo{\red{e^{-f_\theta(x)_i}}}
			\infer1{\entails \atom{y_i}}
		\end{prooftree}
		\qquad
		\begin{prooftree}
			\hypo{\red{\hat{y}_i^*}}
			\infer1{\entails \atom{\hat{y}_i}^\dual}
		\end{prooftree}
		\qquad
		\begin{prooftree}
			\hypo{\red{e^{f_\theta(x)_i}}}
			\infer1{\entails \atom{y_i}^\dual}
		\end{prooftree}
	\end{equation}
	We denote such a theory by $\Theory_{x,\hat{y},f_\theta}$.
\end{definition}

The propositional variables $\atom{y_i}$ represent the outputs of the network, while $\atom{\hat{y}_i}$ represent the labels.
The key idea here is to interpret a classifier's outputs and labels as quantitative propositions in their own right: the atoms of the theory are not statements \emph{about} the network, they are the network's outputs, read logically.

\subsection{Losses as elements of a model}\label{sec:ml-algebraic}
Loss functions for neural models are most commonly given by \emph{additive} quantities: the predictions coming out of classifiers are logits, and so is the cross-entropy used in the prediction part of the loss---with this attitude, $\Reals$ is commonly known as \emph{logspace} \cite{badreddineLogLTNDifferentiableFuzzy2023}.

Therefore, when augmenting a loss function with a logical term $\sem{\varphi}$ as in \eqref{eq:nesy-objective}, we need to target the same `units'---this is our \emph{naturality} principle from \cref{sec:intro}, guiding our choice of semantics.

We therefore adopt the \emph{additive} semantics of \QLL, which is the family of softales defined by the data in \cref{table:add-reals}.
Note this \emph{semantic} `additive' vs `multiplicative' distinction is unrelated to the \emph{proof-theoretic} one involving the connectives.
The semantic distinction concerns which arithmetic operation naturally aggregates such quantities, specifically when the objects they pertain to are combined `independently' (e.g. the entropy of independent random variables, the energy of non-interacting physical systems, etc.).

Clearly, additive and multiplicative reals form isomorphic softales for any choice of $p$, as witnessed by \emph{Napier's isomorphism} $-\log \adj 1/\exp$: conjunctive multiplication $\tensor$ becomes ordinary addition, $\padd$ becomes `harmonic' log-sum-exp, and $\pcoadd$ becomes plain log-sum-exp. To avoid confusion with Table~\ref{table:posreals-operations}, let us use a different notation for additive $\padd$ and $\pcoadd$, and write instead $\pAand[p]$ and $\pAor[p]$:
\begin{equation}\label{eq:lse-defs}
	a \pAand[p] b := -\tfrac{1}{p}\log\bigl(e^{-pa}+e^{-pb}\bigr),
	\qquad
	a \pAor[p] b := \tfrac{1}{p}\log\bigl(e^{pa}+e^{pb}\bigr).
\end{equation}
Note that this isomorphism is order-reversing, meaning additive quantities will be `more true' as they get minimised, and this indeed matches how training losses are interpreted.

\begin{table}[tbp]
	\captionsetup{position=bottom, skip=4pt}
	{\footnotesize
	\centering
	\begin{tabularx}{\textwidth}{|c|c|*{3}{>{\centering\arraybackslash}X|}}
		\cline{2-5}
		\multicolumn{1}{c|}{} & \textbf{polarity} & \multicolumn{2}{c|}{\textbf{additive}} & \textbf{multiplicative} \\
		\hline
		\multirow{2}{9ex}[-2ex]{\centering\parbox[c]{9ex}{\centering\textbf{duality}\\$a^* := -a$}} & \textbf{positive} &
		\multicolumn{2}{>{\columncolor{red!15}}c|}{\begin{tikzcd}[ampersand replacement=\&, column sep=small]
			\begin{array}{c}
				\begin{gathered}
					\bot := +\infty \\[-.5ex]
					a \pAor[p] b
				\end{gathered}
			\end{array}
			\&\&
			\begin{array}{c}
				\begin{gathered}
					\bot := +\infty \\[-.5ex]
					a \pAor[\infty] b := \min(a, b)
				\end{gathered}
			\end{array}
			\arrow["{p \to \infty\ }"', to=1-3, from=1-1]
		\end{tikzcd}} & \cellcolor{blue!15}$\begin{gathered} \One := 0 \\[-.5ex] a + b\end{gathered}\quad { -\infty + \infty := +\infty}$ \\[2ex]
		\hhline{~|----|}
		      & \textbf{negative} & \multicolumn{2}{>{\columncolor{blue!15}}c|}{\begin{tikzcd}[ampersand replacement=\&, column sep=small]
				\begin{array}{c}
					\begin{gathered}
						\top := -\infty \\[-.5ex]
						a \pAand[p] b
					\end{gathered}
				\end{array}
				\&\&
				\begin{array}{c}
					\begin{gathered}
						\top := -\infty \\[-.5ex]
						a \pAand[\infty] b := \max(a, b)
					\end{gathered}
				\end{array}
				\arrow["{p \to \infty\ }"', to=1-3, from=1-1]
			 \end{tikzcd}} & \cellcolor{red!15}$\begin{gathered} \One := 0 \\[-.5ex] a +^- b\end{gathered}\quad { -\infty +^- \infty := -\infty}$ \\[2ex]
		\hline
	\end{tabularx}}

	\vspace*{1ex}

	{\footnotesize
	\begin{tabular}[c]{@{}r@{}}
		\textbf{implication} \\
		$a \mulimp b = a^* \cotensor b$ \\
		\textbf{enriched order} \\
		$a \eleq b := a \mulimp b$
	\end{tabular}}
	\hspace*{2ex}
	\begin{tikzcd}[ampersand replacement=\&]
		\pMulReals \&\& \pAddReals
		\arrow[""{name=0, anchor=center, inner sep=0}, "{-\log}"', shift right=1.5, from=1-1, to=1-3]
		\arrow[""{name=1, anchor=center, inner sep=0}, "{1/\exp}"', shift right=1.5, from=1-3, to=1-1]
		\arrow["\sim"{description}, draw=none, from=1, to=0]
	\end{tikzcd}
	\hspace*{2ex}
	{\footnotesize
	\begin{tabular}[c]{@{}l@{}}
		\textbf{implication} \\
		$a \mulimp b = a^* +^- b$, thus $b - a$ \\
		\textbf{enriched order} \\
		$a \eleq b := e^{-(a \mulimp b)}$
	\end{tabular}}

	\caption{\small Above, the family of softales we collectively refer to as the \textbf{additive reals} $\pAddReals$, analogue of \cref{table:posreals-operations}. Below, \textbf{Napier's isomorphism} linking the softales of multiplicative and additive reals.}
	\label{table:add-reals}
\end{table}

\begin{remark}[Analytical considerations]
	Crucially, for finite $p$, the operations are still smooth, strictly entrywise monotone, and have non-vanishing partial derivatives:
	\begin{equation}
		\partial_x (x \pAand[p] y) = \frac{e^{py}}{e^{px}+e^{py}},
		\qquad
		\partial_x (x \pAor[p] y) = \frac{e^{px}}{e^{px}+e^{py}}.
	\end{equation}
	The price for this---loss of idempotency---is precisely quantified by Derivation~\ref{proof:soft-idempotency}: idempotency can be paid for by a factor of $\tfrac1p \log 2$, a cost which vanishes as $p \to \infty$ and is only relevant in the proof-theory, where this is addressed otherwise.

	The reader will find an independent experimental evaluation of this version of
	\QLL against the existing differentiable logics, and the particulars related to optimisation, in Flinkow et al.~\cite{flinkowQuantitativeLinearLogic2026}.
\end{remark}

\subsection{Adequacy and completeness: loss is satisfaction is provability}
\label{sec:ml-adequacy}
Every theory of neuro-symbolic specifications has a canonical interpretation in $\pAddReals$, induced by the valuation of propositional variables
\begin{equation}
	\sem{\atom{\hat{y}_i}}_p = -\log \hat{y}_i,
	\qquad
	\sem{\atom{y_i}}_p = f_\theta(x)_i.
\end{equation}
For each $0 < p \leq \infty$, we thus get a model we can suggestively denote as
\begin{equation}
	\hat{x}, \hat{y}, x, f_\theta \models^p - : \SynSoftale[p]{\Theory_{x,\hat{y},f_\theta}} \to \pAddReals
\end{equation}

\begin{definition}
	The \textbf{quantitative satisfaction} of a specification $\varphi$ by $f_\theta$ at $(\hat{x}, \hat{y})$ and $a \leq x \leq b$ is the semantic validity of $\entails \varphi$ via the canonical interpretation of $\SynSoftale{\Theory_{x,\hat{y},f_\theta}}$ in the additive reals:
	\begin{equation}
		\hat{x}, \hat{y}, x, f_\theta \models^p \varphi
		\ =\ %
		\One \eleq \sem{\varphi}_p
		\ =\ %
		e^{-\sem{\varphi}_p}.
	\end{equation}
\end{definition}

This is adequacy (\cref{fig:triangle}): every value of the loss $\sem{\varphi}$ has an interpretation in terms of satisfaction by a model.
\cref{th:grounded-completeness} then implies a logical adequacy result for property-driven training:

\begin{theorem}\label{th:nnv-adequacy}
	Given $f_\theta$ as above, given $(\hat{x}, \hat{y}) \in \mathcal{D}$ and for each $a \leq x \leq b$, we have
	\begin{equation}
		\hat{x}, \hat{y}, x, f_\theta \models^p \varphi\ =\ |\entails \varphi|_{\Theory_{x,\hat{y},f_\theta}}.
	\end{equation}
\end{theorem}

These results finally allow us to (1) interpret every value of the loss logically, not just the extremal ones, and (2) formally motivate logical loss minimisation as a way to \emph{improve} satisfaction of the specification, since $\sem{\varphi}_p$ decreases as $\hat{x}, \hat{y}, x, f_\theta \models^p \varphi$ increases.

Moreover, the theorem holds for every $p$, thus working both for the soft semantics used at training time and for the hard semantics which can be used at test time to give $\hat{x}, \hat{y}, x, f_\theta \models^\infty \varphi$ a more traditional logical interpretation.
In fact, by \cref{th:isomixmall-cons-and-adeq-over-infty-qll} (and \cref{prop:underlying-softale-is-res-lat}), we can also deduce a traditional, binary, satisfaction result: $\hat{x}, \hat{y}, x, f_\theta \models^p \varphi$ if and only if $|\entails \varphi|_{\Theory_{x,\hat{y},f_\theta}} > 0$.

\section{Conclusions and future work}
\label{sec:conc}

\begin{figure}[t]
	\centering
	\resizebox{1.0\textwidth}{!}{\begin{tikzpicture}[
  scale=1,
  node distance=2.2cm and 2.2cm,
  node font=\Large,
  box/.style={
    draw,
    rectangle,
    minimum width=3.2cm,
    text width=3.2cm,
    minimum height=1.2cm,
    align=center,
    fill=white
  },
  tool/.style={
    draw,
    rectangle,
    minimum width=2.8cm,
    minimum height=0.9cm,
    align=center,
    fill=white
  },
  arrow/.style={->, thick}
]

\node[box] (surface) {Surface \\ language};
\node[box, above=1.5cm of surface] (core) {Core \\ Language};
\node[box, right=1.3cm of core] (intermediate) {Intermediate \\ ITP Language};

\coordinate (branch) at ($(core) + (-2.65, 0)$);
\draw[] (core) -- (branch);

\node[box, left=1.6cm of core] (training) {Loss Functions \\ for Training};
\draw[arrow] (core) -- (training);


\node[box, below=1.45cm of training] (queries) {VNNLIB Queries \\ for Verification};

\draw[arrow] (branch) |- node(branch2) [pos=0.76] {} (queries);
\coordinate (branch3) at (branch2);

\node[tool, right=1.6cm of intermediate] (isabelle) {Isabelle/HOL \\ code};
\node[tool, below=0.4cm of isabelle] (rocq) {Rocq \\ code};
\node[tool, below=0.4cm of rocq] (agda) {Agda \\ code};
\node[tool, left=0.4cm of agda] (imandra) {Imandra \\ code};

\draw[arrow] (surface) -- node(ttc) [text width=1.2cm,align=center, fill=white, draw]{Type-checker} (core);
\draw[arrow] (core) -- (intermediate);

\draw[arrow] (intermediate) -- (agda);
\draw[arrow] (intermediate) -- (rocq);
\draw[arrow] (intermediate) -- (isabelle);
\draw[arrow] (intermediate) -- (imandra);

\draw[dashed] ($(core) + (2.3,2)$) -- ($(core) + (2.3,-3.9)$);

\draw[dashed] ($(core) + (-2.3,2)$) -- ($(core) + (-2.3,-3.9)$);

\node[above=0.32cm of intermediate, text width=3.5cm, align=center, xshift=2.4cm] (symbolic) {\textbf{\underline{Symbolic world}}};

\node[above=0.5cm of core, text width=3.5cm, align=center, yshift=-0.1cm] (neural) {\textbf{\underline{Interface}}};

\node[above=0.5cm of training, text width=3.5cm, align=center, yshift=-0.1cm] (neural) {\textbf{\underline{Neural world}}};

\draw[double=red, double distance=2pt, {Implies[length=10pt,width=14pt]}-{Implies[length=6pt,width=8pt]}, color=red] (training) -- (queries);

\node[left=6.5cm of ttc, text width=4cm, align=center] (justification) {The correspondence justified by this paper};

\draw[->] (justification) -- ([xshift=-4pt] $(training)!0.5!(queries)$);

\begin{pgfonlayer}{background}
\node[draw=none, fit=(core)(ttc)(surface), fill=pink] (core_outer) {};


\node[draw=none, text width=5cm, fit=(intermediate) (isabelle) (agda) (imandra), fill=green!20] (itp_outer1) {};

\node[draw=none, fit=(queries), fill=yellow!20] (itp_outer1) {};

\node[draw=none, fit=(training), fill=blue!20] (itp_outer1) {};
\end{pgfonlayer}
\end{tikzpicture}}
	\caption{\small{The architecture of \vehicle: a surface language in which neural network specifications are written and type-checked, and then connected to both ``neural'' and ``symbolic'' backends.}}
	\label{fig:vehicle-structure}
\end{figure}

\begin{listing}[t]
\centering
\begin{minipage}[t]{0.5\linewidth}
\begin{minted}[numbers=left, fontsize=\scriptsize, xleftmargin=1.5em, breaklines]{haskell}
@parameter
p : Real

qllAdditive : DifferentiableTensorLogic
qllAdditive =
  { trueElement            = -infinity
  , falseElement           = infinity
  , pointwiseNegation      = \x -> -x
  , pointwiseConjunction   = \{dims} x y ->
      const (1/p) dims *
      log(exp(const p dims * x) +
          exp(const p dims * y))
  , pointwiseDisjunction   = \{dims} x y ->
      const (1/p) dims *
      log(exp(const (-p) dims * x) +
          exp(const (-p) dims * y))
  , pointwiseLessEqualThan = \x y -> x - y
  , pointwiseEqual         = \x y ->
      max (x - y) (y - x)
  , pointwiseNotEqual      = \x y ->
      - max (x - y) (y - x)
  , ...
  }

type Image = Tensor Real [28, 28]
type Label = Index 10

validImage : Image -> Bool
validImage x = forall i j . 0 <= x ! i ! j <= 1
\end{minted}
\end{minipage}%
\begin{minipage}[t]{0.46\linewidth}
\begin{minted}[numbers=left, firstnumber=31, fontsize=\scriptsize, xleftmargin=1.5em, breaklines]{haskell}
@network
classifier : Image -> Tensor Real [10]

advises : Image -> Label -> Bool
advises x i =
  forall j . classifier x ! i >= classifier x ! j

@parameter
epsilon : Real

boundedByEpsilon : Image -> Bool
boundedByEpsilon x =
    forall i j . -epsilon <= x ! i ! j <= epsilon

robustAround : Image -> Label -> Bool
robustAround image label = forall perturbation .
  boundedByEpsilon perturbation and
  validImage (image + perturbation) =>
  advises (image + perturbation) label

@dataset
image : Vector Image n

@dataset
label : Vector Label n

@property
robust : Bool
robust = forall i .
  robustAround (image ! i) (label ! i)
\end{minted}
\end{minipage}
\caption{A \vehicle specification for the \ref{eq:robustness} property on MNIST. It starts by specifying the additive semantics of \QLL (cf. \eqref{eq:lse-defs}). Here \mintinline{haskell}{x ! i} accesses the $i$-th row of the tensor $x$.}
\label{lst:robustness-mnist}
\end{listing}

We have set out to supply logical meaning to neuro-symbolic methods by rethinking the fundamentals---calculi, proof validity, and satisfaction---in a quantitative way. We hope that our results invite the community to reconceive differentiable logics not as numerical heuristics, but as bona fide logics in their own right.

From a practical perspective, this work characterises the relationship between minimising the differentiable loss generated from a logical specification and satisfying that specification after training. Although this correspondence motivates much of the existing literature on differentiable logics, it has largely remained implicit.

\vehicle~\cite{daggittVehicleBridgingEmbedding2025,daggitt2026compositional,vehicleTutorial2026} provides a concrete illustration of why such a correspondence is valuable. It unifies property-driven training, neural network verification, and interactive theorem proving through a single dependently typed specification language, as illustrated in \cref{fig:vehicle-structure}. A specification is written and type-checked \emph{once}, before being compiled to loss functions for training, to satisfiability queries for neural network verifiers~\citep{katzMarabouFrameworkVerification2019,marabou2}, and to code for interactive theorem provers.
\cref{lst:robustness-mnist} shows a \vehicle specification of the robustness property introduced in \cref{sec:problem}. Using additive \pQLL, which is itself expressible in the \vehicle DSL (lines 1-23), this specification can be used to train a neural network. We can then verify that the trained network satisfies it. We illustrate the effectiveness of this pipeline for the specification in \cref{lst:robustness-mnist} in the \vehicle tutorial~\cite[\S4]{vehicleTutorial2026}. Our results provide the first crucial step towards providing a formal justification for why success in the former compilation target should assist success in the latter. The red arrow on the left of \cref{fig:vehicle-structure} is a place where \cref{th:nnv-adequacy} can be applied to give such a justification.

The present work is restricted to the propositional fragment of the language. Consequently, it can only handle quantification over finite sets, such as the quantifiers on Lines 29, 36, 43, 59 of \cref{lst:robustness-mnist}, which are simply syntactic sugar for finite conjunctions.
In contrast, the quantifier on Line~46 ranges over the continuous space $\mathbb{R}^{28\times28}$ and therefore lies beyond the scope of the current theory.
In the longer term, we view \pQLL as the foundation of a quantitative metatheory for increasingly expressive neuro-symbolic specification languages, ultimately encompassing the full range of compilation targets supported by systems such as \vehicle.

The natural next step is therefore to extend \pQLL with first-order quantification.
We envisage \emph{soft} quantifiers interpreted as integration over probability spaces, following the ideas of~\cite{capucciQuantifiersQuantitativeReasoning2024}.
We hope to overcome the hurdles encountered in~\cite{capucciQuantifiersQuantitativeReasoning2024} and define a suitably enriched hyperdoctrine in which $p$-integrals enjoy the universal property of quantifiers (after Lawvere, that of adjoints to reindexing \cite{lawvereAdjointnessFoundations1969}).
The associated calculus is a version of linear first-order calculi \cite{amblerFirstOrderLinear1991,dosenCoherenceLinearPredicate2009} in which variables are annotated by their softness.
Therefore we would be able to mix soft and hard regimes.
The usefulness of such a feature is illustrated already by the neuro-symbolic application we covered here: in \eqref{eq:nesy-objective}, $\varphi$ appears bound by a $\sup$---interpreting an $\infty$-quantifier---but also by an expectation---interpreting a $1$-quantifier.  Ahead of presenting the theory of first-order \pQLL, we gathered initial evidence of its empirical effectiveness in the context of \vehicle~\cite{flinkowQuantitativeLinearLogic2026,komendantskayaSpecificationAllYou2026}.


Many questions remain open already at the propositional level.
From a linear logic perspective, we would study normal forms for \pQLL proofs (`quantitative proof nets'), understand how the exponential modalities fit into the picture, and probe potential connections with differential linear logic \cite{ehrhard}.
Closer to \pQLL, the most pressing puzzle is the status of \MixStarRule as a logical rule, hinting at a duality theory at the sequent level that we have yet to uncover; questions of completeness likewise remain, alongside the further exploration of links with fuzzy logics and hypersequent calculi, and a clarification of the peculiar `additive universal properties' initiated in \cref{sec:softale}.

Ultimately, the programme we envision requires work on all three sides of the logic-computation-categories trinity.
Indeed, while Lawvere tackled the latter aspect and this work begins to approach the first, we are unaware of a computational reading of these ideas.

\section*{Data Availability Statement}
This work has no accompanying artifact for it is a purely theoretical work.

\section*{Disclosure of Use of AI in Research}
AI was not used to conduct research according to the \href{https://www.acm.org/publications/policies/new-acm-policy-on-authorship}{ACM definition}.
LLMs were used for copyediting and help with LaTeX.

\bibliographystyle{ACM-Reference-Format}
\bibliography{popl27.bib}


\begin{thebibliography}{87}


\ifx \showCODEN    \undefined \def \showCODEN     #1{\unskip}     \fi
\ifx \showISBNx    \undefined \def \showISBNx     #1{\unskip}     \fi
\ifx \showISBNxiii \undefined \def \showISBNxiii  #1{\unskip}     \fi
\ifx \showISSN     \undefined \def \showISSN      #1{\unskip}     \fi
\ifx \showLCCN     \undefined \def \showLCCN      #1{\unskip}     \fi
\ifx \shownote     \undefined \def \shownote      #1{#1}          \fi
\ifx \showarticletitle \undefined \def \showarticletitle #1{#1}   \fi
\ifx \showURL      \undefined \def \showURL       {\relax}        \fi
\providecommand\bibfield[2]{#2}
\providecommand\bibinfo[2]{#2}
\providecommand\natexlab[1]{#1}
\providecommand\showeprint[2][]{arXiv:#2}

\bibitem[Affeldt et~al\mbox{.}(2024)]%
        {ldl-coq}
\bibfield{author}{\bibinfo{person}{Reynald Affeldt}, \bibinfo{person}{Alessandro Bruni}, \bibinfo{person}{Ekaterina Komendantskaya}, \bibinfo{person}{Natalia Slusarz}, {and} \bibinfo{person}{Kathrin Stark}.} \bibinfo{year}{2024}\natexlab{}.
\newblock \showarticletitle{Taming Differentiable Logics with Coq Formalisation}. In \bibinfo{booktitle}{\emph{15th International Conference on Interactive Theorem Proving ({ITP} 2024), September 9--14, 2024, Tbilisi, Georgia}} \emph{(\bibinfo{series}{LIPIcs}, Vol.~\bibinfo{volume}{309})}. \bibinfo{publisher}{Schloss Dagstuhl - Leibniz-Zentrum f{\"{u}}r Informatik}, \bibinfo{pages}{4:1--4:19}.
\newblock
\href{https://doi.org/10.4230/LIPICS.ITP.2024.4}{doi:\nolinkurl{10.4230/LIPICS.ITP.2024.4}}


\bibitem[Affeldt et~al\mbox{.}(2026)]%
        {affeldt2026foundationdifferentiablelogicsusing}
\bibfield{author}{\bibinfo{person}{Reynald Affeldt}, \bibinfo{person}{Alessandro Bruni}, \bibinfo{person}{Ekaterina Komendantskaya}, \bibinfo{person}{Natalia Ślusarz}, {and} \bibinfo{person}{Kathrin Stark}.} \bibinfo{year}{2026}\natexlab{}.
\newblock \showarticletitle{A Foundation for Differentiable Logics using Dependent Type Theory}.
\newblock \bibinfo{journal}{\emph{Journal of Automated Reasoning}} (\bibinfo{year}{2026}).
\newblock
\urldef\tempurl%
\url{https://arxiv.org/abs/2602.23878}
\showURL{%
\tempurl}


\bibitem[Ambler(1991)]%
        {amblerFirstOrderLinear1991}
\bibfield{author}{\bibinfo{person}{Simon Ambler}.} \bibinfo{year}{1991}\natexlab{}.
\newblock \emph{\bibinfo{title}{First Order Linear Logic in Symmetric Monoidal Closed Categories}}.
\newblock \bibinfo{thesistype}{Ph.\,D. Dissertation}. \bibinfo{school}{University of Edinburgh}, \bibinfo{address}{Edinburgh}.
\newblock
\urldef\tempurl%
\url{https://www.era.lib.ed.ac.uk/bitstream/handle/1842/11974/Ambler1991.Pdf}
\showURL{%
\tempurl}


\bibitem[Atkey and Kokke(2024)]%
        {AK24}
\bibfield{author}{\bibinfo{person}{Robert Atkey} {and} \bibinfo{person}{Wen Kokke}.} \bibinfo{year}{2024}\natexlab{}.
\newblock \showarticletitle{A Semantic Proof of Generalised Cut Elimination for Deep Inference}.
\newblock \bibinfo{journal}{\emph{CoRR}}  \bibinfo{volume}{abs/2404.06233} (\bibinfo{year}{2024}).
\newblock
\href{https://doi.org/10.48550/ARXIV.2404.06233}{doi:\nolinkurl{10.48550/ARXIV.2404.06233}}
\showeprint[arXiv]{2404.06233}


\bibitem[Aubrun and Nechita(2011)]%
        {aubrunMultiplicativePropertyCharacterizes2011}
\bibfield{author}{\bibinfo{person}{Guillaume Aubrun} {and} \bibinfo{person}{Ion Nechita}.} \bibinfo{year}{2011}\natexlab{}.
\newblock \showarticletitle{The Multiplicative Property Characterizes $\ell_p$ and ${{L}}_p$ Norms}.
\newblock \bibinfo{journal}{\emph{Confluentes Mathematici}} \bibinfo{volume}{03}, \bibinfo{number}{04} (\bibinfo{date}{Dec.} \bibinfo{year}{2011}), \bibinfo{pages}{637--647}.
\newblock
\showISSN{1793-7442, 1793-7434}
\href{https://doi.org/10.1142/S1793744211000485}{doi:\nolinkurl{10.1142/S1793744211000485}}
\showeprint{1102.2618}~[math]


\bibitem[Avron(1987)]%
        {avronConstructiveAnalysisRM1987}
\bibfield{author}{\bibinfo{person}{Arnon Avron}.} \bibinfo{year}{1987}\natexlab{}.
\newblock \showarticletitle{A Constructive Analysis of {{RM}}}.
\newblock \bibinfo{journal}{\emph{The Journal of symbolic logic}} \bibinfo{volume}{52}, \bibinfo{number}{4} (\bibinfo{year}{1987}), \bibinfo{pages}{939--951}.
\newblock
\urldef\tempurl%
\url{https://www.cambridge.org/core/journals/journal-of-symbolic-logic/article/constructive-analysis-of-rm/342F0DE51D0878BEC48C03E65FBA03BC}
\showURL{%
\tempurl}


\bibitem[Baaz et~al\mbox{.}(2003)]%
        {BaazCF03}
\bibfield{author}{\bibinfo{person}{Matthias Baaz}, \bibinfo{person}{Agata Ciabattoni}, {and} \bibinfo{person}{Christian~G. Ferm{\"{u}}ller}.} \bibinfo{year}{2003}\natexlab{}.
\newblock \showarticletitle{Hypersequent Calculi for G{\"{o}}del Logics - a Survey}.
\newblock \bibinfo{journal}{\emph{J. Log. Comput.}} \bibinfo{volume}{13}, \bibinfo{number}{6} (\bibinfo{year}{2003}), \bibinfo{pages}{835--861}.
\newblock
\href{https://doi.org/10.1093/LOGCOM/13.6.835}{doi:\nolinkurl{10.1093/LOGCOM/13.6.835}}


\bibitem[Bacci et~al\mbox{.}(2023)]%
        {bacciPropositionalLogicsLawvere2023}
\bibfield{author}{\bibinfo{person}{Giorgio Bacci}, \bibinfo{person}{Radu Mardare}, \bibinfo{person}{Prakash Panangaden}, {and} \bibinfo{person}{Gordon Plotkin}.} \bibinfo{year}{2023}\natexlab{}.
\newblock \showarticletitle{Propositional {{Logics}} for the {{Lawvere Quantale}}}.
\newblock \bibinfo{journal}{\emph{Electronic Notes in Theoretical Informatics and Computer Science}}  \bibinfo{volume}{Volume 3 - Proceedings of...} (\bibinfo{date}{Nov.} \bibinfo{year}{2023}), \bibinfo{pages}{12292}.
\newblock
\showISSN{2969-2431}
\href{https://doi.org/10.46298/entics.12292}{doi:\nolinkurl{10.46298/entics.12292}}


\bibitem[Bacci et~al\mbox{.}(2024)]%
        {bacciPolynomialLawvereLogic2024}
\bibfield{author}{\bibinfo{person}{Giorgio Bacci}, \bibinfo{person}{Radu Mardare}, \bibinfo{person}{Prakash Panangaden}, {and} \bibinfo{person}{Gordon Plotkin}.} \bibinfo{year}{2024}\natexlab{}.
\newblock \bibinfo{title}{Polynomial {{Lawvere Logic}}}.
\newblock
\href{https://doi.org/10.48550/arXiv.2402.03543}{doi:\nolinkurl{10.48550/arXiv.2402.03543}}
\showeprint{2402.03543}


\bibitem[Bacci et~al\mbox{.}(2026)]%
        {bacciRationalLawvereLogic2026}
\bibfield{author}{\bibinfo{person}{Giorgio Bacci}, \bibinfo{person}{Radu Mardare}, \bibinfo{person}{Prakash Panangaden}, {and} \bibinfo{person}{Gordon Plotkin}.} \bibinfo{year}{2026}\natexlab{}.
\newblock \showarticletitle{Rational {{Lawvere Logic}}}. In \bibinfo{booktitle}{\emph{34th {{EACSL Annual Conference}} on {{Computer Science Logic}} ({{CSL}} 2026)}} \emph{(\bibinfo{series}{Leibniz {{International Proceedings}} in {{Informatics}} ({{LIPIcs}})}, Vol.~\bibinfo{volume}{363})}, \bibfield{editor}{\bibinfo{person}{Stefano Guerrini} {and} \bibinfo{person}{Barbara K{\"o}nig}} (Eds.). \bibinfo{publisher}{Schloss Dagstuhl -- Leibniz-Zentrum f\"ur Informatik}, \bibinfo{address}{Dagstuhl, Germany}, \bibinfo{pages}{3:1--3:21}.
\newblock
\showISBNx{978-3-95977-411-6}
\showISSN{1868-8969}
\href{https://doi.org/10.4230/LIPIcs.CSL.2026.3}{doi:\nolinkurl{10.4230/LIPIcs.CSL.2026.3}}


\bibitem[Bacci and M{\o}gelberg(2026)]%
        {bacciInductionRecursionPrinciples2026}
\bibfield{author}{\bibinfo{person}{Giorgio Bacci} {and} \bibinfo{person}{Rasmus~Ejlers M{\o}gelberg}.} \bibinfo{year}{2026}\natexlab{}.
\newblock \showarticletitle{Induction and {{Recursion Principles}} in a {{Higher-Order Quantitative Logic}} for {{Probability}}}. In \bibinfo{booktitle}{\emph{41st {{Annual Symposium}} on {{Logic}} in {{Computer Science}} ({{LICS}} 2026)}} \emph{(\bibinfo{series}{Leibniz {{International Proceedings}} in {{Informatics}} ({{LIPIcs}})}, Vol.~\bibinfo{volume}{380})}, \bibfield{editor}{\bibinfo{person}{Claudia Faggian} {and} \bibinfo{person}{Joost-Pieter Katoen}} (Eds.). \bibinfo{publisher}{Schloss Dagstuhl -- Leibniz-Zentrum f\"ur Informatik}, \bibinfo{address}{Dagstuhl, Germany}, \bibinfo{pages}{6:1--6:27}.
\newblock
\showISBNx{978-3-95977-434-5}
\showISSN{1868-8969}
\href{https://doi.org/10.4230/LIPIcs.LICS.2026.6}{doi:\nolinkurl{10.4230/LIPIcs.LICS.2026.6}}


\bibitem[Badreddine et~al\mbox{.}(2022)]%
        {badreddineLogicTensorNetworks2022}
\bibfield{author}{\bibinfo{person}{Samy Badreddine}, \bibinfo{person}{Artur d'Avila Garcez}, \bibinfo{person}{Luciano Serafini}, {and} \bibinfo{person}{Michael Spranger}.} \bibinfo{year}{2022}\natexlab{}.
\newblock \showarticletitle{Logic {{Tensor Networks}}}.
\newblock \bibinfo{journal}{\emph{Artificial Intelligence}}  \bibinfo{volume}{303} (\bibinfo{date}{Feb.} \bibinfo{year}{2022}), \bibinfo{pages}{103649}.
\newblock
\showISSN{00043702}
\href{https://doi.org/10.1016/j.artint.2021.103649}{doi:\nolinkurl{10.1016/j.artint.2021.103649}}


\bibitem[Badreddine et~al\mbox{.}(2023)]%
        {badreddineLogLTNDifferentiableFuzzy2023}
\bibfield{author}{\bibinfo{person}{Samy Badreddine}, \bibinfo{person}{Luciano Serafini}, {and} \bibinfo{person}{Michael Spranger}.} \bibinfo{year}{2023}\natexlab{}.
\newblock \bibinfo{title}{{{logLTN}}: {{Differentiable Fuzzy Logic}} in the {{Logarithm Space}}}.
\newblock
\href{https://doi.org/10.48550/arXiv.2306.14546}{doi:\nolinkurl{10.48550/arXiv.2306.14546}}
\showeprint[arxiv]{2306.14546}~[cs.AI]


\bibitem[Baez(2024)]%
        {baezWhatEntropy2024}
\bibfield{author}{\bibinfo{person}{John~C. Baez}.} \bibinfo{year}{2024}\natexlab{}.
\newblock \bibinfo{title}{What Is {{Entropy}}?}
\newblock
\href{https://doi.org/10.48550/arXiv.2409.09232}{doi:\nolinkurl{10.48550/arXiv.2409.09232}}
\showeprint{2409.09232}~[cond-mat]


\bibitem[Bailitis et~al\mbox{.}(2026)]%
        {Bailitis2026}
\bibfield{author}{\bibinfo{person}{Janis Bailitis}, \bibinfo{person}{Reynald Affeldt}, \bibinfo{person}{Alessandro Bruni}, \bibinfo{person}{Alessio Coltellacci}, \bibinfo{person}{Ekaterina Komendantskaya}, {and} \bibinfo{person}{Kathrin Stark}.} \bibinfo{year}{2026}\natexlab{}.
\newblock \showarticletitle{{Towards Quantitative Logics in Rocq}}.
\newblock \bibinfo{journal}{\emph{The Rocqshop}} (\bibinfo{year}{2026}).
\newblock
\urldef\tempurl%
\url{https://coq-workshop.gitlab.io/2026/files/EA1.pdf}
\showURL{%
\tempurl}


\bibitem[Barr(1979)]%
        {barrAutonomousCategories1979}
\bibfield{author}{\bibinfo{person}{Michael Barr}.} \bibinfo{year}{1979}\natexlab{}.
\newblock \bibinfo{booktitle}{\emph{$*$-{{Autonomous Categories}}}}. \bibinfo{series}{Lecture {{Notes}} in {{Mathematics}}}, Vol.~\bibinfo{volume}{752}.
\newblock \bibinfo{publisher}{Springer}, \bibinfo{address}{Berlin, Heidelberg}.
\newblock
\showISBNx{978-3-540-09563-7 978-3-540-34850-4}
\href{https://doi.org/10.1007/BFb0064579}{doi:\nolinkurl{10.1007/BFb0064579}}


\bibitem[Besold et~al\mbox{.}(2021)]%
        {DBLP:series/faia/BesoldGBBDHKLLPPPZ21}
\bibfield{author}{\bibinfo{person}{Tarek~R. Besold}, \bibinfo{person}{Artur~S. d'Avila Garcez}, \bibinfo{person}{Sebastian Bader}, \bibinfo{person}{Howard Bowman}, \bibinfo{person}{Pedro~M. Domingos}, \bibinfo{person}{Pascal Hitzler}, \bibinfo{person}{Kai{-}Uwe K{\"{u}}hnberger}, \bibinfo{person}{Lu{\'{\i}}s~C. Lamb}, \bibinfo{person}{Priscila Machado~Vieira Lima}, \bibinfo{person}{Leo de Penning}, \bibinfo{person}{Gadi Pinkas}, \bibinfo{person}{Hoifung Poon}, {and} \bibinfo{person}{Gerson Zaverucha}.} \bibinfo{year}{2021}\natexlab{}.
\newblock \showarticletitle{Neural-Symbolic Learning and Reasoning: {A} Survey and Interpretation}.
\newblock In \bibinfo{booktitle}{\emph{Neuro-Symbolic Artificial Intelligence: The State of the Art}}, \bibfield{editor}{\bibinfo{person}{Pascal Hitzler} {and} \bibinfo{person}{Md.~Kamruzzaman Sarker}} (Eds.). \bibinfo{series}{Frontiers in Artificial Intelligence and Applications}, Vol.~\bibinfo{volume}{342}. \bibinfo{publisher}{{IOS} Press}, \bibinfo{pages}{1--51}.
\newblock
\href{https://doi.org/10.3233/FAIA210348}{doi:\nolinkurl{10.3233/FAIA210348}}


\bibitem[Buss(1998)]%
        {Buss98}
\bibfield{author}{\bibinfo{person}{Samuel~R. Buss}.} \bibinfo{year}{1998}\natexlab{}.
\newblock \showarticletitle{An Introduction to Proof Theory}.
\newblock In \bibinfo{booktitle}{\emph{Handbook of Proof Theory}}, \bibfield{editor}{\bibinfo{person}{Samuel~R. Buss}} (Ed.). \bibinfo{publisher}{Elsevier}, \bibinfo{pages}{1--78}.
\newblock


\bibitem[Capucci(2024)]%
        {capucciQuantifiersQuantitativeReasoning2024}
\bibfield{author}{\bibinfo{person}{Matteo Capucci}.} \bibinfo{year}{2024}\natexlab{}.
\newblock \bibinfo{title}{On {{Quantifiers}} for {{Quantitative Reasoning}}}.
\newblock
\href{https://doi.org/10.48550/arXiv.2406.04936}{doi:\nolinkurl{10.48550/arXiv.2406.04936}}
\showeprint{2406.04936}~[cs, math]


\bibitem[Casadio et~al\mbox{.}(2022)]%
        {casadioNeuralNetworkRobustness2022}
\bibfield{author}{\bibinfo{person}{Marco Casadio}, \bibinfo{person}{Ekaterina Komendantskaya}, \bibinfo{person}{Matthew~L. Daggitt}, \bibinfo{person}{Wen Kokke}, \bibinfo{person}{Guy Katz}, \bibinfo{person}{Guy Amir}, {and} \bibinfo{person}{Idan Refaeli}.} \bibinfo{year}{2022}\natexlab{}.
\newblock \bibinfo{title}{Neural {{Network Robustness}} as a {{Verification Property}}: {{A Principled Case Study}}}.
\newblock
\href{https://doi.org/10.48550/arXiv.2104.01396}{doi:\nolinkurl{10.48550/arXiv.2104.01396}}
\showeprint[arxiv]{2104.01396}~[cs]


\bibitem[Ciabattoni and Genco(2018)]%
        {CiabattoniG18}
\bibfield{author}{\bibinfo{person}{Agata Ciabattoni} {and} \bibinfo{person}{Francesco~A. Genco}.} \bibinfo{year}{2018}\natexlab{}.
\newblock \showarticletitle{Hypersequents and Systems of Rules: Embeddings and Applications}.
\newblock \bibinfo{journal}{\emph{{ACM} Trans. Comput. Log.}} \bibinfo{volume}{19}, \bibinfo{number}{2} (\bibinfo{year}{2018}), \bibinfo{pages}{11:1--11:27}.
\newblock
\href{https://doi.org/10.1145/3180075}{doi:\nolinkurl{10.1145/3180075}}


\bibitem[Ciabattoni et~al\mbox{.}(2021)]%
        {CiabattoniLR21}
\bibfield{author}{\bibinfo{person}{Agata Ciabattoni}, \bibinfo{person}{Timo Lang}, {and} \bibinfo{person}{Revantha Ramanayake}.} \bibinfo{year}{2021}\natexlab{}.
\newblock \showarticletitle{Bounded-analytic Sequent Calculi and Embeddings for Hypersequent Logics}.
\newblock \bibinfo{journal}{\emph{J. Symb. Log.}} \bibinfo{volume}{86}, \bibinfo{number}{2} (\bibinfo{year}{2021}), \bibinfo{pages}{635--668}.
\newblock
\href{https://doi.org/10.1017/JSL.2021.42}{doi:\nolinkurl{10.1017/JSL.2021.42}}


\bibitem[Cintula et~al\mbox{.}(2011a)]%
        {cintulaHandbookMathematicalFuzzy2011}
\bibfield{author}{\bibinfo{person}{Petr Cintula}, \bibinfo{person}{Petr Hajek}, {and} \bibinfo{person}{Carles Noguera}.} \bibinfo{year}{2011}\natexlab{a}.
\newblock \bibinfo{booktitle}{\emph{Handbook of Mathematical Fuzzy Logic}}. \bibinfo{series}{Studies in Logic}, Vol.~\bibinfo{volume}{1}.
\newblock \bibinfo{publisher}{College Publications}, \bibinfo{address}{London, UK}.
\newblock
\showISBNx{978-1-84890-039-4}


\bibitem[Cintula et~al\mbox{.}(2011b)]%
        {cintulaHandbookMathematicalFuzzy2011a}
\bibfield{author}{\bibinfo{person}{Petr Cintula}, \bibinfo{person}{Petr Hajek}, {and} \bibinfo{person}{Carles Noguera}.} \bibinfo{year}{2011}\natexlab{b}.
\newblock \bibinfo{booktitle}{\emph{Handbook of Mathematical Fuzzy Logic}}. \bibinfo{series}{Studies in Logic}, Vol.~\bibinfo{volume}{2}.
\newblock \bibinfo{publisher}{College Publications}, \bibinfo{address}{London, UK}.
\newblock
\showISBNx{978-1-84890-054-7}


\bibitem[Cockett and Seely(1997)]%
        {cockettProofTheoryFull1997}
\bibfield{author}{\bibinfo{person}{J.~R.~B. Cockett} {and} \bibinfo{person}{R.~A.~G. Seely}.} \bibinfo{year}{1997}\natexlab{}.
\newblock \showarticletitle{Proof Theory for Full Intuitionistic Linear Logic, Bilinear Logic, and {{MIX}} Categories.}
\newblock \bibinfo{journal}{\emph{Theory and Applications of Categories}}  \bibinfo{volume}{3} (\bibinfo{year}{1997}), \bibinfo{pages}{85--131}.
\newblock
\showISSN{1201-561X}
\urldef\tempurl%
\url{http://www.tac.mta.ca/tac/volumes/1997/n5/n5.pdf}
\showURL{%
\tempurl}


\bibitem[Cockett and Seely(1999)]%
        {cockettLinearlyDistributiveFunctors1999}
\bibfield{author}{\bibinfo{person}{J.~R.~B. Cockett} {and} \bibinfo{person}{R.~A.~G. Seely}.} \bibinfo{year}{1999}\natexlab{}.
\newblock \showarticletitle{Linearly Distributive Functors}.
\newblock \bibinfo{journal}{\emph{Journal of Pure and Applied Algebra}} \bibinfo{volume}{143}, \bibinfo{number}{1} (\bibinfo{date}{Nov.} \bibinfo{year}{1999}), \bibinfo{pages}{155--203}.
\newblock
\showISSN{0022-4049}
\href{https://doi.org/10.1016/S0022-4049(98)00110-8}{doi:\nolinkurl{10.1016/S0022-4049(98)00110-8}}


\bibitem[Cordeiro et~al\mbox{.}(2025)]%
        {CordeiroDGIJKKLMSW25}
\bibfield{author}{\bibinfo{person}{Lucas~C. Cordeiro}, \bibinfo{person}{Matthew~L. Daggitt}, \bibinfo{person}{Julien Girard{-}Satabin}, \bibinfo{person}{Omri Isac}, \bibinfo{person}{Taylor~T. Johnson}, \bibinfo{person}{Guy Katz}, \bibinfo{person}{Ekaterina Komendantskaya}, \bibinfo{person}{Augustin Lemesle}, \bibinfo{person}{Edoardo Manino}, \bibinfo{person}{Artjoms Sinkarovs}, {and} \bibinfo{person}{Haoze Wu}.} \bibinfo{year}{2025}\natexlab{}.
\newblock \showarticletitle{Neural Network Verification is a Programming Language Challenge}. In \bibinfo{booktitle}{\emph{34th European Symposium on Programming ({ESOP} 2025), Hamilton, ON, Canada, May 5--8, 2025}} \emph{(\bibinfo{series}{Lecture Notes in Computer Science}, Vol.~\bibinfo{volume}{15694})}. \bibinfo{publisher}{Springer}, \bibinfo{pages}{206--235}.
\newblock
\href{https://doi.org/10.1007/978-3-031-91118-7\_9}{doi:\nolinkurl{10.1007/978-3-031-91118-7\_9}}


\bibitem[Croce et~al\mbox{.}(2020)]%
        {croceScalingRandomizedGradientFree2020}
\bibfield{author}{\bibinfo{person}{Francesco Croce}, \bibinfo{person}{Jonas Rauber}, {and} \bibinfo{person}{Matthias Hein}.} \bibinfo{year}{2020}\natexlab{}.
\newblock \showarticletitle{Scaling up the {{Randomized Gradient-Free Adversarial Attack Reveals Overestimation}} of {{Robustness Using Established Attacks}}}.
\newblock \bibinfo{journal}{\emph{International Journal of Computer Vision}} \bibinfo{volume}{128}, \bibinfo{number}{4} (\bibinfo{date}{April} \bibinfo{year}{2020}), \bibinfo{pages}{1028--1046}.
\newblock
\showISSN{1573-1405}
\href{https://doi.org/10.1007/s11263-019-01213-0}{doi:\nolinkurl{10.1007/s11263-019-01213-0}}


\bibitem[Daggitt et~al\mbox{.}(2025)]%
        {daggittVehicleBridgingEmbedding2025}
\bibfield{author}{\bibinfo{person}{Matthew~L. Daggitt}, \bibinfo{person}{Wen Kokke}, \bibinfo{person}{Robert Atkey}, \bibinfo{person}{Ekaterina Komendantskaya}, \bibinfo{person}{Natalia Slusarz}, {and} \bibinfo{person}{Luca Arnaboldi}.} \bibinfo{year}{2025}\natexlab{}.
\newblock \showarticletitle{Vehicle: {{Bridging}} the {{Embedding Gap}} in the {{Verification}} of {{Neuro-Symbolic Programs}}}.
\newblock \bibinfo{journal}{\emph{LIPIcs, Volume 337, FSCD 2025}}  \bibinfo{volume}{337} (\bibinfo{year}{2025}), \bibinfo{pages}{2:1--2:20}.
\newblock
\showISBNx{9783959773744}
\showISSN{1868-8969}
\href{https://doi.org/10.4230/LIPICS.FSCD.2025.2}{doi:\nolinkurl{10.4230/LIPICS.FSCD.2025.2}}


\bibitem[Daggitt et~al\mbox{.}(2026)]%
        {daggitt2026compositional}
\bibfield{author}{\bibinfo{person}{Matthew~L. Daggitt}, \bibinfo{person}{Ekaterina Komendantskaya}, \bibinfo{person}{Alistair Sirman}, \bibinfo{person}{Alessandro Bruni}, \bibinfo{person}{Samuel Teuber}, \bibinfo{person}{Josh Smart}, {and} \bibinfo{person}{Grant Passmore}.} \bibinfo{year}{2026}\natexlab{}.
\newblock \showarticletitle{Compositional Neural-Cyber-Physical System Verification in the Interactive Theorem Prover of Your Choice}.
\newblock \bibinfo{journal}{\emph{Proceedings of the ACM on Programming Languages}} \bibinfo{volume}{10}, \bibinfo{number}{ICFP} (\bibinfo{year}{2026}).
\newblock
\newblock
\shownote{To appear}.


\bibitem[Dal~Lago and Hofmann(2005)]%
        {dallagoQuantitative}
\bibfield{author}{\bibinfo{person}{Ugo Dal~Lago} {and} \bibinfo{person}{Martin Hofmann}.} \bibinfo{year}{2005}\natexlab{}.
\newblock \showarticletitle{Quantitative {{Models}} and {{Implicit Complexity}}}. In \bibinfo{booktitle}{\emph{{{FSTTCS}} 2005: {{Foundations}} of {{Software Technology}} and {{Theoretical Computer Science}}}} \emph{(\bibinfo{series}{Lecture Notes in Computer Science}, Vol.~\bibinfo{volume}{3821})}, \bibfield{editor}{\bibinfo{person}{Sundar Sarukkai} {and} \bibinfo{person}{Sandeep Sen}} (Eds.). \bibinfo{publisher}{Springer}, \bibinfo{address}{Berlin, Heidelberg}, \bibinfo{pages}{189--200}.
\newblock
\showISBNx{978-3-540-32419-5}
\href{https://doi.org/10.1007/11590156_15}{doi:\nolinkurl{10.1007/11590156_15}}


\bibitem[d'Avila Garcez et~al\mbox{.}(2002)]%
        {garcezNeuralSymbolicLearningSystems2002}
\bibfield{author}{\bibinfo{person}{Artur~S. d'Avila Garcez}, \bibinfo{person}{Krysia~B. Broda}, {and} \bibinfo{person}{Dov~M. Gabbay}.} \bibinfo{year}{2002}\natexlab{}.
\newblock \bibinfo{booktitle}{\emph{Neural-Symbolic Learning Systems: Foundations and Applications}}.
\newblock \bibinfo{publisher}{Springer}, \bibinfo{address}{London}.
\newblock
\showISBNx{978-1-85233-512-0}
\href{https://doi.org/10.1007/978-1-4471-0211-3}{doi:\nolinkurl{10.1007/978-1-4471-0211-3}}


\bibitem[d'Avila Garcez et~al\mbox{.}(2009)]%
        {DBLP:series/cogtech/GarcezLG2009}
\bibfield{author}{\bibinfo{person}{Artur~S. d'Avila Garcez}, \bibinfo{person}{Lu{\'{\i}}s~C. Lamb}, {and} \bibinfo{person}{Dov~M. Gabbay}.} \bibinfo{year}{2009}\natexlab{}.
\newblock \bibinfo{booktitle}{\emph{Neural-Symbolic Cognitive Reasoning}}.
\newblock \bibinfo{publisher}{Springer}.
\newblock
\showISBNx{978-3-540-73245-7}
\href{https://doi.org/10.1007/978-3-540-73246-4}{doi:\nolinkurl{10.1007/978-3-540-73246-4}}


\bibitem[Day(1970)]%
        {dayClosedCategoriesFunctors1970}
\bibfield{author}{\bibinfo{person}{Brian Day}.} \bibinfo{year}{1970}\natexlab{}.
\newblock \showarticletitle{On Closed Categories of Functors}.
\newblock In \bibinfo{booktitle}{\emph{Reports of the {{Midwest Category Seminar IV}}}}, \bibfield{editor}{\bibinfo{person}{S.~MacLane}, \bibinfo{person}{H.~Applegate}, \bibinfo{person}{M.~Barr}, \bibinfo{person}{B.~Day}, \bibinfo{person}{E.~Dubuc}, \bibinfo{person}{{Phreilambud}}, \bibinfo{person}{A.~Pultr}, \bibinfo{person}{R.~Street}, \bibinfo{person}{M.~Tierney}, {and} \bibinfo{person}{S.~Swierczkowski}} (Eds.). Vol.~\bibinfo{volume}{137}. \bibinfo{publisher}{Springer Berlin Heidelberg}, \bibinfo{address}{Berlin, Heidelberg}, \bibinfo{pages}{1--38}.
\newblock
\showISBNx{978-3-540-04926-5 978-3-540-36292-0}
\href{https://doi.org/10.1007/BFb0060438}{doi:\nolinkurl{10.1007/BFb0060438}}


\bibitem[De~Finetti(2017)]%
        {definettiTheoryProbabilityCritical2017}
\bibfield{author}{\bibinfo{person}{Bruno De~Finetti}.} \bibinfo{year}{2017}\natexlab{}.
\newblock \bibinfo{booktitle}{\emph{Theory of Probability: {{A}} Critical Introductory Treatment}}.
\newblock \bibinfo{publisher}{John Wiley \& Sons}.
\newblock


\bibitem[Do{\v s}en and Petri{\'c}(2009)]%
        {dosenCoherenceLinearPredicate2009}
\bibfield{author}{\bibinfo{person}{Kosta Do{\v s}en} {and} \bibinfo{person}{Zoran Petri{\'c}}.} \bibinfo{year}{2009}\natexlab{}.
\newblock \showarticletitle{Coherence in Linear Predicate Logic}.
\newblock \bibinfo{journal}{\emph{Annals of Pure and Applied Logic}} \bibinfo{volume}{158}, \bibinfo{number}{1} (\bibinfo{date}{March} \bibinfo{year}{2009}), \bibinfo{pages}{125--153}.
\newblock
\showISSN{0168-0072}
\href{https://doi.org/10.1016/j.apal.2008.10.003}{doi:\nolinkurl{10.1016/j.apal.2008.10.003}}


\bibitem[Ehrhard(2018)]%
        {ehrhard}
\bibfield{author}{\bibinfo{person}{Thomas Ehrhard}.} \bibinfo{year}{2018}\natexlab{}.
\newblock \showarticletitle{An Introduction to Differential Linear Logic: Proof-Nets, Models and Antiderivatives}.
\newblock \bibinfo{journal}{\emph{Mathematical Structures in Computer Science}} \bibinfo{volume}{28}, \bibinfo{number}{7} (\bibinfo{date}{Aug.} \bibinfo{year}{2018}), \bibinfo{pages}{995--1060}.
\newblock
\showISSN{0960-1295, 1469-8072}
\href{https://doi.org/10.1017/S0960129516000372}{doi:\nolinkurl{10.1017/S0960129516000372}}


\bibitem[Ehrhard et~al\mbox{.}(2023)]%
        {ehrhardSumProduct2023}
\bibfield{author}{\bibinfo{person}{Thomas Ehrhard}, \bibinfo{person}{Claudia Faggian}, {and} \bibinfo{person}{Michele Pagani}.} \bibinfo{year}{2023}\natexlab{}.
\newblock \showarticletitle{The {{Sum-Product Algorithm}} for {{Quantitative Multiplicative Linear Logic}}}. In \bibinfo{booktitle}{\emph{{{LIPIcs}}, {{Volume}} 260, {{FSCD}} 2023}}, \bibfield{editor}{\bibinfo{person}{Marco Gaboardi} {and} \bibinfo{person}{Femke {van Raamsdonk}}} (Eds.), Vol.~\bibinfo{volume}{260}. \bibinfo{publisher}{Schloss Dagstuhl -- Leibniz-Zentrum f\"ur Informatik}, \bibinfo{pages}{8:1--8:18}.
\newblock
\showISSN{1868-8969}
\href{https://doi.org/10.4230/LIPIcs.FSCD.2023.8}{doi:\nolinkurl{10.4230/LIPIcs.FSCD.2023.8}}


\bibitem[Fischer et~al\mbox{.}(2019)]%
        {fischer2019dl2}
\bibfield{author}{\bibinfo{person}{Marc Fischer}, \bibinfo{person}{Mislav Balunovic}, \bibinfo{person}{Dana {Drachsler-Cohen}}, \bibinfo{person}{Timon Gehr}, \bibinfo{person}{Ce Zhang}, {and} \bibinfo{person}{Martin Vechev}.} \bibinfo{year}{2019}\natexlab{}.
\newblock \showarticletitle{{DL2}: Training and Querying Neural Networks with Logic}. In \bibinfo{booktitle}{\emph{Proceedings of the 36th International Conference on Machine Learning (ICML 2019)}} \emph{(\bibinfo{series}{Proceedings of Machine Learning Research}, Vol.~\bibinfo{volume}{97})}, \bibfield{editor}{\bibinfo{person}{Kamalika Chaudhuri} {and} \bibinfo{person}{Ruslan Salakhutdinov}} (Eds.). \bibinfo{publisher}{PMLR}, \bibinfo{pages}{1931--1941}.
\newblock
\showISSN{2640-3498}
\urldef\tempurl%
\url{https://proceedings.mlr.press/v97/fischer19a.html}
\showURL{%
\tempurl}


\bibitem[Flinkow et~al\mbox{.}(2026)]%
        {flinkowQuantitativeLinearLogic2026}
\bibfield{author}{\bibinfo{person}{Thomas Flinkow}, \bibinfo{person}{Ekaterina Komendantskaya}, \bibinfo{person}{Matteo Capucci}, {and} \bibinfo{person}{Rosemary Monahan}.} \bibinfo{year}{2026}\natexlab{}.
\newblock \bibinfo{title}{Quantitative {{Linear Logic}} for {{Neuro-Symbolic Learning}} and {{Verification}}}.
\newblock
\href{https://doi.org/10.48550/ARXIV.2605.13845}{doi:\nolinkurl{10.48550/ARXIV.2605.13845}}


\bibitem[Flinkow et~al\mbox{.}(2025)]%
        {flinkow2025comparing}
\bibfield{author}{\bibinfo{person}{Thomas Flinkow}, \bibinfo{person}{Barak~A Pearlmutter}, {and} \bibinfo{person}{Rosemary Monahan}.} \bibinfo{year}{2025}\natexlab{}.
\newblock \showarticletitle{Comparing differentiable logics for learning with logical constraints}.
\newblock \bibinfo{journal}{\emph{Science of Computer Programming}}  \bibinfo{volume}{244} (\bibinfo{year}{2025}), \bibinfo{pages}{103280}.
\newblock


\bibitem[Galatos et~al\mbox{.}(2007)]%
        {galatos2007residuated}
\bibfield{author}{\bibinfo{person}{Nikolaos Galatos}, \bibinfo{person}{Peter Jipsen}, \bibinfo{person}{Tomasz Kowalski}, {and} \bibinfo{person}{Hiroakira Ono}.} \bibinfo{year}{2007}\natexlab{}.
\newblock \bibinfo{booktitle}{\emph{Residuated Lattices: An Algebraic Glimpse at Substructural Logics, Volume 151}}.
\newblock \bibinfo{publisher}{Elsevier Science}.
\newblock
\showISBNx{0444521410}


\bibitem[Girard(1995)]%
        {Gir95}
\bibfield{author}{\bibinfo{person}{Jean-Yves Girard}.} \bibinfo{year}{1995}\natexlab{}.
\newblock \showarticletitle{Linear logic: its syntax and semantics}. In \bibinfo{booktitle}{\emph{Proceedings of the Workshop on Advances in Linear Logic}}. \bibinfo{publisher}{Cambridge University Press}, \bibinfo{address}{USA}, \bibinfo{pages}{1–42}.
\newblock
\showISBNx{0521559618}


\bibitem[Giunchiglia et~al\mbox{.}(2022)]%
        {ijcai2022p767}
\bibfield{author}{\bibinfo{person}{Eleonora Giunchiglia}, \bibinfo{person}{Mihaela~Catalina Stoian}, {and} \bibinfo{person}{Thomas Lukasiewicz}.} \bibinfo{year}{2022}\natexlab{}.
\newblock \showarticletitle{Deep Learning with Logical Constraints}. In \bibinfo{booktitle}{\emph{31st International Joint Conference on Artificial Intelligence ({IJCAI-22})}}. \bibinfo{publisher}{International Joint Conferences on Artificial Intelligence Organization}, \bibinfo{pages}{5478--5485}.
\newblock
\href{https://doi.org/10.24963/ijcai.2022/767}{doi:\nolinkurl{10.24963/ijcai.2022/767}}
\newblock
\shownote{Survey Track}.


\bibitem[Goguen(1969)]%
        {goguenLogicInexactConcepts1969}
\bibfield{author}{\bibinfo{person}{J.~A. Goguen}.} \bibinfo{year}{1969}\natexlab{}.
\newblock \showarticletitle{The Logic of Inexact Concepts}.
\newblock \bibinfo{journal}{\emph{Synthese}} \bibinfo{volume}{19}, \bibinfo{number}{3} (\bibinfo{date}{April} \bibinfo{year}{1969}), \bibinfo{pages}{325--373}.
\newblock
\showISSN{1573-0964}
\href{https://doi.org/10.1007/BF00485654}{doi:\nolinkurl{10.1007/BF00485654}}


\bibitem[Grandis(2007)]%
        {grandisCategoriesNormsWeights2007}
\bibfield{author}{\bibinfo{person}{Marco Grandis}.} \bibinfo{year}{2007}\natexlab{}.
\newblock \showarticletitle{Categories, Norms and Weights}.
\newblock \bibinfo{journal}{\emph{Journal of Homotopy and Related Structures}}  \bibinfo{volume}{2} (\bibinfo{year}{2007}).
\newblock


\bibitem[Guglielmi(2007)]%
        {guglielmiSystemInteractionStructure2007}
\bibfield{author}{\bibinfo{person}{Alessio Guglielmi}.} \bibinfo{year}{2007}\natexlab{}.
\newblock \showarticletitle{A System of Interaction and Structure}.
\newblock \bibinfo{journal}{\emph{ACM Trans. Comput. Logic}} \bibinfo{volume}{8}, \bibinfo{number}{1} (\bibinfo{date}{Jan.} \bibinfo{year}{2007}), \bibinfo{pages}{1--es}.
\newblock
\showISSN{1529-3785}
\href{https://doi.org/10.1145/1182613.1182614}{doi:\nolinkurl{10.1145/1182613.1182614}}


\bibitem[Heiman et~al\mbox{.}(2026)]%
        {heimanNeuralNetworksFuzzy2026a}
\bibfield{author}{\bibinfo{person}{Damian Heiman}, \bibinfo{person}{Antti Kuusisto}, {and} \bibinfo{person}{Esko Turunen}.} \bibinfo{year}{2026}\natexlab{}.
\newblock \bibinfo{title}{Neural Networks as Fuzzy Logic Formulas}.
\newblock
\href{https://doi.org/10.48550/arXiv.2605.03064}{doi:\nolinkurl{10.48550/arXiv.2605.03064}}
\showeprint[arxiv]{2605.03064}~[cs.LO]


\bibitem[H{\"{o}}lldobler et~al\mbox{.}(1999)]%
        {DBLP:journals/apin/HolldoblerKS99}
\bibfield{author}{\bibinfo{person}{Steffen H{\"{o}}lldobler}, \bibinfo{person}{Yvonne Kalinke}, {and} \bibinfo{person}{Hans{-}Peter St{\"{o}}rr}.} \bibinfo{year}{1999}\natexlab{}.
\newblock \showarticletitle{Approximating the Semantics of Logic Programs by Recurrent Neural Networks}.
\newblock \bibinfo{journal}{\emph{Appl. Intell.}} \bibinfo{volume}{11}, \bibinfo{number}{1} (\bibinfo{year}{1999}), \bibinfo{pages}{45--58}.
\newblock
\href{https://doi.org/10.1023/A:1008376514077}{doi:\nolinkurl{10.1023/A:1008376514077}}


\bibitem[Horne(2019)]%
        {horneSubAdditivesProofTheory2019}
\bibfield{author}{\bibinfo{person}{Ross Horne}.} \bibinfo{year}{2019}\natexlab{}.
\newblock \showarticletitle{The {{Sub-Additives}}: {{A Proof Theory}} for {{Probabilistic Choice}} Extending {{Linear Logic}}}. In \bibinfo{booktitle}{\emph{{{LIPIcs}}, {{Volume}} 131, {{FSCD}} 2019}}, \bibfield{editor}{\bibinfo{person}{Herman Geuvers}} (Ed.), Vol.~\bibinfo{volume}{131}. \bibinfo{publisher}{Schloss Dagstuhl -- Leibniz-Zentrum f\"ur Informatik}, \bibinfo{pages}{23:1--23:16}.
\newblock
\showISBNx{978-3-95977-107-8}
\showISSN{1868-8969}
\href{https://doi.org/10.4230/LIPICS.FSCD.2019.23}{doi:\nolinkurl{10.4230/LIPICS.FSCD.2019.23}}


\bibitem[Katz et~al\mbox{.}(2019b)]%
        {katz2019marabou}
\bibfield{author}{\bibinfo{person}{Guy Katz}, \bibinfo{person}{Derek~A Huang}, \bibinfo{person}{Duligur Ibeling}, \bibinfo{person}{Kyle Julian}, \bibinfo{person}{Christopher Lazarus}, \bibinfo{person}{Rachel Lim}, \bibinfo{person}{Parth Shah}, \bibinfo{person}{Shantanu Thakoor}, \bibinfo{person}{Haoze Wu}, \bibinfo{person}{Aleksandar Zelji{\'c}}, {et~al\mbox{.}}} \bibinfo{year}{2019}\natexlab{b}.
\newblock \showarticletitle{The marabou framework for verification and analysis of deep neural networks}. In \bibinfo{booktitle}{\emph{Computer Aided Verification: 31st International Conference, CAV 2019, New York City, NY, USA, July 15-18, 2019, Proceedings, Part I 31}}. Springer, \bibinfo{pages}{443--452}.
\newblock


\bibitem[Katz et~al\mbox{.}(2019a)]%
        {katzMarabouFrameworkVerification2019}
\bibfield{author}{\bibinfo{person}{Guy Katz}, \bibinfo{person}{Derek~A. Huang}, \bibinfo{person}{Duligur Ibeling}, \bibinfo{person}{Kyle Julian}, \bibinfo{person}{Christopher Lazarus}, \bibinfo{person}{Rachel Lim}, \bibinfo{person}{Parth Shah}, \bibinfo{person}{Shantanu Thakoor}, \bibinfo{person}{Haoze Wu}, \bibinfo{person}{Aleksandar Zelji{\'c}}, \bibinfo{person}{David~L. Dill}, \bibinfo{person}{Mykel~J. Kochenderfer}, {and} \bibinfo{person}{Clark Barrett}.} \bibinfo{year}{2019}\natexlab{a}.
\newblock \showarticletitle{The {{Marabou Framework}} for {{Verification}} and {{Analysis}} of {{Deep Neural Networks}}}. In \bibinfo{booktitle}{\emph{Computer {{Aided Verification}}}} \emph{(\bibinfo{series}{Lecture {{Notes}} in {{Computer Science}}})}, \bibfield{editor}{\bibinfo{person}{Isil Dillig} {and} \bibinfo{person}{Serdar Tasiran}} (Eds.). \bibinfo{publisher}{Springer International Publishing}, \bibinfo{address}{Cham}, \bibinfo{pages}{443--452}.
\newblock
\showISBNx{978-3-030-25540-4}
\href{https://doi.org/10.1007/978-3-030-25540-4_26}{doi:\nolinkurl{10.1007/978-3-030-25540-4_26}}


\bibitem[Kelly(1982)]%
        {kellyBasicConceptsEnriched1982}
\bibfield{author}{\bibinfo{person}{Max Kelly}.} \bibinfo{year}{1982}\natexlab{}.
\newblock \bibinfo{booktitle}{\emph{Basic Concepts of Enriched Category Theory}}. Vol.~\bibinfo{volume}{64}.
\newblock \bibinfo{publisher}{CUP Archive}.
\newblock


\bibitem[Komendantskaya et~al\mbox{.}(2026a)]%
        {komendantskayaSpecificationAllYou2026}
\bibfield{author}{\bibinfo{person}{Ekaterina Komendantskaya}, \bibinfo{person}{Marco Casadio}, \bibinfo{person}{Jessica Ciupa}, \bibinfo{person}{Thomas Flinkow}, \bibinfo{person}{Joe Hall}, \bibinfo{person}{Colin Kessler}, {and} \bibinfo{person}{Alistair Sirman}.} \bibinfo{year}{2026}\natexlab{a}.
\newblock \showarticletitle{Specification is All You Need}. In \bibinfo{booktitle}{\emph{Formal Methods and Software Engineering: 27th International Conference on Formal Engineering Methods, {ICFEM} 2026, Southampton, UK, November 17--20, 2026}} \emph{(\bibinfo{series}{Lecture Notes in Computer Science})}. \bibinfo{publisher}{Springer}.
\newblock
\urldef\tempurl%
\url{https://www.macs.hw.ac.uk/~ek19/ICFEM.pdf}
\showURL{%
\tempurl}
\newblock
\shownote{Invited paper. To appear}.


\bibitem[Komendantskaya et~al\mbox{.}(2026b)]%
        {vehicleTutorial2026}
\bibfield{author}{\bibinfo{person}{Ekaterina Komendantskaya}, \bibinfo{person}{Matthew Daggitt}, \bibinfo{person}{Wen Kokke}, \bibinfo{person}{Alistair Sirman}, \bibinfo{person}{Joe Hall}, {and} \bibinfo{person}{Aashish Misra}.} \bibinfo{year}{2026}\natexlab{b}.
\newblock \bibinfo{booktitle}{\emph{The {Vehicle} Tutorial: Neural Network Verification with {Vehicle}}}.
\newblock
\urldef\tempurl%
\url{https://vehicle-lang.github.io/tutorial/}
\showURL{%
Retrieved October 25, 2026 from \tempurl}


\bibitem[Komendantskaya et~al\mbox{.}(2007)]%
        {DBLP:series/sci/KomendantskayaLS07}
\bibfield{author}{\bibinfo{person}{Ekaterina Komendantskaya}, \bibinfo{person}{M{\'{a}}ire Lane}, {and} \bibinfo{person}{Anthony~Karel Seda}.} \bibinfo{year}{2007}\natexlab{}.
\newblock \showarticletitle{Connectionist Representation of Multi-Valued Logic Programs}.
\newblock In \bibinfo{booktitle}{\emph{Perspectives of Neural-Symbolic Integration}}, \bibfield{editor}{\bibinfo{person}{Barbara Hammer} {and} \bibinfo{person}{Pascal Hitzler}} (Eds.). \bibinfo{series}{Studies in Computational Intelligence}, Vol.~\bibinfo{volume}{77}. \bibinfo{publisher}{Springer}, \bibinfo{pages}{283--313}.
\newblock
\href{https://doi.org/10.1007/978-3-540-73954-8\_12}{doi:\nolinkurl{10.1007/978-3-540-73954-8\_12}}


\bibitem[Laird et~al\mbox{.}(2013)]%
        {laird2013weighted}
\bibfield{author}{\bibinfo{person}{Jim Laird}, \bibinfo{person}{Giulio Manzonetto}, \bibinfo{person}{Guy McCusker}, {and} \bibinfo{person}{Michele Pagani}.} \bibinfo{year}{2013}\natexlab{}.
\newblock \showarticletitle{Weighted {{Relational Models}} of {{Typed Lambda-Calculi}}}. In \bibinfo{booktitle}{\emph{2013 28th {{Annual ACM}}/{{IEEE Symposium}} on {{Logic}} in {{Computer Science}}}}. \bibinfo{publisher}{IEEE}, \bibinfo{pages}{301--310}.
\newblock
\href{https://doi.org/10.1109/LICS.2013.36}{doi:\nolinkurl{10.1109/LICS.2013.36}}


\bibitem[Laurent(2026)]%
        {Laurent26}
\bibfield{author}{\bibinfo{person}{Olivier Laurent}.} \bibinfo{year}{2026}\natexlab{}.
\newblock \showarticletitle{{YALLA: Yet Another Deep Embedding of Linear Logic in Rocq}}.
\newblock \bibinfo{journal}{\emph{Journal of Automated Reasoning}} \bibinfo{volume}{70}, \bibinfo{number}{1} (\bibinfo{year}{2026}), \bibinfo{pages}{9}.
\newblock
\href{https://doi.org/10.1007/S10817-026-09755-Y}{doi:\nolinkurl{10.1007/S10817-026-09755-Y}}


\bibitem[Lawvere(1969)]%
        {lawvereAdjointnessFoundations1969}
\bibfield{author}{\bibinfo{person}{F.~William Lawvere}.} \bibinfo{year}{1969}\natexlab{}.
\newblock \showarticletitle{Adjointness in Foundations}.
\newblock \bibinfo{journal}{\emph{Dialectica}} \bibinfo{volume}{23}, \bibinfo{number}{3/4} (\bibinfo{year}{1969}), \bibinfo{pages}{281--296}.
\newblock
\href{https://doi.org/10.1111/j.1746-8361.1969.tb01194.x}{doi:\nolinkurl{10.1111/j.1746-8361.1969.tb01194.x}}


\bibitem[Lawvere(1973)]%
        {lawvereMetricSpacesGeneralized1973}
\bibfield{author}{\bibinfo{person}{F.~William Lawvere}.} \bibinfo{year}{1973}\natexlab{}.
\newblock \showarticletitle{Metric Spaces, Generalized Logic, and Closed Categories}.
\newblock \bibinfo{journal}{\emph{Rendiconti del seminario mat\'ematico e fisico di Milano}}  \bibinfo{volume}{43} (\bibinfo{year}{1973}), \bibinfo{pages}{135--166}.
\newblock


\bibitem[Leinster(2012)]%
        {leinsterMultiplicativeCharacterizationPower2012}
\bibfield{author}{\bibinfo{person}{Tom Leinster}.} \bibinfo{year}{2012}\natexlab{}.
\newblock \showarticletitle{A Multiplicative Characterization of the Power Means}.
\newblock \bibinfo{journal}{\emph{Bulletin of the London Mathematical Society}} \bibinfo{volume}{44}, \bibinfo{number}{1} (\bibinfo{date}{Feb.} \bibinfo{year}{2012}), \bibinfo{pages}{106--112}.
\newblock
\showISSN{00246093}
\href{https://doi.org/10.1112/blms/bdr073}{doi:\nolinkurl{10.1112/blms/bdr073}}
\showeprint{1103.2574}~[cs, math]


\bibitem[Li et~al\mbox{.}(2023)]%
        {liScallopLanguageNeurosymbolic2023}
\bibfield{author}{\bibinfo{person}{Ziyang Li}, \bibinfo{person}{Jiani Huang}, {and} \bibinfo{person}{Mayur Naik}.} \bibinfo{year}{2023}\natexlab{}.
\newblock \showarticletitle{Scallop: A Language for Neurosymbolic Programming}.
\newblock \bibinfo{journal}{\emph{Proceedings of the ACM on Programming Languages}} \bibinfo{volume}{7}, \bibinfo{number}{PLDI}, Article \bibinfo{articleno}{166} (\bibinfo{date}{June} \bibinfo{year}{2023}), \bibinfo{numpages}{1463--1487}~pages.
\newblock
\href{https://doi.org/10.1145/3591280}{doi:\nolinkurl{10.1145/3591280}}


\bibitem[Manhaeve et~al\mbox{.}(2018)]%
        {manhaeveDeepProbLogNeuralProbabilistic2018}
\bibfield{author}{\bibinfo{person}{Robin Manhaeve}, \bibinfo{person}{Sebastijan Duman{\v{c}}i{\'{c}}}, \bibinfo{person}{Angelika Kimmig}, \bibinfo{person}{Thomas Demeester}, {and} \bibinfo{person}{Luc De~Raedt}.} \bibinfo{year}{2018}\natexlab{}.
\newblock \showarticletitle{{DeepProbLog}: Neural Probabilistic Logic Programming}. In \bibinfo{booktitle}{\emph{Advances in Neural Information Processing Systems 31 (NeurIPS 2018)}} \emph{(\bibinfo{series}{Advances in Neural Information Processing Systems}, Vol.~\bibinfo{volume}{31})}, \bibfield{editor}{\bibinfo{person}{S.~Bengio}, \bibinfo{person}{H.~Wallach}, \bibinfo{person}{H.~Larochelle}, \bibinfo{person}{K.~Grauman}, \bibinfo{person}{N.~Cesa-Bianchi}, {and} \bibinfo{person}{R.~Garnett}} (Eds.). \bibinfo{publisher}{Curran Associates, Inc.}, \bibinfo{pages}{3753--3763}.
\newblock
\urldef\tempurl%
\url{https://proceedings.neurips.cc/paper_files/paper/2018/hash/dc5d637ed5e62c36ecb73b654b05ba2a-Abstract.html}
\showURL{%
\tempurl}


\bibitem[Mardare et~al\mbox{.}(2016)]%
        {mardareQuantitativeAlgebraicReasoning2016}
\bibfield{author}{\bibinfo{person}{Radu Mardare}, \bibinfo{person}{Prakash Panangaden}, {and} \bibinfo{person}{Gordon Plotkin}.} \bibinfo{year}{2016}\natexlab{}.
\newblock \showarticletitle{Quantitative {{Algebraic Reasoning}}}. In \bibinfo{booktitle}{\emph{Proceedings of the 31st {{Annual ACM}}/{{IEEE Symposium}} on {{Logic}} in {{Computer Science}}}}. \bibinfo{publisher}{ACM}, \bibinfo{address}{New York NY USA}, \bibinfo{pages}{700--709}.
\newblock
\showISBNx{978-1-4503-4391-6}
\href{https://doi.org/10.1145/2933575.2934518}{doi:\nolinkurl{10.1145/2933575.2934518}}


\bibitem[{Martin-L{\"o}f}(1996)]%
        {martin-lofMeaningsLogicalConstants1996}
\bibfield{author}{\bibinfo{person}{Per {Martin-L{\"o}f}}.} \bibinfo{year}{1996}\natexlab{}.
\newblock \showarticletitle{On the Meanings of the Logical Constants and the Justifications of the Logical Laws}.
\newblock \bibinfo{journal}{\emph{Nordic Journal of Philosophical Logic}} \bibinfo{volume}{1}, \bibinfo{number}{1} (\bibinfo{year}{1996}), \bibinfo{pages}{11--60}.
\newblock


\bibitem[McCulloch and Pitts(1943)]%
        {mccullochLogicalCalculusIdeas1943}
\bibfield{author}{\bibinfo{person}{Warren~S. McCulloch} {and} \bibinfo{person}{Walter Pitts}.} \bibinfo{year}{1943}\natexlab{}.
\newblock \showarticletitle{A Logical Calculus of the Ideas Immanent in Nervous Activity}.
\newblock \bibinfo{journal}{\emph{Bulletin of Mathematical Biophysics}} \bibinfo{volume}{5}, \bibinfo{number}{4} (\bibinfo{year}{1943}), \bibinfo{pages}{115--133}.
\newblock
\href{https://doi.org/10.1007/BF02478259}{doi:\nolinkurl{10.1007/BF02478259}}


\bibitem[Melliès(2009)]%
        {mellies:categorical-sem-LL}
\bibfield{author}{\bibinfo{person}{Paul-André Melliès}.} \bibinfo{year}{2009}\natexlab{}.
\newblock \showarticletitle{Categorical semantics of linear logic}. In \bibinfo{booktitle}{\emph{Interactive models of computation and program behaviour, panoramas et synthèses}}, Vol.~\bibinfo{volume}{27}. \bibinfo{publisher}{Société Mathématique de France}, \bibinfo{pages}{1--196}.
\newblock


\bibitem[Metcalfe et~al\mbox{.}(2009)]%
        {metcalfeProofTheoryFuzzy2009}
\bibfield{author}{\bibinfo{person}{George Metcalfe}, \bibinfo{person}{Nicola Olivetti}, {and} \bibinfo{person}{Dov Gabbay}.} \bibinfo{year}{2009}\natexlab{}.
\newblock \bibinfo{booktitle}{\emph{Proof {{Theory}} for {{Fuzzy Logics}}}}. \bibinfo{series}{Applied {{Logic Series}}}, Vol.~\bibinfo{volume}{36}.
\newblock \bibinfo{publisher}{Springer Netherlands}, \bibinfo{address}{Dordrecht}.
\newblock
\showISBNx{978-1-4020-9408-8 978-1-4020-9409-5}
\href{https://doi.org/10.1007/978-1-4020-9409-5}{doi:\nolinkurl{10.1007/978-1-4020-9409-5}}


\bibitem[Mitrinovic(1970)]%
        {mitrinovicAnalyticInequalities1970}
\bibfield{author}{\bibinfo{person}{D.~S. Mitrinovic}.} \bibinfo{year}{1970}\natexlab{}.
\newblock \bibinfo{booktitle}{\emph{Analytic Inequalities}}.
\newblock \bibinfo{publisher}{Springer-Verlag}.
\newblock


\bibitem[Mosbach et~al\mbox{.}(2019)]%
        {mosbachLogitPairingMethods2019}
\bibfield{author}{\bibinfo{person}{Marius Mosbach}, \bibinfo{person}{Maksym Andriushchenko}, \bibinfo{person}{Thomas Trost}, \bibinfo{person}{Matthias Hein}, {and} \bibinfo{person}{Dietrich Klakow}.} \bibinfo{year}{2019}\natexlab{}.
\newblock \bibinfo{title}{Logit {{Pairing Methods Can Fool Gradient-Based Attacks}}}.
\newblock
\href{https://doi.org/10.48550/arXiv.1810.12042}{doi:\nolinkurl{10.48550/arXiv.1810.12042}}
\showeprint[arxiv]{1810.12042}~[cs]


\bibitem[Pavelka(1979a)]%
        {pavelkaFuzzyLogicManyvalued1979}
\bibfield{author}{\bibinfo{person}{Jan Pavelka}.} \bibinfo{year}{1979}\natexlab{a}.
\newblock \showarticletitle{On {{Fuzzy Logic I Many}}-Valued Rules of Inference}.
\newblock \bibinfo{journal}{\emph{Mathematical Logic Quarterly}} \bibinfo{volume}{25}, \bibinfo{number}{3-6} (\bibinfo{year}{1979}), \bibinfo{pages}{45--52}.
\newblock
\showISSN{1521-3870}
\href{https://doi.org/10.1002/malq.19790250304}{doi:\nolinkurl{10.1002/malq.19790250304}}


\bibitem[Pavelka(1979b)]%
        {pavelkaFuzzyLogicII1979}
\bibfield{author}{\bibinfo{person}{Jan Pavelka}.} \bibinfo{year}{1979}\natexlab{b}.
\newblock \showarticletitle{On {{Fuzzy Logic II}}. {{Enriched}} Residuated Lattices and Semantics of Propositional Calculi}.
\newblock \bibinfo{journal}{\emph{Mathematical Logic Quarterly}} \bibinfo{volume}{25}, \bibinfo{number}{7-12} (\bibinfo{year}{1979}), \bibinfo{pages}{119--134}.
\newblock
\showISSN{1521-3870}
\href{https://doi.org/10.1002/malq.19790250706}{doi:\nolinkurl{10.1002/malq.19790250706}}


\bibitem[Pavelka(1979c)]%
        {pavelkaFuzzyLogicIII1979}
\bibfield{author}{\bibinfo{person}{Jan Pavelka}.} \bibinfo{year}{1979}\natexlab{c}.
\newblock \showarticletitle{On {{Fuzzy Logic III}}. {{Semantical}} Completeness of Some Many-Valued Propositional Calculi}.
\newblock \bibinfo{journal}{\emph{Mathematical Logic Quarterly}} \bibinfo{volume}{25}, \bibinfo{number}{25-29} (\bibinfo{year}{1979}), \bibinfo{pages}{447--464}.
\newblock
\showISSN{1521-3870}
\href{https://doi.org/10.1002/malq.19790252510}{doi:\nolinkurl{10.1002/malq.19790252510}}


\bibitem[Perovi{\'c} et~al\mbox{.}(2025)]%
        {perovicLogicsMostCountable2025}
\bibfield{author}{\bibinfo{person}{Aleksandar Perovi{\'c}}, \bibinfo{person}{Zoran Ognjanovi{\'c}}, {and} \bibinfo{person}{Tatjana Stojanovi{\'c}}.} \bibinfo{year}{2025}\natexlab{}.
\newblock \showarticletitle{Logics for at Most Countable First-Order Structures}.
\newblock \bibinfo{journal}{\emph{Journal of Logic and Computation}} \bibinfo{volume}{35}, \bibinfo{number}{8} (\bibinfo{year}{2025}), \bibinfo{pages}{exae067}.
\newblock
\href{https://doi.org/10.1093/logcom/exae067}{doi:\nolinkurl{10.1093/logcom/exae067}}


\bibitem[Pottinger(1983)]%
        {pottinger1983uniform}
\bibfield{author}{\bibinfo{person}{Garrel Pottinger}.} \bibinfo{year}{1983}\natexlab{}.
\newblock \showarticletitle{Uniform, cut-free formulations of T, S4 and S5}.
\newblock \bibinfo{journal}{\emph{Journal of Symbolic Logic}} \bibinfo{volume}{48}, \bibinfo{number}{3} (\bibinfo{year}{1983}), \bibinfo{pages}{900}.
\newblock


\bibitem[project(2023)]%
        {LLHandbook}
\bibfield{author}{\bibinfo{person}{The~LLHandbook project}.} \bibinfo{year}{2023}\natexlab{}.
\newblock \bibinfo{title}{Handbook of Linear Logic}.
\newblock \bibinfo{howpublished}{online draft}.
\newblock


\bibitem[{van Krieken} et~al\mbox{.}(2022)]%
        {van2022analyzing}
\bibfield{author}{\bibinfo{person}{Emile {van Krieken}}, \bibinfo{person}{Erman Acar}, {and} \bibinfo{person}{Frank {van Harmelen}}.} \bibinfo{year}{2022}\natexlab{}.
\newblock \showarticletitle{Analyzing Differentiable Fuzzy Logic Operators}.
\newblock \bibinfo{journal}{\emph{Artificial Intelligence}}  \bibinfo{volume}{302} (\bibinfo{date}{Jan.} \bibinfo{year}{2022}), \bibinfo{pages}{103602}.
\newblock
\showISSN{0004-3702}
\href{https://doi.org/10.1016/j.artint.2021.103602}{doi:\nolinkurl{10.1016/j.artint.2021.103602}}
\showeprint[arxiv]{2002.06100}~[cs.AI]


\bibitem[Varnai and Dimarogonas(2020)]%
        {varnaiRobustnessMetricsLearning2020}
\bibfield{author}{\bibinfo{person}{Peter Varnai} {and} \bibinfo{person}{Dimos~V. Dimarogonas}.} \bibinfo{year}{2020}\natexlab{}.
\newblock \showarticletitle{On {{Robustness Metrics}} for {{Learning STL Tasks}}}. In \bibinfo{booktitle}{\emph{2020 {{American Control Conference}} ({{ACC}})}}. \bibinfo{pages}{5394--5399}.
\newblock
\showISSN{2378-5861}
\href{https://doi.org/10.23919/ACC45564.2020.9147692}{doi:\nolinkurl{10.23919/ACC45564.2020.9147692}}


\bibitem[Vychodil(2015)]%
        {vychodilFuzzyInequationalLogic2015}
\bibfield{author}{\bibinfo{person}{Vilem Vychodil}.} \bibinfo{year}{2015}\natexlab{}.
\newblock \showarticletitle{Fuzzy Inequational Logic}.
\newblock \bibinfo{journal}{\emph{International Journal of Approximate Reasoning}}  \bibinfo{volume}{60} (\bibinfo{date}{May} \bibinfo{year}{2015}), \bibinfo{pages}{23--36}.
\newblock
\showISSN{0888613X}
\href{https://doi.org/10.1016/j.ijar.2015.03.003}{doi:\nolinkurl{10.1016/j.ijar.2015.03.003}}


\bibitem[Wu et~al\mbox{.}(2024)]%
        {marabou2}
\bibfield{author}{\bibinfo{person}{Haoze Wu}, \bibinfo{person}{Omri Isac}, \bibinfo{person}{Aleksandar Zelji{\'c}}, \bibinfo{person}{Teruhiro Tagomori}, \bibinfo{person}{Matthew Daggitt}, \bibinfo{person}{Wen Kokke}, \bibinfo{person}{Idan Refaeli}, \bibinfo{person}{Guy Amir}, \bibinfo{person}{Kyle Julian}, \bibinfo{person}{Shahaf Bassan}, \bibinfo{person}{Pei Huang}, \bibinfo{person}{Ori Lahav}, \bibinfo{person}{Min Wu}, \bibinfo{person}{Min Zhang}, \bibinfo{person}{Ekaterina Komendantskaya}, \bibinfo{person}{Guy Katz}, {and} \bibinfo{person}{Clark Barrett}.} \bibinfo{year}{2024}\natexlab{}.
\newblock \showarticletitle{Marabou 2.0: {{A Versatile Formal Analyzer}} of~{{Neural Networks}}}. In \bibinfo{booktitle}{\emph{Computer {{Aided Verification}}}}, \bibfield{editor}{\bibinfo{person}{Arie Gurfinkel} {and} \bibinfo{person}{Vijay Ganesh}} (Eds.). \bibinfo{publisher}{Springer Nature Switzerland}, \bibinfo{address}{Cham}, \bibinfo{pages}{249--264}.
\newblock
\showISBNx{978-3-031-65630-9}
\href{https://doi.org/10.1007/978-3-031-65630-9_13}{doi:\nolinkurl{10.1007/978-3-031-65630-9_13}}


\bibitem[Yaacov et~al\mbox{.}(2008a)]%
        {yaacovModelTheoryMetric2008}
\bibfield{author}{\bibinfo{person}{I.~Ben Yaacov}, \bibinfo{person}{Alexander Berenstein}, \bibinfo{person}{C.~Ward Henson}, {and} \bibinfo{person}{Alexander Usvyatsov}.} \bibinfo{year}{2008}\natexlab{a}.
\newblock \showarticletitle{Model Theory for Metric Structures}.
\newblock \bibinfo{journal}{\emph{London Mathematical Society Lecture Note Series}}  \bibinfo{volume}{350} (\bibinfo{year}{2008}), \bibinfo{pages}{315}.
\newblock
\urldef\tempurl%
\url{http://math.univ-lyon1.fr/homes-www/begnac/articles/mtfms.pdf}
\showURL{%
\tempurl}


\bibitem[Yaacov et~al\mbox{.}(2008b)]%
        {yaacovModelTheoryMetric2008a}
\bibfield{author}{\bibinfo{person}{Ita{\"i}~Ben Yaacov}, \bibinfo{person}{Alexander Berenstein}, \bibinfo{person}{C.~Ward Henson}, {and} \bibinfo{person}{Alexander Usvyatsov}.} \bibinfo{year}{2008}\natexlab{b}.
\newblock \showarticletitle{Model Theory for Metric Structures}.
\newblock In \bibinfo{booktitle}{\emph{Model {{Theory}} with {{Applications}} to {{Algebra}} and {{Analysis}}}}, \bibfield{editor}{\bibinfo{person}{Alex Wilkie}, \bibinfo{person}{Anand Pillay}, \bibinfo{person}{Dugald Macpherson}, {and} \bibinfo{person}{Zo{\'e} Chatzidakis}} (Eds.). \bibinfo{series}{London {{Mathematical Society Lecture Note Series}}}, Vol.~\bibinfo{volume}{2}. \bibinfo{publisher}{Cambridge University Press}, \bibinfo{address}{Cambridge}, \bibinfo{pages}{315--427}.
\newblock
\showISBNx{978-0-521-70908-8}
\href{https://doi.org/10.1017/CBO9780511735219.011}{doi:\nolinkurl{10.1017/CBO9780511735219.011}}


\bibitem[Yager(1980)]%
        {yagerGeneralClassFuzzy1980}
\bibfield{author}{\bibinfo{person}{Ronald~R. Yager}.} \bibinfo{year}{1980}\natexlab{}.
\newblock \showarticletitle{On a General Class of Fuzzy Connectives}.
\newblock \bibinfo{journal}{\emph{Fuzzy Sets and Systems}} \bibinfo{volume}{4}, \bibinfo{number}{3} (\bibinfo{year}{1980}), \bibinfo{pages}{235--242}.
\newblock
\showISSN{0165-0114}
\href{https://doi.org/10.1016/0165-0114(80)90013-5}{doi:\nolinkurl{10.1016/0165-0114(80)90013-5}}


\bibitem[Yuviler and Drachsler-Cohen(2025)]%
        {yuviler2025enhancing}
\bibfield{author}{\bibinfo{person}{Tom Yuviler} {and} \bibinfo{person}{Dana Drachsler-Cohen}.} \bibinfo{year}{2025}\natexlab{}.
\newblock \showarticletitle{Enhancing Neural Network Robustness via Synthesis of Repair Programs}. In \bibinfo{booktitle}{\emph{Static Analysis (SAS 2025)}} \emph{(\bibinfo{series}{Lecture Notes in Computer Science}, Vol.~\bibinfo{volume}{16100})}, \bibfield{editor}{\bibinfo{person}{Hakjoo Oh} {and} \bibinfo{person}{Yulei Sui}} (Eds.). \bibinfo{publisher}{Springer}, \bibinfo{pages}{221--248}.
\newblock
\href{https://doi.org/10.1007/978-3-032-07106-4_10}{doi:\nolinkurl{10.1007/978-3-032-07106-4_10}}


\bibitem[Zadeh(1965)]%
        {zadehFuzzySets1965}
\bibfield{author}{\bibinfo{person}{Lotfi~Asker Zadeh}.} \bibinfo{year}{1965}\natexlab{}.
\newblock \showarticletitle{Fuzzy Sets}.
\newblock \bibinfo{journal}{\emph{Information and control}} \bibinfo{volume}{8}, \bibinfo{number}{3} (\bibinfo{year}{1965}), \bibinfo{pages}{338--353}.
\newblock
\urldef\tempurl%
\url{https://www.sciencedirect.com/science/article/pii/S001999586590241X}
\showURL{%
\tempurl}


\bibitem[Zadeh(1988)]%
        {zadehFuzzyLogic1988}
\bibfield{author}{\bibinfo{person}{Lotfi~Asker Zadeh}.} \bibinfo{year}{1988}\natexlab{}.
\newblock \showarticletitle{Fuzzy Logic}.
\newblock \bibinfo{journal}{\emph{Computer}} \bibinfo{volume}{21}, \bibinfo{number}{4} (\bibinfo{year}{1988}), \bibinfo{pages}{83--93}.
\newblock
\urldef\tempurl%
\url{https://ieeexplore.ieee.org/abstract/document/53/}
\showURL{%
\tempurl}


\bibitem[Łukasiewicz(1920)]%
        {lukasiewicz1920three}
\bibfield{author}{\bibinfo{person}{J Łukasiewicz}.} \bibinfo{year}{1920}\natexlab{}.
\newblock \bibinfo{booktitle}{\emph{O logice trójwartościowej (in Polish). English translation: On Three-Valued Logic, in Borkowski, L.(ed.) 1970. Jan Łukasiewicz: Selected Works, Amsterdam: North Holland}}.
\newblock \bibinfo{publisher}{Ruch Filozoficzny}. 87–--88 pages.
\newblock
\showISBNx{0-7204-2252-3}


\end{thebibliography}

\pagebreak
\appendix

\section{Algebraic lemmata}
\label[appendix]{app:algebra}



\begin{lemma}[Arithmetic of generalized fractions]
\label{lemma:simp}
	For all $a,b,c,d \in \PosReals$:
	\begin{enumerate}
		\item $1 \leq a \mulimp a$,
		\item $(b \mulimp a) \tensor (c \mulimp b) \leq (c \mulimp a)$,
		\item $(b \mulimp a) \tensor (d \mulimp c) \leq (b \tensor d) \mulimp (a \tensor c)$,
		\item $a \mulimp (b \mulimp c) = (a \tensor b) \mulimp c$.
	\end{enumerate}
\end{lemma}
\begin{proof}
	For (1), start from $1 \tensor a \leq a^{**}$ and apply the defining property of $(-)^*$ to get $1 \leq (a \tensor a^*)^* = (a^* \cotensor a) = {a \mulimp a}$.
	Claim (2)--(4)  hold in any symmetric monoidal closed category.
\end{proof}

\begin{remark}
	Note the above inequalities are all equalities as soon as the numbers involved are finite and non-zero.
\end{remark}

\begin{lemma}
\label{lemma:max-prod-ineq}
	Let $a,b,c,d \in \PosReals$, $0 < p \leq \infty$. Then
	\begin{equation}
		(a \padd[p] b) \tensor (c \pcoadd[p] d) \leq (a \tensor c) \lor (b \tensor d).
	\end{equation}
\end{lemma}
\begin{proof}
	For simplicity, and without risk of ambiguity, below we write $\tensor$ as juxtaposition, i.e. we default to conjunctive multiplication as our multiplication of choice.
	We also assume, without loss of generality, that $ac = ac \lor bd$.
	We proceed by case disjunction.

	If $a,b,c,d \in (0,\infty)$, then:
	\begin{eqalign}
		\ & bd \leq ac\\
		\iff\ & b^p d^p \leq a^p c^p \\
		\iff\ & b^p \leq a^p c^p d^{-p} \\
		\iff\ & a^p + b^p \leq a^p + a^p c^p d^{-p} \\
		\iff\ & a^p + b^p \leq a^p c^p \left(c^{-p} + d^{-p} \right) \\
		\iff\ & \left(a^p + b^p\right) \left(c^{-p} + d^{-p}\right)^{-1} \leq a^p c^p \\
		\iff\ & \left(a^p + b^p\right)^{1/p} \left(c^{-p} + d^{-p}\right)^{-1/p} \leq ac \\
		\iff\ & (a \padd[p] b) (c \pcoadd[p] d) \leq a c \\
	\end{eqalign}

	If $a = 0$ or $c = 0$, since $ac \geq bd$, it follows that $b = 0$ or $d = 0$. We need to prove that $(a \padd[p] b) (c \pcoadd[p] d) = 0$.
	\begin{itemize}
		\item If $a = 0$ and $b = 0$ then $a \padd[p] b = 0$ and 0 is absorbing for $\tensor$.
		\item If $a = 0$ and $d = 0$ then $c \pcoadd[p] d = 0$ as 0 is absorbing for $\pcoadd[p]$. Then we conclude using that $0$ is absorbing for $\tensor$.
	\end{itemize}
	By symmetry, the same proof works for $c = 0$.

	If $b = 0$ or $d = 0$ then either $a \padd[p] b = 0$ or $c \pcoadd[p] d = 0$, so that $(a \padd[p] b) \tensor (c \pcoadd[p] d) = 0$.

	If $a = \infty$ or $c = \infty$, since $a = 0$ and $c = 0$ have already been treated, we consider here that either $a = \infty$ and  $c > 0$ ($\infty$ included) or the converse way.
	This rules out $0 \tensor \infty = 0$, and forces $ac = \infty$, from which we conclude.

	If $b = \infty$ or $d = \infty$, the case $d = 0$ has already been treated. So we can assume $d > 0$, but then $bd = \infty$. Since $ac \geq bd$ by hypothesis, $ac = \infty$ and the statement holds.
\end{proof}



\section{\texorpdfstring{Cut-elimination}{Cut-elimination}}
\label[appendix]{app:cut-elim}

As \pQLL's closest relative is multiplicative-additive linear logic~\cite{Gir95}, our cut-elimination proof resembles the cut-elimination proof given for linear logic in~\cite{mellies:categorical-sem-LL} by Melliès.
Note that Melliès' presentation is for \emph{intuitionistic} linear logic, so that a sequent always has a single formula as consequence. We have carefully adapted the argument to \pQLL, in which multiple formulas are allowed in the conclusion of a sequent.
The Linear Logic Handbook~\cite{LLHandbook}, in which another cut-elimination proof is provided for classical linear logic (but in a single-sided sequent style), has been an inspiration for this adaptation in several cases.

\begin{figure*}
	\centering
	\footnotesize{
		\begin{tikzpicture}[
				every node/.style={font=\small\sffamily},
				block/.style={text width=3.4cm, align=center}
			]
			\node [block] (rule) {Rules preceding \CutRule};
			\node [block, below left=0.7cm and -0.5cm of rule] (struct) {\emph{Structural Rules}\\(\cref{sec:cutting-structural-rules}):\\(\AxRule\ref{par:CvsA}); (\ExFalsoRule~\ref{par:CvsEFQ}); (\MixRule \ref{par:CvsM}); (\EmptyRule \ref{par:CvsE})};
			\node [block, below right=0.7cm and -0.5cm of rule] (other) {\emph{Logical Rules}\\
			(\cref{sec:principal,sec:non-principal})};
			\node [block, below left=0.7cm and -0.5cm of other] (principal) {Principal vs Principal\\ (\cref{sec:principal})\\ ($\tensor$:
				\ref{sec:tensor}), ($\parr$: \ref{sec:cotensor}), ($\plor$: \ref{sec:sdis}), ($\pland$: \ref{sec:scon}), ($(-)^\dual$: \ref{sec:neg}), ($\One$: \ref{sec:mult-units}), ($\bot/\top$: \ref{sec:add-units})};
			\node [block, below right=0.7cm and -0.7cm of other] (non-principal) {Non-principal vs principal\\(\cref{sec:non-principal})};
			\node [block, below left=1.3cm and -1.5cm of non-principal] (hyp) {Principal in Antecedent of Conclusion\\ (\cref{sec:prihyp})};
			\node [block, below right=1.3cm and -1.5cm of non-principal] (conc) {Principal in Consequent of Conclusion\\ (\cref{sec:pricon})};
			\draw[->] (rule.south) to (struct.north);
			\draw[->] (rule.south) to (other.north);
			\draw[->] (other.south) to (principal.north);
			\draw[->] (other.south) to (non-principal.north);
			\draw[->] (non-principal.south) to (hyp.north);
			\draw[->] (non-principal.south) to (conc.north);

		\end{tikzpicture}}
	\caption{The structure of the cut-elimination cases, with active hyperreferences to the relevant subsections.}
	\label{fig:cut-structure}
\end{figure*}

\subsection{The General Setup}
\label[appendix]{sec:general}
Consider a proof $\pi$ ending in \CutRule:
\begin{equation}
	\label{eq:exemplar-cut}
	\begin{prooftree}
		\hypo{}
		\ellipsis{$\pi_1$}{}
		\infer1[$R_1$]{\Gamma \entails A, \Delta}
		\hypo{\red{\tensor}}
		\hypo{}
		\ellipsis{$\pi_2$}{}
		\infer1[$R_2$]{\Gamma', A \entails \Delta'}
		\infer3[\CutRule]{\Gamma, \Gamma' \entails \Delta, \Delta'}
	\end{prooftree}
\end{equation}
Here $A$ is called the \textbf{cut formula}, and the formulae appearing in the other cedents of the premises are thus \textbf{non-cut}.
Since \pQLL is counary (recall \cref{sec:QLL}), we know $\pi$ does indeed decompose into two proofs as depicted above.
The sequent $\Gamma \entails A, \Delta$ is called the \textbf{conclusion} while $\Gamma', A \entails \Delta'$ is the \textbf{hypothesis}.
We thus say that this proof presents `\textbf{$R_1$ vs $R_2$}'.

Here we present an inductive argument that proceeds by rewriting $\pi$ so as to make \CutRule either appear deeper in the proof or concerning `smaller' cut formulae.
Formally, for each occurrence $c$ of a \CutRule in $\pi$, we define its \textbf{depth} $\operatorname{depth}(c)$ as the number of rules from the leaves of $\pi$ to $c$, and its \textbf{rank} $\operatorname{rank}(c)$ as the number of connectives in the cut formula.

In line with other linear logic approaches to cut-elimination~\cite{mellies:categorical-sem-LL,LLHandbook}, our inductive argument proceeds by induction on the depth of the \CutRule and the formula rank.

The cases to analyze are distinguished by whether $R_1$ and $R_2$ are structural (i.e. from \cref{subfig:struct-frag}) or non-structural.
In the former case, we will proceed, in \cref{sec:cutting-structural-rules}, by showing that $\pi$ can be rewritten to a $\pi'$ that either contains fewer {\CutRule}s or where the {\CutRule}s decrease their depth.
In the latter case, we must additionally reason by induction on the formula rank.
Recall that a formula in the conclusion of a rule (from e.g. \pQLL) is \textbf{principal} when it does not occur in the cedents of its premises---note that, by definition, non-structural rules (\cref{subfig:mult-frag,subfig:add-frag}) have precisely one principal formula.\footnote{Note all the formulae in the conclusion of \ExFalsoRule are principal. However, for the purposes of cut-elimination, we still treat this rule as a structural one in \cref{par:CvsEFQ}.}
This will be the formula whose rank decreases in this case of rewriting.


Thus, we first consider cases when the cut formula appears as principal formula in either $R_1$ or $R_2$ (\cref{sec:principal}), and then when a non-cut formula appears as the principal formula in either the hypothesis or conclusion (\cref{sec:non-principal}).
The resulting case structure is shown in \cref{fig:cut-structure}, with references to our subsections detailing the proof.

Finally, we give a more rigorous termination argument in \cref{app:termination}.

\subsection{Structural rules}\label[appendix]{sec:cutting-structural-rules}

For these cases, we consider rules applied on the left branch.
By admissibility of \ExtExchRule, the other case is completely analogous.

As for termination, note that for \cref{par:CvsA,par:CvsEFQ,par:CvsE} the number of \CutRule applications decreases from the original proof to its transformed version, while for \cref{par:CvsM} the depth of the \CutRule application decreases.

\subsubsection{\AxRule vs hypothesis}\label{par:CvsA}

Rewrite
\begin{equation*}
	\begin{prooftree}
		\hypo{\red{1}}
		\infer1[\AxRule]{A \entails A}
		\hypo{\red{\tensor}}
		\hypo{}
		\ellipsis{$\pi$}{\Gamma,A \entails \Delta}
		\infer3[\CutRule]{\Gamma, A \entails \Delta}
	\end{prooftree}
\end{equation*}
to
\begin{equation*}
	\begin{prooftree}
		\hypo{}
		\ellipsis{$\pi$}{\Gamma, A \entails \Delta}
	\end{prooftree}
\end{equation*}
Validity is preserved since $1 \tensor \validity{\pi} = \validity{\pi}$.


\subsubsection{\ExFalsoRule vs hypothesis}\label{par:CvsEFQ}

Rewrite
\begin{equation*}
	\begin{prooftree}
		\hypo{\red{0}}
		\infer1[\ExFalsoRule]{\Gamma \entails A, \Delta}
		\hypo{\red{\tensor}}
		\hypo{}
		\ellipsis{$\pi$}{\Gamma', A \entails \Delta'}
		\infer3[\CutRule]{\Gamma,\Gamma' \entails \Delta,\Delta'}
	\end{prooftree}
\end{equation*}
to
\begin{equation*}
	\begin{prooftree}
		\hypo{\red{0}}
		\infer1[\ExFalsoRule]{\Gamma,\Gamma' \entails \Delta,\Delta'}
	\end{prooftree}
\end{equation*}
Validity is preserved as $0 \tensor \validity{\pi} = 0$.

\subsubsection{\MixRule vs hypothesis}\label{par:CvsM}

Rewrite
\begin{equation*}
	\begin{prooftree}
		\hypo{}
		\ellipsis{$\pi_1$}{\Gamma_1' \entails \Delta_1'}
		\hypo{\red{\tensor}}
		\hypo{}
		\ellipsis{$\pi_2$}{\Gamma_2' \entails \Delta_2', A}
		\infer3[\MixRule]{\Gamma_1',\Gamma'_2 \entails A, \Delta_1',\Delta_2'}
		\hypo{\red{\tensor}}
		\hypo{}
		\ellipsis{$\pi_3$}{\Gamma, A \entails \Delta}
		\infer3[\CutRule]{\Gamma,\Gamma_1',\Gamma_2' \entails \Delta, \Delta_1',\Delta_2'}
	\end{prooftree}
\end{equation*}
to
\begin{equation*}
	\begin{prooftree}
		\hypo{}
		\ellipsis{$\pi_2$}{\Gamma_2' \entails \Delta_2',A}
		\hypo{\red{\tensor}}
		\hypo{}
		\ellipsis{$\pi_3$}{\Gamma, A \entails \Delta}
		\infer3[\CutRule]{\Gamma,\Gamma'_2 \entails \Delta,\Delta_2'}
		\hypo{\red{\tensor}}
		\hypo{}
		\ellipsis{$\pi_1$}{\Gamma_1' \entails \Delta_1'}
		\infer3[\MixRule]{\Gamma,\Gamma_1',\Gamma_2' \entails \Delta, \Delta_1',\Delta_2'}
	\end{prooftree}
\end{equation*}

Then $(\validity{\pi_1} \tensor \validity{\pi_2}) \tensor \validity{\pi_3} \leq \validity{\pi_1} \tensor (\validity{\pi_2} \tensor \validity{\pi_3})$ holds by associativity of $\tensor$.
 The case in which $A$ appears in the other branch of \MixRule is analogous.

\subsubsection{\EmptyRule vs hypothesis}\label{par:CvsE}

This case is vacuous since the conclusion of \EmptyRule does not contain any formula, and thus cannot be the premise of a \CutRule.

\subsection{Principal formula vs principal formula}\label[appendix]{sec:principal}

In all the cases of this subsection (except for \ref{sec:mult-units} and \ref{sec:add-units}), our proof shows that the rank of the cut formula decreases from the original proof to its transformed version.
In those last two cases, we reduce the number of cut occurrences.

\subsubsection{Dualities (rules \DualRightRule and \DualLeftRule).}\label[appendix]{sec:neg}

Rewrite
\begin{equation*}
	\begin{prooftree}
		\hypo{}
		\ellipsis{$\pi_1$}{\Gamma,A \entails \Delta}
		\infer1[\DualRightRule]{\Gamma \entails A^\dual,\Delta}
		\hypo{\red{\tensor}}
		\hypo{}
		\ellipsis{$\pi_2$}{\Gamma \entails A, \Delta'}
		\infer1[\DualLeftRule]{\Gamma',A^\dual \entails \Delta'}
		\infer3[\CutRule]{\Gamma,\Gamma' \entails \Delta, \Delta'}
	\end{prooftree}
\end{equation*}
to
\begin{equation*}
	\begin{prooftree}
		\hypo{}
		\ellipsis{$\pi_2$}{\Gamma' \entails A, \Delta'}
		\hypo{\red{\tensor}}
		\hypo{}
		\ellipsis{$\pi_1$}{\Gamma,A \entails \Delta}
		\infer3[\CutRule]{\Gamma,\Gamma' \entails \Delta, \Delta'}
	\end{prooftree}
\end{equation*}
We know that $\validity{\pi_1} \tensor \validity{\pi_2} = \validity{\pi_2} \tensor \validity{\pi_1}$ since $\tensor$ is commutative.

\subsubsection{Tensor (rules \TensorLeftRule and \TensorRightRule)}\label[appendix]{sec:tensor}

Rewrite
\begin{equation*}
	\begin{prooftree}
		\hypo{}
		\ellipsis{$\pi_1$}{\Gamma \entails A, \Delta}
		\hypo{\red{\tensor}}
		\hypo{}
		\ellipsis{$\pi_2$}{\Gamma' \entails B, \Delta'}
		\infer3[\TensorRightRule]{\Gamma, \Gamma' \entails A \tensor B, \Delta,\Delta'}
		\hypo{\red{\tensor}}
		\hypo{}
		\ellipsis{$\pi_3$}{\Gamma'', A, B \entails \Delta''}
		\infer1[\TensorLeftRule]{\Gamma'', A \tensor B \entails \Delta''}
		\infer3[\CutRule]{\Gamma,\Gamma',\Gamma'' \entails \Delta,\Delta',\Delta''}
	\end{prooftree}
\end{equation*}
to
\begin{equation*}
	\begin{prooftree}
		\hypo{}
		\ellipsis{$\pi_1$}{\Gamma \entails A, \Delta}
		\hypo{\red{\tensor}}
		\hypo{}
		\ellipsis{$\pi_2$}{\Gamma' \entails B, \Delta'}
		\hypo{\red{\tensor}}
		\hypo{}
		\ellipsis{$\pi_3$}{\Gamma'', A, B \entails \Delta''}
		\infer3[\CutRule]{\Gamma',\Gamma'',A \entails \Delta',\Delta''}
		\infer3[\CutRule]{\Gamma,\Gamma',\Gamma'' \entails \Delta,\Delta',\Delta''}
	\end{prooftree}
\end{equation*}
Then $(\validity{\pi_1} \tensor \validity{\pi_2}) \tensor \validity{\pi_3}  = \validity{\pi_1} \tensor (\validity{\pi_2} \tensor \validity{\pi_3})$, by associativity of tensor.

\subsubsection{Par (rules \ParLeftRule and \ParRightRule)}\label[appendix]{sec:cotensor}
Follows by \cref{sec:neg} and \cref{sec:tensor}, since, by duality, a cut along $A \parr B$ corresponds to a cut along $A^\dual \tensor B^\dual$.
Indeed the proof:
\begin{equation*}
	\begin{prooftree}
		\hypo{}
		\ellipsis{$\pi_1$}{\Gamma \entails \Delta,A,B}
		\infer1[\ParRightRule]{\Gamma \entails \Delta, A \parr B}
		\hypo{\red{\tensor}}
		\hypo{}
		\ellipsis{$\pi_2$}{\Gamma',A \entails \Delta'}
		\hypo{\red{\tensor}}
		\hypo{}
		\ellipsis{$\pi_3$}{\Gamma'',B \entails \Delta''}
		\infer3[\ParLeftRule]{\Gamma',\Gamma'', A \parr B \entails \Delta',\Delta''}
		\infer3[\CutRule]{\Gamma, \Gamma',\Gamma'' \entails \Delta,\Delta',\Delta''}
	\end{prooftree}
\end{equation*}
can be rewritten to the following proof, which falls into the cases previously covered:
\begin{equation*}
	\begin{prooftree}
		\hypo{}
		\ellipsis{$\pi_2$}{\Gamma', A \entails \Delta'}
		\infer1[\DualRightRule]{\Gamma' \entails A^\dual, \Delta'}
		\hypo{\red{\tensor}}
		\hypo{}
		\ellipsis{$\pi_3$}{\Gamma'', B \entails \Delta''}
		\infer1[\DualRightRule]{\Gamma'' \entails B^\dual, \Delta''}
		\infer3[\TensorRightRule]{\Gamma', \Gamma'' \entails A^\dual \tensor B^\dual, \Delta',\Delta''}
		\hypo{\red{\tensor}}
		\hypo{}
		\ellipsis{$\pi_1$}{\Gamma \entails \Delta,A,B}
		\infer1[\DualLeftRule]{\Gamma, B^\dual \entails \Delta,A}
		\infer1[\DualLeftRule]{\Gamma, A^\dual, B^\dual \entails \Delta}
		\infer1[\TensorLeftRule]{\Gamma, A^\dual \tensor B^\dual \entails \Delta}
		\infer3[\CutRule]{\Gamma',\Gamma'',\Gamma \entails \Delta',\Delta'',\Delta}
	\end{prooftree}
\end{equation*}

\subsubsection{Soft Disjunction (rules \pOrLeftRule and \pOrRightRule)}\label[appendix]{sec:sdis}

We have:
\begin{equation*}
	\begin{prooftree}
		\hypo{}
		\ellipsis{$\pi_1$}{\Gamma \vdash A, \Delta}
		\hypo{\red{\padd[p]}}
		\hypo{}
		\ellipsis{$\pi_2$}{\Gamma \vdash B, \Delta}
		\infer3[\pOrRightRule]{\Gamma \vdash A \plor B, \Delta}
		\hypo{\red{\tensor}}
		\hypo{}
		\ellipsis{$\pi_3$}{\Gamma', A \vdash \Delta'}
		\hypo{\red{\pcoadd[p]}}
		\hypo{}
		\ellipsis{$\pi_4$}{\Gamma', B \vdash \Delta'}
		\infer3[\pOrLeftRule]{\Gamma', A \plor B \vdash \Delta'}
		\infer3[\CutRule]{\Gamma,\Gamma' \vdash \Delta,\Delta'}
	\end{prooftree}
\end{equation*}
There are two possible cut-elimination strategies here.
\begin{description}
	\item[\emph{Choice $A$.}] Rewrite to:
	      \begin{equation*}
		      \begin{prooftree}
			      \hypo{}
			      \ellipsis{$\pi_1$}{\Gamma \vdash A, \Delta}
			      \hypo{\red{\tensor}}
			      \hypo{}
			      \ellipsis{$\pi_3$}{\Gamma', A \vdash \Delta'}
			      \infer3[\CutRule]{\Gamma,\Gamma' \vdash \Delta,\Delta'}
		      \end{prooftree}
	      \end{equation*}

	\item[\emph{Choice $B$.}] Rewrite to:
	      \begin{equation*}
		      \begin{prooftree}
			      \hypo{}
			      \ellipsis{$\pi_2$}{\Gamma \vdash B, \Delta}
			      \hypo{\red{\tensor}}
			      \hypo{}
			      \ellipsis{$\pi_4$}{\Gamma', B \vdash \Delta'}
			      \infer3[\CutRule]{\Gamma,\Gamma' \vdash \Delta,\Delta'}
		      \end{prooftree}
	      \end{equation*}
\end{description}
We showed in \cref{lemma:max-prod-ineq} that:
\begin{equation*}
	(\validity{\pi_1} \padd[p] \validity{\pi_2}) \tensor (\validity{\pi_3} \pcoadd[p] \validity{\pi_4}) \leq (\validity{\pi_1} \tensor \validity{\pi_3}) \lor (\validity{\pi_2} \tensor \validity{\pi_4}).
\end{equation*}
Thus perform the rewrite of maximal validity (and, by convention, the one involving $A$ if both validities are equal).
Note this is a deterministic rewrite strategy.

\subsubsection{Soft Conjunction (rules \pAndLeftRule and \pAndRightRule)}\label[appendix]{sec:scon}
Follows by \cref{sec:neg} and \cref{sec:sdis}, since a cut along $A \pland B$ corresponds to a cut along $A^\dual \plor B^\dual$ (analogously to \cref{sec:cotensor}).

\subsubsection{Multiplicative unit (\OneLeftRule and \OneRightRule)}\label[appendix]{sec:mult-units}
Due to the context restrictions of \OneLeftRule and \OneRightRule, there is only one possible such proof:
\begin{equation*}
	\begin{prooftree}
		\hypo{\redbbin{1}}
		\infer1[\OneRightRule]{\ \entails \One}
		\hypo{\red{\tensor}}
		\hypo{\redbbin{1}}
		\infer1[\OneLeftRule]{\One \entails\ }
		\infer3[\CutRule]{\quad\ \entails }
	\end{prooftree}
\end{equation*}
Rewrite this to
\begin{equation*}
	\begin{prooftree}
		\hypo{\red{1}}
		\infer1[\EmptyRule]{\ \entails\ }
	\end{prooftree}
\end{equation*}
Clearly $1 \tensor 1 = 1$, and the number of cut applications decreases.

\subsubsection{Additive units (\BotLeftRule and \TopRightRule)}\label[appendix]{sec:add-units}
Since $\bot$ (resp. $\top$) only has an introduction rule for the left (resp. right) side of a sequent,
the only way for a \CutRule to use $\top$ or $\bot$ as cut formula is to rely on \ExFalsoRule to cover the missing introduction rule.
Then we proceed as in (\cref{par:CvsEFQ}), by rewriting the cut to a \ExFalsoRule rule introducing the same conclusion, and this is validity-preserving (always with validity 0).

\subsection{Non-principal formula vs principal formula}\label[appendix]{sec:non-principal}
The final case we consider is when the cut formula is not principal in both of the rules applied to the premises of the \CutRule.
Let us use the usual terminology and say that in $\Gamma \entails \Delta$, the cedent $\Gamma$ is the antecedent and $\Delta$ is the consequent.
This case splits into four subcases, depending on:
\begin{itemize}
	\item whether the principal formula appears in the conclusion or in the hypothesis of the \CutRule rule,
	\item and whether the principal formula is in the antecedent or the consequent of that conclusion/hypothesis.
\end{itemize}

We treat all the cases in which this principal formula belongs to the conclusion (left premise) of the \CutRule. The cases when the principal formula belongs to the \CutRule's hypothesis are symmetric.

\subsubsection{Principal Formula is in the antecedent of the conclusion of the \CutRule}\label[appendix]{sec:prihyp}

In this subsection, we consider the following cases when a principal formula $A \bullet B$ is contained in the antecedent of the conclusion of the \CutRule.
%
In all cases of this section, the proof transformation decreases the depth of the cut.

\begin{enumerate}[label=\ref{sec:prihyp}.\arabic*, leftmargin=*]
	\item \emph{Rule \TensorLeftRule.}
	      Rewrite
	      \begin{equation*}
		      \begin{prooftree}
			      \hypo{}
			      \ellipsis{$\pi_1$}{\Gamma,C,D \entails A,\Delta}
			      \infer1[\TensorLeftRule]{\Gamma, C \tensor D \entails A,\Delta}
			      \hypo{\red{\tensor}}
			      \hypo{}
			      \ellipsis{$\pi_2$}{\Gamma',A \entails \Delta'}
			      \infer3[\CutRule]{\Gamma, C \tensor D, \Gamma' \entails \Delta,\Delta'}
		      \end{prooftree}
	      \end{equation*}
	      to
	      \begin{equation*}
		      \begin{prooftree}
			      \hypo{}
			      \ellipsis{$\pi_1$}{\Gamma, C, D \entails A,\Delta}
			      \hypo{\red{\tensor}}
			      \hypo{}
			      \ellipsis{$\pi_2$}{\Gamma',A \entails \Delta'}
			      \infer3[\CutRule]{\Gamma, C, D, \Gamma' \entails \Delta,\Delta'}
			      \infer1[\TensorLeftRule]{\Gamma, C \tensor D, \Gamma' \entails \Delta,\Delta'}
		      \end{prooftree}
	      \end{equation*}

	\item \emph{Rule \ParLeftRule.}
	      Without loss of generality, suppose the cut formula comes from the left premise of \ParLeftRule.
		  Rewrite:
	      \begin{equation*}
			      \begin{prooftree}
				      \hypo{}
				      \ellipsis{$\pi_1$}{\Gamma_1, C \entails A, \Delta_1}
				      \hypo{\red{\tensor}}
				      \hypo{}
				      \ellipsis{$\pi_2$}{\Gamma_2,D \entails \Delta_2}
				      \infer3[\ParLeftRule]{\Gamma_1, C \parr D,\Gamma_2 \entails A, \Delta_1,\Delta_2}
				      \hypo{\red{\tensor}}
				      \hypo{}
				      \ellipsis{$\pi_3$}{\Gamma', A \entails \Delta'}
				      \infer3[\CutRule]{\Gamma_1, C \parr D,\Gamma_2, \Gamma' \entails \Delta_1,\Delta_2,\Delta'}
			      \end{prooftree}
	      \end{equation*}
	      to
	      \begin{equation*}
			      \begin{prooftree}
				      \hypo{}
				      \ellipsis{$\pi_1$}{\Gamma_1, C \entails A, \Delta_1}
				      \hypo{\red{\tensor}}
				      \hypo{}
				      \ellipsis{$\pi_3$}{\Gamma', A \entails \Delta'}
				      \infer3[\CutRule]{\Gamma_1, C ,\Gamma' \entails \Delta_1,\Delta'}
				      \hypo{\red{\tensor}}
				      \hypo{}
				      \ellipsis{$\pi_2$}{\Gamma_2, D \entails \Delta_2}
				      \infer3[\ParLeftRule]{\Gamma_1, C \parr D,\Gamma_2, \Gamma' \entails \Delta_1,\Delta_2,\Delta'}
			      \end{prooftree}
	      \end{equation*}
	      We have $\left(\validity{\pi_1} \tensor \validity{\pi_2} \right) \tensor \validity{\pi_3} =\left(\validity{\pi_1} \tensor \validity{\pi_3} \right) \tensor \validity{\pi_2}$ by associativity and commutativity of $\tensor$.


	\item \emph{Rule \pOrLeftRule.}
	      Rewrite
	      \begin{equation*}
			      \begin{prooftree}
				      \hypo{}
				      \ellipsis{$\pi_1$}{\Gamma, A \entails C,\Delta}
				      \hypo{\red{\pcoadd[p]}}
				      \hypo{}
				      \ellipsis{$\pi_2$}{\Gamma, B \entails C,\Delta}
				      \infer3[\pOrLeftRule]{\Gamma, A \plor B \entails C, \Delta}
				      \hypo{\red{\tensor}}
				      \hypo{}
				      \ellipsis{$\pi_3$}{\Gamma', C \entails \Delta'}
				      \infer3[\CutRule]{\Gamma, A \plor B, \Gamma' \entails \Delta,\Delta'}
			      \end{prooftree}
	      \end{equation*}
	      to
	      \begin{equation*}
			      \begin{prooftree}
				      \hypo{}
				      \ellipsis{$\pi_1$}{\Gamma, A \entails C, \Delta}
				      \hypo{\red{\tensor}}
				      \hypo{}
				      \ellipsis{$\pi_3$}{\Gamma',C \entails \Delta'}
				      \infer3[\CutRule]{\Gamma, A , \Gamma' \entails \Delta,\Delta'}
				      \hypo{\red{\pcoadd[p]}}
				      \hypo{}
				      \ellipsis{$\pi_2$}{\Gamma, B \entails C, \Delta}
				      \hypo{\red{\tensor}}
				      \hypo{}
				      \ellipsis{$\pi_3$}{\Gamma',C \entails \Delta'}
				      \infer3[\CutRule]{\Gamma, B, \Gamma' \entails \Delta,\Delta'}
				      \infer3[\pOrLeftRule]{\Gamma, A \plor B, \Gamma' \entails \Delta,\Delta'}
			      \end{prooftree}
	      \end{equation*}
	      Validity is preserved by \cref{lemma:cotensor-over-add-conormal-dist}.

	\item \emph{Rule \pAndLeftRule.}
	      Rewrite
	      \begin{equation*}
			      \begin{prooftree}
				      \hypo{}
				      \ellipsis{$\pi_1$}{\Gamma, A \entails C,\Delta}
				      \hypo{\red{\padd[p]}}
				      \hypo{}
				      \ellipsis{$\pi_2$}{\Gamma, B \entails C,\Delta}
				      \infer3[\pAndLeftRule]{\Gamma, A \pland B \entails C, \Delta}
				      \hypo{\red{\tensor}}
				      \hypo{}
				      \ellipsis{$\pi_3$}{\Gamma', C \entails \Delta'}
				      \infer3[\CutRule]{\Gamma, A \pland B, \Gamma' \entails \Delta,\Delta'}
			      \end{prooftree}
	      \end{equation*}
	      to
	      \begin{equation*}
			      \begin{prooftree}
				      \hypo{}
				      \ellipsis{$\pi_1$}{\Gamma, A \entails C, \Delta}
				      \hypo{\red{\tensor}}
				      \hypo{}
				      \ellipsis{$\pi_3$}{\Gamma',C \entails \Delta'}
				      \infer3[\CutRule]{\Gamma, A , \Gamma' \entails \Delta,\Delta'}
				      \hypo{\red{\padd[p]}}
				      \hypo{}
				      \ellipsis{$\pi_2$}{\Gamma, B \entails C, \Delta}
				      \hypo{\red{\tensor}}
				      \hypo{}
				      \ellipsis{$\pi_3$}{\Gamma',C \entails \Delta'}
				      \infer3[\CutRule]{\Gamma, B, \Gamma' \entails \Delta,\Delta'}
				      \infer3[\pAndLeftRule]{\Gamma, A \pland B, \Gamma' \entails \Delta,\Delta'}
			      \end{prooftree}
	      \end{equation*}
	      Validity is preserved by \cref{lemma:tensor-over-add-dist}.

	\item \emph{Rule \DualLeftRule.}
	      Rewrite
	      \begin{equation*}
		      \begin{prooftree}
			      \hypo{}
			      \ellipsis{$\pi_1$}{\Gamma \entails A,B, \Delta}
			      \infer1[\DualLeftRule]{\Gamma, A^\dual \entails B,\Delta}
			      \hypo{\red{\tensor}}
			      \hypo{}
			      \ellipsis{$\pi_2$}{\Gamma',B \entails \Delta'}
			      \infer3[\CutRule]{\Gamma, A^\dual, \Gamma' \entails \Delta, \Delta'}
		      \end{prooftree}
	      \end{equation*}
	      to
	      \begin{equation*}
		      \begin{prooftree}
			      \hypo{}
			      \ellipsis{$\pi_1$}{\Gamma \entails B,A,\Delta}
			      \hypo{\red{\tensor}}
			      \hypo{}
			      \ellipsis{$\pi_2$}{\Gamma',B \entails \Delta'}
			      \infer3[\CutRule]{\Gamma, \Gamma' \entails A,\Delta,\Delta'}
			      \infer1[\DualLeftRule]{\Gamma, A^\dual, \Gamma' \entails \Delta,\Delta'}
		      \end{prooftree}
	      \end{equation*}

	\item \emph{Rule \OneLeftRule.}
	This rule's conclusion has no side cedents, therefore it cannot appear as non-principal, and this case is vacuous.
	\item \emph{Rule \BotLeftRule.}
	Rewrite
	\begin{equation*}
		\begin{prooftree}
			\hypo{\red{\infty}}
			\infer1[\BotLeftRule]{\Gamma, \bot \entails A, \Delta}
			\hypo{\red{\tensor}}
			\hypo{}
			\ellipsis{$\pi$}{\Gamma', A \entails \Delta'}
			\infer3[\CutRule]{\Gamma, \bot, \Gamma' \entails \Delta, \Delta'}
		\end{prooftree}
	\end{equation*}
	to
	\begin{equation*}
		\begin{prooftree}
			\hypo{\red{\infty}}
			\infer1[\BotLeftRule]{\Gamma, \bot, \Gamma' \entails \Delta, \Delta'}
		\end{prooftree}
	\end{equation*}
	Validity is non-decreasing since $\infty \tensor \validity{\pi} \leq \infty$, and the number of cut occurrences decreases.
\end{enumerate}

\subsubsection{Principal Formula is in the consequent of the conclusion of the \CutRule}\label[appendix]{sec:pricon}
We consider now the case in which the cut formula is again not principal in \emph{both} premises of the \CutRule, but now we focus on a principal formula that appears as the consequent of the conclusion of the \CutRule rule.

This case is perfectly dual to the previous one.

Indeed, consider the following case of induction, where \CutRule needs to be commuted with a right introduction rule for some logical connective we denote by $\bullet$:
\begin{equation*}
	\begin{prooftree}
		\hypo{}
		\ellipsis{$\pi_1$}{}
		\infer1[$\bullet\rightrule$]{\Gamma \entails \Delta, C \bullet D, A}
		\hypo{\red{\tensor}}
		\hypo{}
		\ellipsis{$\pi_2$}{\Gamma',A \entails \Delta'}
		\infer3[\CutRule]{\Gamma, \Gamma' \entails \Delta,\Delta', C \bullet D}
	\end{prooftree}
\end{equation*}
This can be rewritten to the following proof, which falls into the cases previously covered:
\begin{equation*}
	\begin{prooftree}
		\hypo{}
		\ellipsis{$\pi_3$}{}
		\infer1[$\bullet^\dual\rightrule$]{\Gamma, C^\dual \bullet^\dual D^\dual \entails \Delta, A}
		\hypo{\red{\tensor}}
		\hypo{}
		\ellipsis{$\pi_2$}{\Gamma',A \entails \Delta'}
		\infer3[\CutRule]{\Gamma, \Gamma', C^\dual \bullet^\dual D^\dual \entails \Delta,\Delta'}
		\infer1[\DualRightRule]{\Gamma, \Gamma' \entails \Delta,\Delta', C \bullet D}
	\end{prooftree}
\end{equation*}

To define $\pi_3$, note that $\bullet^\dual\rightrule$ has the same amount of hypotheses as $\bullet\leftrule$ does.
In the unary case, $\pi_3$ is obtained by applying \DualLeftRule twice (to $C$ and $D$), and then composing with $\pi_1$.
In the binary case, $\pi_1$ has for conclusion hypersequent a tensor of two sequents, so that $\pi_1\ =\ \pi_{1L} \redbin{\bullet} \pi_{1R}$ for the appropriate $\redbin{\bullet}$ used by the rule under study. Then define $\pi_{3L}$ (resp. $\pi_{3R}$) by applying \DualLeftRule twice (to $C$ and $D$), and then composing with $\pi_{1L}$ (resp. $\pi_{1R}$). Remark that $\pi_3\ =\ \pi_{3L} \redbin{\bullet} \pi_{3R}$ provides a proof of the same validity.








\subsection{The Termination argument}
\label[appendix]{app:termination}
Let us use the notation from \cref{sec:general}.
For a proof $\pi$, define the metrics $\operatorname{met}(\pi) = (r_m,n_m,n,d)$ such that if $\pi$ is cut-free, then $\operatorname{met}(\pi) = (0,0,0,0)$, and else:
\begin{enumerate}
	\item $r_m$ is the maximal rank of the cut formula of a cut occurrence in $\pi$,
	\item $n_m$ is the number of cut occurrences of maximal rank in $\pi$,
	\item $n$ is the total number of cut occurrences in $\pi$,
	\item $d$ is the sum of the depths of all cut occurrences in $\pi$.
\end{enumerate}

We consider the lexicographic order over such tuples, which is well-founded.
Now every step of the cut-elimination process is strictly decreasing for this order, so that the process needs to terminate.
Note that the process is defined for every possible case involving a cut occurrence, so it can only halt on a cut-free proof: else, a rewrite could be applied.

\subsection{Proof of Decidability (\cref{th:decidability})}
\label[appendix]{app:decidability}

Cut-elimination guarantees that the provability of a sequent is realized by a \CutRule-free proof (of which there is always at least one, \ExFalsoRule).
Note moreover that:
\begin{enumerate}
	\item Explicit dualities in \pQLL (introduced via the rules \DualLeftRule and \DualRightRule) can be eliminated without affecting validity, so that we can restrict our attention to derivations of $\Gamma \entails \Delta$ that do not use \CutRule nor \DualLeftRule nor \DualRightRule. Call these \textbf{reduced}.
	\item We can, again without change in validity, remove any \MixRule which uses the empty sequent in one of its branches.
	\item By inspecting the remaining rules and cases, we observe the subformula property: each premise of a rule, if any, is of strictly smaller complexity, understood in the traditional sense as the sum of the number of logical connectives in each formula, over all the formulas appearing in the sequent.
\end{enumerate}
Therefore reduced proofs are in finite number, and they can be effectively enumerated. The validity of each \pQLL derivation is computable up to any desired precision by real arithmetic.%
\footnote{Observe, in particular, that it is trivial to decide whether the validity of a proof is non-zero, since such validity is given as a finite expression involving only $0$, $1$, and $\infty$ whose positivity can be checked symbolically.}
It follows that provability is realized by at least one proof (since there are finitely many proofs of potential maximal validity), and that proof is computable via the arguments just raised.

\section{\texorpdfstring{Equivalence of \ISOMIXMALL{} and \pQLL[\infty]}{Equivalence of ISOMIX-MALL and hard QLL}}
\label[appendix]{app:isomix-mall}


\begin{figure}
	\centering
	\begin{adjustbox}{width=\linewidth,center}
		\begin{tabular}{c}
			\begin{subfigure}{\linewidth}
				\caption{Structural Rules}
				\vspace*{1ex}
				\label{subfig:mall-struct-frag}
				\centering
				\begin{tabular}{c}
					\begin{tabular}{rl}
						 \begin{prooftree}
							\hypo{\phantom{T}}
							\infer1[\MALLAxRule]{A \entails A}
						\end{prooftree}
						&
						\begin{prooftree}
							\hypo{\phantom{T}}
							\infer1[\MALLEmptyRule]{\ \entails \ }
						\end{prooftree}
						\\[4ex]
						\begin{prooftree}
							\hypo{\Gamma \entails A, \Delta}
							\hypo{\Gamma' , A \entails \Delta'}
							\infer2[\MALLCutRule]{\Gamma, \Gamma' \entails \Delta, \Delta'}
						\end{prooftree}
						&
						\begin{prooftree}
							\hypo{\Gamma \entails \Delta}
							\hypo{\Gamma' \entails \Delta'}
							\infer2[\MALLMixRule]{\Gamma, \Gamma' \entails \Delta, \Delta'}
						\end{prooftree}
					\end{tabular}
				\end{tabular}
			\end{subfigure}
			\\[2ex]
			\begin{subfigure}{\linewidth}
				\caption{Rules for Multiplicatives}
				\vspace*{1ex}
				\label{subfig:mall-mult-frag}
				\centering
				\begin{tabular}{c}
					\begin{tabular}{rl}
						\begin{prooftree}
							\hypo{\Gamma, A,B \entails \Delta}
							\infer1[\MALLTensorLeftRule]{\Gamma, A \tensor B \entails \Delta}
						\end{prooftree}
						&
							\begin{prooftree}
								\hypo{\Gamma \entails A,B, \Delta}
								\infer1[\MALLParRightRule]{\Gamma \entails A \parr B, \Delta}
							\end{prooftree}
					\\[4ex]
						\begin{prooftree}
							\hypo{\Gamma, A \entails \Delta}
							\hypo{\Gamma', B \entails \Delta'}
							\infer2[\MALLParLeftRule]{\Gamma, \Gamma', A \parr B \entails \Delta, \Delta'}
						\end{prooftree}
						&
							\begin{prooftree}
								\hypo{\Gamma \entails A, \Delta}
								\hypo{\Gamma' \entails B, \Delta'}
								\infer2[\MALLTensorRightRule]{\Gamma, \Gamma' \entails A \tensor B, \Delta, \Delta'}
							\end{prooftree}
					\end{tabular}
					\\[8ex]
					\begin{tabular}{rrll}
						\begin{prooftree}
							\hypo{}
							\infer1[\MALLOneLeftRule]{\One \entails}
						\end{prooftree}
						&
						\begin{prooftree}
							\hypo{\Gamma \entails A, \Delta}
							\infer1[\MALLStarLeftRule]{\Gamma, A^\dual \entails \Delta}
						\end{prooftree}
						&
						\begin{prooftree}
							\hypo{\Gamma, A \entails \Delta}
							\infer1[\MALLStarRightRule]{\Gamma \entails A^\dual, \Delta}
						\end{prooftree}
						&
							\begin{prooftree}
								\hypo{}
								\infer1[\MALLOneRightRule]{\entails \One}
							\end{prooftree}
					\end{tabular}
				\end{tabular}
			\end{subfigure}
			\\[2ex]
			\begin{subfigure}{\linewidth}
				\caption{Rules for Additives}
				\vspace*{1ex}
				\label{subfig:mall-add-frag}
				\centering
				\begin{tabular}{c}
					\begin{tabular}{rl}
						\begin{prooftree}
							\hypo{\Gamma, A \entails \Delta}
							\hypo{\Gamma, B \entails \Delta}
							\infer2[\MALLOrLeftRule]{\Gamma, A \lor B \entails \Delta}
						\end{prooftree}
						&
						\begin{prooftree}
							\hypo{\Gamma \entails A, \Delta}
							\hypo{\Gamma \entails B, \Delta}
							\infer2[\MALLAndRightRule]{\Gamma \entails A \land B, \Delta}
						\end{prooftree}
					\end{tabular}
					\\[4ex]
					\begin{tabular}{rcl}
						\begin{prooftree}
							\hypo{\Gamma, A_i \entails \Delta}
							\infer1[\MALLAndLeftRule{i}]{\Gamma, A_1 \land A_2 \entails \Delta}
						\end{prooftree}
						&
						$i=1,2$
						&
						\begin{prooftree}
							\hypo{\Gamma \entails A_i, \Delta}
							\infer1[\MALLOrRightRule{i}]{\Gamma \entails A_1 \lor A_2, \Delta}
						\end{prooftree}
					\end{tabular}
					\\[4ex]
					\begin{tabular}{rl}
						\begin{prooftree}
							\hypo{}
							\infer1[\MALLBotLeftRule]{\Gamma, \bot \entails \Delta}
						\end{prooftree}
						\quad & \quad
							\begin{prooftree}
								\hypo{}
								\infer1[\MALLTopRightRule]{\Gamma \entails \top, \Delta}
							\end{prooftree}
					\end{tabular}
				\end{tabular}
			\end{subfigure}
		\end{tabular}
	\end{adjustbox}
	\caption{\ISOMIXMALL}
	\label{fig:isomixmall}
\end{figure}

\begin{figure}
	\centering
	\begin{adjustbox}{width=\linewidth,center}
		\begin{tabular}{c}
			\begin{subfigure}{\linewidth}
				\caption{Structural Rules}
				\vspace*{1ex}
				\label{subfig:qmall-struct-frag}
				\centering
				\begin{tabular}{c}
					\begin{tabular}{rcl}
						\begin{prooftree}
							\hypo{\red{\true}}
							\infer1[\qMALLAxRule]{A \entails A}
						\end{prooftree}
						&
						\begin{prooftree}
							\hypo{\red{\true}}
							\infer1[\qMALLEmptyRule]{\ \entails \ }
						\end{prooftree}
						&
						\begin{prooftree}
							\hypo{\red{\false}}
							\infer1[\qMALLExFalsoRule]{\Gamma \entails \Delta}
						\end{prooftree}
					\end{tabular}
					\\[4ex]
					\begin{tabular}{rlcc}
						\begin{prooftree}
							\hypo{\Gamma \entails A, \Delta \redbbin{\land} \Gamma' , A \entails \Delta'}
							\infer1[\qMALLCutRule]{\Gamma, \Gamma' \entails \Delta, \Delta'}
						\end{prooftree}
						&
						\begin{prooftree}
							\hypo{\Gamma \entails \Delta \redbbin{\land} \Gamma' \entails \Delta'}
							\infer1[\qMALLMixRule]{\Gamma, \Gamma' \entails \Delta, \Delta'}
						\end{prooftree}
					\end{tabular}
				\end{tabular}
			\end{subfigure}
			\\[2ex]
			\begin{subfigure}{\linewidth}
				\caption{Rules for Multiplicatives}
				\vspace*{1ex}
				\label{subfig:qmall-mult-frag}
				\centering
				\begin{tabular}{c}
					\begin{tabular}{rl}
						\begin{prooftree}
							\hypo{\Gamma, A,B \entails \Delta}
							\infer1[\qMALLTensorLeftRule]{\Gamma, A \tensor B \entails \Delta}
						\end{prooftree}
						&
						\begin{prooftree}
							\hypo{\Gamma \entails A,B, \Delta}
							\infer1[\qMALLParRightRule]{\Gamma \entails A \parr B, \Delta}
						\end{prooftree}
					\\[4ex]
						\begin{prooftree}
							\hypo{\Gamma, A \entails \Delta \redbbin{\land} \Gamma', B \entails \Delta'}
							\infer1[\qMALLParLeftRule]{\Gamma, \Gamma', A \parr B \entails \Delta, \Delta'}
						\end{prooftree}
						&
						\begin{prooftree}
							\hypo{\Gamma \entails A, \Delta \redbbin{\land} \Gamma',\entails B, \Delta'}
							\infer1[\qMALLTensorRightRule]{\Gamma,\Gamma' \entails A \tensor B, \Delta, \Delta'}
						\end{prooftree}
					\end{tabular}
					\\[8ex]
					\begin{tabular}{rrll}
						\begin{prooftree}
							\hypo{\red{\true}}
							\infer1[\qMALLOneLeftRule]{\One \entails}
						\end{prooftree}
						&
						\begin{prooftree}
							\hypo{\Gamma \entails A, \Delta}
							\infer1[\qMALLStarLeftRule]{\Gamma, A^\dual \entails \Delta}
						\end{prooftree}
						&
						\begin{prooftree}
							\hypo{\Gamma, A \entails \Delta}
							\infer1[\qMALLStarRightRule]{\Gamma \entails A^\dual, \Delta}
						\end{prooftree}
						&
						\begin{prooftree}
							\hypo{\red{\true}}
							\infer1[\qMALLOneRightRule]{\entails \One}
						\end{prooftree}
					\end{tabular}
				\end{tabular}
			\end{subfigure}
			\\[2ex]
			\begin{subfigure}{\linewidth}
				\caption{Rules for Additives}
				\vspace*{1ex}
				\label{subfig:qmall-add-frag}
				\centering
				\begin{tabular}{c}
					\begin{tabular}{rl}
						\begin{prooftree}
							\hypo{\Gamma, A \entails \Delta \redbbin{\land} \Gamma, B \entails \Delta}
							\infer1[\qMALLpOrLeftRule]{\Gamma, A \plor B \entails \Delta}
						\end{prooftree}
						&
						\begin{prooftree}
							\hypo{\Gamma \entails A , \Delta \redbbin{\land} \Gamma \entails B , \Delta}
							\infer1[\qMALLpAndRightRule]{\Gamma \entails A \pland B, \Delta}
						\end{prooftree}
					\\[4ex]
						\begin{prooftree}
							\hypo{\Gamma, A \entails \Delta \redbbin{\lor} \Gamma, B \entails \Delta}
							\infer1[\qMALLpAndLeftRule] {\Gamma, A \pland B \entails \Delta}
						\end{prooftree}
						&
						\begin{prooftree}
							\hypo{\Gamma \entails A , \Delta \redbbin{\lor} \Gamma \entails B, \Delta}
							\infer1[\qMALLpOrRightRule]{\Gamma \entails A \plor B , \Delta }
						\end{prooftree}
					\\[4ex]
						\begin{prooftree}
							\hypo{\red{\true}}
							\infer1[\qMALLBotLeftRule]{\Gamma, \bot \entails \Delta}
						\end{prooftree}
						\quad & \quad
						\begin{prooftree}
							\hypo{\red{\true}}
							\infer1[\qMALLTopRightRule]{\Gamma \entails \top, \Delta}
						\end{prooftree}
					\end{tabular}
				\end{tabular}
			\end{subfigure}
		\end{tabular}
	\end{adjustbox}
	\caption{The qualitative version of \pQLL.}
	\label{fig:isomixmall-q}
\end{figure}

\begin{figure}
	\centering
	\begin{adjustbox}{width=\linewidth,center}
		\begin{tabular}{c}
			\begin{subfigure}{\linewidth}
				\caption{Internal Structural Rules}
				\vspace*{1ex}
				\label{subfig:hmall-struct-frag}
				\centering
				\begin{tabular}{c}
					\begin{tabular}{rl}
						 \begin{prooftree}
							\hypo{}
							\infer1[\hMALLAxRule]{\hyper{G} \mid A \entails A}
						\end{prooftree}
						&
						\begin{prooftree}
							\hypo{}
							\infer1[\hMALLEmptyRule]{\hyper{G} \mid \ \entails \ }
						\end{prooftree}
						\\[4ex]
						\begin{prooftree}
							\hypo{\hyper{G} \mid \Gamma \entails A, \Delta}
							\hypo{\hyper{G} \mid \Gamma' , A \entails \Delta'}
							\infer2[\hMALLCutRule]{\hyper{G} \mid \Gamma, \Gamma' \entails \Delta, \Delta'}
						\end{prooftree}
						&
						\begin{prooftree}
							\hypo{\hyper{G} \mid \Gamma \entails \Delta}
							\hypo{\hyper{G} \mid \Gamma' \entails \Delta'}
							\infer2[\hMALLMixRule]{\hyper{G} \mid \Gamma, \Gamma' \entails \Delta, \Delta'}
						\end{prooftree}
					\end{tabular}
				\end{tabular}
			\end{subfigure}
			\\[4ex]
			\begin{subfigure}{\linewidth}
				\caption{External Structural Rules}
				\vspace*{1ex}
				\label{subfig:hmall-external-structural}
				\centering
				\begin{tabular}{c}
					\begin{tabular}{rl}
						\begin{prooftree}
							\hypo{\hyper{G}}
							\infer1[\hMALLExtWeakRule]{\hyper{G} \mid \hyper{H}}
						\end{prooftree}
						&
						\begin{prooftree}
							\hypo{\hyper{G} \mid \hyper{H} \mid \hyper{H}}
							\infer1[\hMALLExtContRule]{\hyper{G} \mid \hyper{H}}
						\end{prooftree}
					\end{tabular}
				\end{tabular}
			\end{subfigure}
			\\[2ex]
			\begin{subfigure}{\linewidth}
				\caption{Rules for Multiplicatives}
				\vspace*{1ex}
				\label{subfig:hmall-mult-frag}
				\centering
				\begin{tabular}{c}
					\begin{tabular}{rl}
						\begin{prooftree}
							\hypo{\hyper{G} \mid \Gamma, A,B \entails \Delta}
							\infer1[\hMALLTensorLeftRule]{\hyper{G} \mid \Gamma, A \tensor B \entails \Delta}
						\end{prooftree}
						&
							\begin{prooftree}
								\hypo{\hyper{G} \mid \Gamma \entails A,B, \Delta}
								\infer1[\hMALLParRightRule]{\hyper{G} \mid \Gamma \entails A \parr B, \Delta}
							\end{prooftree}
					\\[4ex]
						\begin{prooftree}
							\hypo{\hyper{G} \mid \Gamma, A \entails \Delta}
							\hypo{\hyper{G} \mid \Gamma', B \entails \Delta'}
							\infer2[\hMALLParLeftRule]{\hyper{G} \mid \Gamma, \Gamma', A \parr B \entails \Delta, \Delta'}
						\end{prooftree}
						&
							\begin{prooftree}
								\hypo{\hyper{G} \mid \Gamma \entails A, \Delta}
								\hypo{\hyper{G} \mid \Gamma' \entails B, \Delta'}
								\infer2[\hMALLTensorRightRule]{\hyper{G} \mid \Gamma, \Gamma' \entails A \tensor B, \Delta, \Delta'}
							\end{prooftree}
					\end{tabular}
					\\[8ex]
					\begin{tabular}{rrll}
						\begin{prooftree}
							\hypo{}
							\infer1[\hMALLOneLeftRule]{\hyper{G} \mid \One \entails\ }
						\end{prooftree}
						&
						\begin{prooftree}
							\hypo{\hyper{G} \mid \Gamma \entails A, \Delta}
							\infer1[\hMALLStarLeftRule]{\hyper{G} \mid \Gamma, A^\dual \entails \Delta}
						\end{prooftree}
						&
						\begin{prooftree}
							\hypo{\hyper{G} \mid \Gamma, A \entails \Delta}
							\infer1[\hMALLStarRightRule]{\hyper{G} \mid \Gamma \entails A^\dual, \Delta}
						\end{prooftree}
						&
							\begin{prooftree}
								\hypo{}
								\infer1[\hMALLOneRightRule]{\hyper{G} \mid \ \entails \One}
							\end{prooftree}
					\end{tabular}
				\end{tabular}
			\end{subfigure}
			\\[2ex]
			\begin{subfigure}{\linewidth}
				\caption{Rules for Additives}
				\vspace*{1ex}
				\label{subfig:hmall-add-frag}
				\centering
				\begin{tabular}{c}
					\begin{tabular}{rl}
						\begin{prooftree}
							\hypo{\hyper{G} \mid \Gamma, A \entails \Delta}
							\hypo{\hyper{G} \mid \Gamma, B \entails \Delta}
							\infer2[\hMALLOrLeftRule]{\hyper{G} \mid \Gamma, A \lor B \entails \Delta}
						\end{prooftree}
						&
							\begin{prooftree}
								\hypo{\hyper{G} \mid \Gamma \entails A, \Delta}
								\hypo{\hyper{G} \mid \Gamma \entails B, \Delta}
								\infer2[\hMALLAndRightRule]{\hyper{G} \mid \Gamma \entails A \land B, \Delta}
							\end{prooftree}
						\\[4ex]
						\begin{prooftree}
							\hypo{\hyper{G} \mid \Gamma, A \entails \Delta \mid \Gamma, B \entails \Delta}
							\infer1[\hMALLAndLeftRule]{\hyper{G} \mid \Gamma, A \land B \entails \Delta}
						\end{prooftree}
						&
						\begin{prooftree}
							\hypo{\hyper{G} \mid \Gamma \entails A, \Delta \mid \Gamma \entails B, \Delta}
							\infer1[\hMALLOrRightRule]{\hyper{G} \mid \Gamma \entails A \lor B, \Delta}
						\end{prooftree}
					\end{tabular}
					\\[8ex]
					\begin{tabular}{rl}
						\begin{prooftree}
							\hypo{}
							\infer1[\hMALLBotLeftRule]{\hyper{G} \mid \Gamma, \bot \entails \Delta }
						\end{prooftree}
						\quad & \quad
							\begin{prooftree}
								\hypo{}
								\infer1[\hMALLTopRightRule]{\hyper{G} \mid \Gamma \entails \top, \Delta}
							\end{prooftree}
					\end{tabular}
				\end{tabular}
			\end{subfigure}
		\end{tabular}
	\end{adjustbox}
	\caption{The hypersequent calculus for isomix MALL (adapted from \cite[Fig~4.1]{metcalfeProofTheoryFuzzy2009}), \hISOMIXMALL.}
	\label{fig:isomixmall-h}
\end{figure}

\subsection{Proof of \cref{th:isomixmall-cons-and-adeq-over-qual-pqll}}
To show that qualitative \pQLL is equivalent to \hISOMIXMALL we turn to the observations of \cref{sec:hyperseq} regarding hypersequent and $\Prop$-enriched calculi.
In particular, it is straightforward to translate proof rules from a hypersequent calculus to admissible proof rules in a $\Prop$-enriched one admitting \ExFalsoRule, following the same procedure outlined in \cref{lemma:struct-schema}.
Specifically, given proofs $\pi_i$ as below left, any of them either proves $\hyper{G}$ or $\hyper{H}_i$.
Thus define $\pi_{\hyper{G}}$ to be any of those proving $\hyper{G}$, or \ExFalsoRule in case there is none, and likewise for $\pi_{\hyper{H}_i}$.
We thus obtain a proof as below right:
\begin{equation*}
		\hspace*{-8ex}
		\begin{prooftree}
					\hypo{}
				\ellipsis{$\pi_1$}{\red{(}\hyper{G} \redbin{\lor} \hyper{H}_1 \red{)}}
				\hypo{\red{\land}\ \cdots\ \red{\land}}
					\hypo{}
				\ellipsis{$\pi_n$}{\red{(}\hyper{G} \redbin{\lor} \hyper{H}_n \red{)}}
			\infer[separation=.5em]3{\hyper{G} \redbin{\lor} \hyper{K}}
		\end{prooftree}
		\quad \mapsto \quad
		\begin{prooftree}
				\hypo{}
			\ellipsis{$\pi_{\hyper{G}}$}{\hyper{G}}
			\hypo{\red{\lor}}
				\hypo{}
			\ellipsis{$\pi_{\hyper{H}_1} \redbin{\land} \cdots \redbin{\land} \pi_{\hyper{H}_n}$}{\red{(}\hyper{H}_1 \redbin{\land} \cdots \redbin{\land} \hyper{H}_n \red{)}}
		\infer[separation=.5em]3{\hyper{G} \redbin{\lor} \hyper{K}}
		\end{prooftree}
\end{equation*}
\emph{Vice versa}, given proofs in a $\Prop$-enriched calculus like qualitative \pQLL, we define
\begin{eqalign}
	\begin{prooftree}
			\hypo{}
		\ellipsis{$\pi_1$}{\hyper{H}_1}
		\hypo{\red{\lor}}
			\hypo{}
		\ellipsis{$\pi_2$}{\hyper{H}_2}
		\infer[separation=.5em]3{\hyper{K}}
	\end{prooftree}
	\qquad &\mapsto \qquad
	\begin{prooftree}
		\hypo{}
		\ellipsis{$\pi_2 \mid \pi_2$}{\varnothing \mid \hyper{H}_1 \mid \hyper{H}_2}
		\infer1{\varnothing \mid \hyper{K}}
	\end{prooftree}
	\\
	\begin{prooftree}
			\hypo{}
		\ellipsis{$\pi_1$}{\hyper{H}_1}
		\hypo{\red{\land}}
			\hypo{}
		\ellipsis{$\pi_2$}{\hyper{H}_2}
		\infer3{\hyper{K}}
	\end{prooftree}
	\qquad &\mapsto \qquad
	\begin{prooftree}
		\hypo{}
		\ellipsis{$\pi_2$}{\varnothing \mid \hyper{H}_1}
		\hypo{}
		\ellipsis{$\pi_2$}{\varnothing \mid \hyper{H}_2}
		\infer2{\varnothing \mid \hyper{K}}
	\end{prooftree}
\end{eqalign}
The language of \ISOMIXMALL is, up to relabeling of connectives, the same as the one introduced in \cref{def:qll-formulae} for $p=\infty$.
We can then easily draw a correspondence between the rules of one and the other (compare \cref{fig:isomixmall-q} and \cref{fig:isomixmall-h}), with a few exceptions to highlight:
\begin{enumerate}
	\item \ExFalsoRule in qualitative \pQLL actually translates to external weakening in the hypersequent calculus, with $\hyper{G} = \varnothing$,
	\item \emph{vice versa}, the analogue of \cref{lemma:struct-schema} for $\Prop$-enriched calculi makes the translations of external weakening and contraction admissible in quantitative \pQLL.
\end{enumerate}

Thus we obtain \cref{th:isomixmall-cons-and-adeq-over-qual-pqll}.

\subsection{Proof of \cref{th:isomixmall-cons-and-adeq-over-infty-qll}}
The claim boils down to a correspondence between rules that allows us to translate proofs directly between the two systems.
Such a correspondence is mostly evident by comparing \cref{fig:pqll} with \cref{fig:isomixmall}.

Indeed, when going from \ISOMIXMALL to \pQLL[\infty], the only sticking point is the admissibility of the additive fragment of the former in the latter, which we established in \cref{lemma:adm-of-MALL-additives}.
As for the converse, note that a proof in \pQLL[\infty] can only have $0$, $1$, or $\infty$ as validities.
Moreover, a proof translated from \ISOMIXMALL does not use \ExFalsoRule, and thus has validity greater or equal than $1$, thus proving one half of the claim.

Conversely, let $P$ be a \pQLL[\infty] proof of $\Gamma \entails \Delta$ of validity greater than $1$.
By \cref{lemma:adm-of-MALL-additives} and \cref{lemma:efq-elimination}, we can assume $P$ is \ExFalsoRule-free and only uses the unary versions of \pOrRightRule[\infty] and \pAndLeftRule[\infty].
Then the translation of rules in this direction is straightforward, once noted that, since $a \tensor b = 0$ iff either $a=0$ or $b=0$, the rules of \pQLL[\infty] involving $\red{\tensor}$ can be replaced as-is with those in \ISOMIXMALL.

\section{Soundness}
\label[appendix]{app:soundness}

We prove now \cref{prop:soundness}.

For brevity and without loss of generality, when a cedent $\Gamma^\monop$ or $\Delta^\dumonop$ has to appear, we replace it by a single variable for an element of $\cat{S}$.

We start by noting that our translation of cedents is well-posed (that is, respects multiset structure) thanks to the unitality, associativity, and commutativity axioms, that is \eqref{eq:softale-unit-ass-comm}.
Likewise, we can use $\One$ on both sides because $\One^{\dual\dual} = \One$, that is \eqref{eq:softale-involutivity}.
We thus employ these facts tacitly below.

We start from the following rules, which together with the axioms for involutivity \eqref{eq:softale-involutivity} and duality \eqref{eq:softale-duality} allow us to cut the subsequent cases in half.
\begin{description}
	\item[\DualRightRule] is sound by \eqref{eq:softale-star-aut}: $(a, b \eleq c) = (a \eleq (b \monop c^\dual)^\dual) = (a \eleq b^\dual \dumonop c)$.
	\item[\DualLeftRule] holds by duality, whose application we illustrate here once and for all:
	\begin{eqalign}
		(a \eleq b \dumonop c)
		\overset{\eqref{eq:softale-duality}}{=} (b \monop c \eleq a)
		\overset{\DualRightRule}{=} (c \eleq b^\dual \dumonop a)
		\overset{\eqref{eq:softale-duality}}{=} (a \monop b \eleq c)
	\end{eqalign}
	Note the use of commutativity as well as the fact $a \dumonop b = (a^\dual \monop b^\dual)^\dual$.
\end{description}

We then have:

\begin{description}
	\item[\AxRule] $1 \leq (a \eleq a)$ by \eqref{eq:enriched-refl} for enriched categories.

	\item[\EmptyRule, \OneLeftRule, and \OneRightRule] These all amount (thanks also to $\One^\dual = \One$ from \eqref{eq:softale-involutivity}) to checking that $1 \leq (\One \eleq \One)$, which holds by \eqref{eq:enriched-refl}.

	\item[\ExFalsoRule] This rule corresponds to the trivial fact $0 \leq (a \eleq b)$.

	\item[\CutRule] After moving the context to the outer cedents, it is enough to check that $(a \eleq b) \tensor (b \eleq c) \leq (a \eleq c)$, which holds by \eqref{eq:enriched-trans}.

	\item[\MixRule] This follows from \eqref{eq:softale-interchange} and \eqref{eq:softale-mix} (right):
	\begin{eqalign}
		(a \eleq b) \tensor (a' \eleq b') \tensor 1 &\overset{\eqref{eq:softale-interchange} \tensor \eqref{eq:softale-mix}} \leq ((a \monop a') \eleq (b \monop b')) \tensor ((b \monop b') \eleq (b \dumonop b'))\\
		&\leq (a \monop a') \eleq (b \dumonop b').
	\end{eqalign}

	\item[\TensorLeftRule and \ParRightRule] hold by definition,

	\item[\TensorRightRule and \ParLeftRule] It suffices to prove the first as the second is formally dual, and, once the context is all moved on the left side, the first is directly \eqref{eq:softale-interchange}.

	\item[\pOrLeftRule and \pAndRightRule] It suffices to prove the first as the second is formally dual, and the first is directly \eqref{eq:softale-p-joins-minimality}.
	Specifically, we move all the context on the right then invoke \eqref{eq:softale-p-joins-minimality}.

	\item[\pAndLeftRule and \pOrRightRule] It suffices to prove the first as the second is formally dual, and the first is directly \eqref{eq:softale-p-joins-upper-bound}.

	\item[\BotLeftRule and \TopRightRule] It suffices to prove the first as the second is formally dual, and the first is directly \eqref{eq:softale-soft-bottom}.
\end{description}

\section{Construction of the classifying softale}
\label[appendix]{app:syn-softale}

Here we prove that \cref{def:classifying-softale} is well-posed.

We start by observing $\SynSoftale{\Theory}$ is a $\MulReals$-preorder.
Indeed, reflexivity is granted by \AxRule and transitivity by \CutRule.
The way this is proven is paradigmatic for all the other inequations below, namely by exhibiting a derivation starting from the left-hand side and ending in the right-hand side, e.g.
\begin{eqalign*}
	1 \leq \validity{A \entails A} &\iff
	\begin{prooftree}
		\hypo{\red{1}}
		\infer1[\AxRule]{A \entails A}
	\end{prooftree}
	\\[2ex]
	\validity{A \entails B} \tensor \validity{B \entails C} \leq \validity{A \entails C} &\iff
	\begin{prooftree}
		\hypo{A \entails B \redbbin{\tensor} B \entails C}
		\infer1[\CutRule]{A \entails C}
	\end{prooftree}
\end{eqalign*}

We now similarly prove each structural equation (the items below are numbered by the equation they prove in \cref{def:p-softale}).
\begin{description}

	\item[\eqref{eq:softale-interchange}] $\validity{A \entails B} \tensor \validity{D \entails C} \leq \validity{A \tensor D \entails B \tensor C}$ is witnessed by the following derivation:
	\begin{equation}
	   \begin{prooftree}
		   \hypo{A \entails B}
		   \hypo{\red{\tensor}}
		   \hypo{D \entails C}
		   \infer3[\TensorRightRule]{A \tensor D \entails B, C}
		   \infer1[\TensorLeftRule]{A \tensor D \entails B \tensor C}
	   \end{prooftree}
	\end{equation}

	\item[\eqref{eq:softale-unit-ass-comm}] We must prove $1 \leq \validity{A \tensor \One \entails A}$ and $1 \leq \validity{A \entails A \tensor \One}$:
	\begin{equation}
		\begin{prooftree}
				\hypo{\red{1}}
				\infer1[\AxRule]{A \entails A}
				\hypo{\red{\tensor}}
				\hypo{\red{1}}
				\infer1[\OneLeftRule]{\One \entails\ }
			\infer3[\MixRule]{A, \One \entails A}
			\infer1[\TensorLeftRule]{A \tensor \One \entails A}
		\end{prooftree}
		\hspace*{10ex}
		\begin{prooftree}
					\hypo{\red{1}}
				\infer1[\AxRule]{A \entails A}
			\hypo{\red{\tensor}}
					\hypo{\red{1}}
				\infer1[\OneRightRule]{\ \entails \One}
			\infer3[\TensorRightRule]{A &\entails A \tensor \One}
		\end{prooftree}
	\end{equation}
	The associativity and commutativity of $\tensor$ are witnessed by the same proofs that witness them for \ISOMIXMALL.
	Here are the two directions of associativity, both witnessed to have validity greater than $1$:
	\begin{equation}
		\begin{prooftree}
					\hypo{\red{1}}
				\infer1[\AxRule]{A \entails A}
				\hypo{\red{\tensor}}
						\hypo{\red{1}}
					\infer1[\AxRule]{C \entails C}
					\hypo{\red{\tensor}}
						\hypo{\red{1}}
					\infer1[\AxRule]{B \entails B}
				\infer3[\TensorRightRule]{B, C \entails B \tensor C}
			\infer3[\TensorRightRule]{A, B, C \entails A \tensor (B \tensor C)}
			\infer1[\TensorLeftRule]{A \tensor B, C \entails A \tensor (B \tensor C)}
			\infer1[\TensorLeftRule]{(A \tensor B) \tensor C \entails A \tensor (B \tensor C)}
		\end{prooftree}
	\end{equation}
	\begin{equation}
		\begin{prooftree}
						\hypo{\red{1}}
					\infer1[\AxRule]{A \entails A}
					\hypo{\red{\tensor}}
						\hypo{\red{1}}
					\infer1[\AxRule]{B \entails B}
				\infer3[\TensorRightRule]{A, B \entails A \tensor B}
				\hypo{\red{\tensor}}
					\hypo{\red{1}}
				\infer1[\AxRule]{C \entails C}
			\infer3[\TensorRightRule]{A, B, C \entails (A \tensor B) \tensor C}
			\infer1[\TensorLeftRule]{A, B \tensor C \entails (A \tensor B) \tensor C}
			\infer1[\TensorLeftRule]{A \tensor (B \tensor C) \entails (A \tensor B) \tensor C}
		\end{prooftree}
	\end{equation}
	Commutativity is witnessed in both directions by the same proof:
	\begin{equation}
		\begin{prooftree}
					\hypo{\red{1}}
				\infer1[\AxRule]{A \entails A}
				\hypo{\red{\tensor}}
					\hypo{\red{1}}
				\infer1[\AxRule]{B \entails B}
			\infer3[\TensorRightRule]{B, A &\entails B \tensor A}
			\infer1{A, B &\entails B \tensor A}
			\infer1[\TensorLeftRule]{A \tensor B &\entails B \tensor A}
		\end{prooftree}
	\end{equation}
	where we freely reordered the antecedent by means of our convention that cedents are not considered to have an order (see the \emph{incipit} of \cref{sec:defs}).

	\item[\eqref{eq:softale-involutivity}] holds definitionally, recall \eqref{eq:qll-negation},
	\item[]
	\item[\eqref{eq:softale-duality}] $\validity{A \entails B} = \validity{B^\dual \entails A^\dual}$ by \DualLeftRule and \DualRightRule,
	\item[\eqref{eq:softale-star-aut}] $\validity{A \tensor B \entails C^\dual} = \validity{A \entails (B \tensor C)^\dual}$ is witnessed by the proofs
	\begin{equation}
		\begin{prooftree}
			\hypo{A, B \entails C^\dual}
			\infer1[\DualLeftRule]{A, B, C \entails}
			\infer1[\TensorLeftRule]{A, B \tensor C \entails}
			\infer1[\DualRightRule]{A \entails (B \tensor C)^\dual}
		\end{prooftree}
		\hspace*{10ex}
		\begin{prooftree}
				\hypo{\red{1}}
				\infer1[\AxRule]{B \entails B}
				\hypo{\red{\tensor}}
				\hypo{\red{1}}
				\infer1[\AxRule]{C \entails C}
			\infer3[\TensorRightRule]{B, C \entails B \tensor C}
			\hypo{\red{\tensor}}
				\hypo{A \entails (B \tensor C)^\dual}
			\infer1[\DualLeftRule]{A, B \tensor C \entails}
			\infer3[\CutRule]{A, B, C \entails}
			\infer1[\DualRightRule]{A, B \entails C^\dual}
		\end{prooftree}
	\end{equation}

	\item[\eqref{eq:softale-soft-bottom}] $\validity{\bot \entails A} = \infty$ is directly witnessed by \BotLeftRule.
	\item[\eqref{eq:softale-p-joins-minimality} and \eqref{eq:softale-p-joins-upper-bound}] are proven by \pOrRightRule and \pOrLeftRule, respectively.
\end{description}

\section{Proof of Grounded Completeness (\cref{th:grounded-completeness})}\label[appendix]{app:grounded-completeness-proof}

In the following, we fix the valuation $v:\Atoms \to \PosReals$ and work in the grounded theory $\Theory_v$, so that $\entails$ will mean entailment in that theory.

We start by noting that \MixStarRule does not interfere with cut-elimination: inspecting the proof at \cref{app:cut-elim} for the cases involving \MixRule, one can note that \cref{eq:duality} suffices to make the analogous proof transformations go through for \MixStarRule.

The next tool is the following observation:

\begin{lemma}\label{lemma:dualism}
	For all formulae $A,B$ in the language of $\Theory_v$,
	\begin{equation}
		\validity{A \entails B}^* \leq \validity{B \entails A}
	\end{equation}
\end{lemma}
\begin{proof}
	We start by observing that\footnotemark~$\validity{A \entails B} = \validity{\ \entails A^\dual, B} = \validity{\ \entails A^\dual \parr B}$, thus without loss of generality we can prove $\validity{A \entails\ } \geq \validity{\ \entails A}^*$.
	\footnotetext{
		This holds because \ParRightRule is invertible. Its inverse is obtained by cutting $\ \entails A \parr B$ with $A \parr B \entails A, B$.
	}
	Moreover, we also have $\validity{A \entails\ } = \validity{\ \entails A^\dual}$, and thus proving the inequality for $A$ proves it for $A^\dual$ too.
	Therefore, without loss of generality, we can assume $A$ is positive (i.e. made of propositional variables, $\bot$, $\One$, $\plor$, and $\tensor$).

	We prove the theorem by induction on the structure of $A$.

	For constants, we show that for every proof $\pi$ of $A \entails\ $ there is a proof $\overline{\pi}$ of $\ \entails A$ such that $\validity{\pi} \geq \validity{\overline{\pi}}^*$.
	It follows that $\validity{A \entails\ } \geq \validity{\pi} \geq \validity{\overline{\pi}}^* \geq \validity{\ \entails A}^*$.
	This is how we get $\overline{\pi}$:
	\begin{align}
		\begin{prooftree}
			\hypo{\red{v(a)}}
			\infer1[\AtomRightRule{a}]{a^\dual \entails\ }
		\end{prooftree}
		\qquad &\mapsto \qquad
		\begin{prooftree}
			\hypo{\red{v(a)^*}}
			\infer1[\AtomLeftRule{a}]{\ \entails a^\dual}
		\end{prooftree}
		\\[1.5ex]
		\begin{prooftree}
			\hypo{\red{v(a)^*}}
			\infer1[\AtomLeftRule{a}]{a \entails\ }
		\end{prooftree}
		\qquad &\mapsto \qquad
		\begin{prooftree}
			\hypo{\red{v(a)}}
			\infer1[\AtomRightRule{a}]{\ \entails a}
		\end{prooftree}
		\\[1.5ex]
		\begin{prooftree}
			\hypo{\red{\infty}}
			\infer1[\BotLeftRule]{\bot \entails\ }
		\end{prooftree}
		\qquad &\mapsto \qquad
		\begin{prooftree}
			\hypo{\red{0}}
			\infer1[\ExFalsoRule]{\ \entails \bot}
		\end{prooftree}
		\\[1.5ex]
		\begin{prooftree}
			\hypo{\red{1}}
			\infer1[\OneLeftRule]{\One \entails\ }
		\end{prooftree}
		\qquad &\mapsto \qquad
		\begin{prooftree}
			\hypo{\red{1}}
			\infer1[\OneRightRule]{\ \entails \One}
		\end{prooftree}
	\end{align}

	For $\tensor$, we first observe that $\validity{A \tensor B \entails} = \validity{A, B \entails}$ by invertibility of \TensorLeftRule.
	We then have
	\begin{equation}
		\validity{A \tensor B \entails} \geq \validity{A \entails\ } \cotensor \validity{B \entails} \geq \validity{\ \entails A}^* \cotensor \validity{\entails B}^*,
	\end{equation}
	with the first inequality witnessed by \MixStarRule and the second being the induction hypothesis.
	Now since $\validity{\ \entails A \tensor B} \geq \validity{\ \entails A} \tensor \validity{\entails B}$ as witnessed by \TensorRightRule, we have
	\begin{equation}
		\validity{\ \entails A}^* \cotensor \validity{\entails B}^* = \left(\validity{\ \entails A} \tensor \validity{\entails B}\right)^* \geq \validity{\ \entails A \tensor B}^*.
	\end{equation}
	Similarly, $\validity{A \lor B \entails} \geq \validity{A \entails\ } \coadd \validity{B \entails}$ by \pOrLeftRule, and thus
	\begin{equation}
		\validity{A \lor B \entails} \geq \validity{\ \entails A}^* \coadd \validity{\entails B}^*.
	\end{equation}
	Now since $\validity{\ \entails A \lor B} \geq \validity{\ \entails A} \add \validity{\entails B}$ because of \pOrRightRule, we have
	\begin{equation}
		\validity{\ \entails A}^* \coadd \validity{\entails B}^* = \left(\validity{\ \entails A} \add \validity{\entails B}\right)^* \geq \validity{\ \entails A \lor B}^*.
	\end{equation}
\end{proof}


\begin{proof}[Proof of \cref{th:grounded-completeness}]
	By induction on $\varphi$:
	\begin{enumerate}
		\item $\varphi = \One$, then $1 \leq 1$,
		\item $\varphi = \bot$, then $0$ is smaller than anything,
		\item $\varphi = \top$, then $\infty \leq \infty$,
		\item $\varphi = a$ propositional variable, then by definition of $\sem{a}_v$ and the subformula property,
		\begin{equation}
			v(a) = \sem{a}_v = \validity{\ \entails^{\Theory_v} a} = v(a),
		\end{equation}
		\item $\varphi = a^\dual$---same as above.
		\item $\varphi = \psi^\dual$, then $\sem{\psi^\dual}_v = \sem{\psi}_v^* = \validity{\ \entails^{\Theory_v} \psi}^*$ and by \cref{lemma:dualism} $\validity{\ \entails^{\Theory_v} \psi}^* \leq \validity{\ \entails^{\Theory_v} \psi^\dual}$.
		\item $\varphi = \alpha \ast \beta$ for $\ast \in \{\tensor, \parr, \pland, \plor\}$, we have
		\begin{equation}
			\sem{\alpha \ast \beta}_v = \sem{\alpha}_v \ast \sem{\beta}_v \leq \validity{\ \entails^{\Theory_v} \alpha} \ast \validity{\ \entails^{\Theory_v} \beta} \leq \validity{\ \entails^{\Theory_v} \alpha \ast \beta}
		\end{equation}
		Here the last inequality is induced by the right introduction rule of $\tensor, \plor, \pland$ for those connectives, and by \MixStarRule for $\parr$.
	\end{enumerate}
\end{proof}

\end{document}